\documentclass{article}
\usepackage{authblk}
\usepackage{graphicx} 
\usepackage{amsmath,amsthm,dsfont, amsfonts}
\usepackage{color}
\usepackage[dvipsnames]{xcolor}
\usepackage[hyphens]{url}
\usepackage[colorlinks=true,citecolor=Emerald,linkcolor=Bittersweet, urlcolor=Black,breaklinks=true]{hyperref} 
\usepackage{float}
\usepackage[margin=0.8in, a4paper]{geometry}
\usepackage[dvipsnames]{xcolor}
\usepackage[backgroundcolor=orange!20, textsize=tiny]{todonotes}
\usepackage{comment}
\usepackage{enumitem}
\usepackage{tikz}
\usepackage{multirow}
\usepackage{xfrac}
\usepackage{hhline,booktabs}
\usepackage{shadowtext}
\usepackage{array}
\usepackage{bold-extra}
\newcolumntype{P}[1]{>{\centering\arraybackslash}m{#1}}
\usepackage{makecell}

\usepackage[
    style=wikibibstyle, 
    citestyle=numeric-comp, sorting=none, backref=true, minnames=3, maxnames=12, giveninits=true, defernumbers=true
    ]{biblatex}
\DefineBibliographyStrings{english}{backrefpage={page}, backrefpages={pages}}

\makeatletter
\let\blx@rerun@biber\relax
\makeatother

\newcommand{\N}{\mathds{N}}
\newcommand{\R}{\mathds{R}}
\newcommand{\C}{\mathds{C}}
\newcommand{\E}{\mathbb{E}}

\newcommand{\bra}[1]{\langle #1 |}
\newcommand{\ket}[1]{| #1 \rangle}

\newcommand{\ketbra}[2]{| \hspace{0.5pt} #1 \rangle \!\langle #2 \hspace{0.5pt} |}
\newcommand{\Tr}{\operatorname{Tr}}

\newcommand{\opNorm}[1]{\left\lVert #1 \right\rVert}

\newcommand{\norm}[2]{\left\lVert #1 \right \rVert_{#2}}

\newcommand{\BQP}{\textsc{BQP}}
\newcommand{\BQPhard}{\textsc{BQP-hard}}
\newcommand{\BQPcomplete}{\textsc{BQP-complete}}
\newcommand{\prBQP}{\textsc{PromiseBQP}}

\newcommand{\BPP}{\textsc{BPP}}

\newcommand{\accessModel}{\textsc{AModel}}
\newcommand{\sparseAccess}{\textsc{SparseAccess}}
\newcommand{\PauliAccess}{\textsc{PauliAccess}}
\newcommand{\PauliSparse}{\textsc{PauliAccess}}

\newcommand{\problem}[4]{\textbf{Problem}: #1 \newline
\textbf{Input}: #2 \newline \textbf{Output}: #4
\vspace{0.3cm}}
\newcommand{\problemDiagonalEntryEstimation}{\textsc{DiagonalEntryEstimationSparseAccess}}
\newcommand{\problemMatrixPower}[2]{\textsc{Monomial}$_{#1}^{#2}$}
\newcommand{\problemMatrixPowerLM}[2]{\textsc{LM-Monomial}$_{#1}^{#2}$}
\newcommand{\problemMatrixInversion}[2]{\textsc{Inverse}$_{#1}^{#2}$}

\newcommand{\problemMatrixInversionLM}[2]{\textsc{LM-Inverse}$_{#1}^{#2}$}

\newcommand{\problemMatrixChebyshevPolynomial}[2]{\textsc{Chebyshev}$_{#1}^{#2}$}
\newcommand{\problemMatrixChebyshevPolynomialLM}[2]{\textsc{LM-Chebyshev}$_{#1}^{#2}$}
\newcommand{\problemHamiltonianSimulation}[2]{\textsc{TimeEvolution}$_{#1}^{#2}$}
\newcommand{\problemHamiltonianSimulationLM}[2]{\textsc{LM-TimeEvolution}$_{#1}^{#2}$}

\newcommand{\problemImaginaryTimeEvolutionSparseAccess}[1]{\textsc{ImaginaryTimeEvolutionSparseAccess}$_{#1}$}
\newcommand{\problemBQPCircuitSimulation}{\textsc{BQPCircuitSimulation}}

\newcommand{\polylog}[1]{\operatorname{poly} \log\!\left(#1 \right)}
\newcommand{\poly}[1]{\operatorname{poly}\!\left(#1 \right)}
\newcommand{\until}[1]{[#1]}
\newcommand{\ceil}[1]{\lceil #1 \rceil}

\newcommand{\OC}{\mathcal{O}}

\newtheorem{theorem}{Theorem}
\theoremstyle{plain} 
\newtheorem{proposition}[theorem]{Proposition}
\newtheorem{lemma}[theorem]{Lemma}

\newtheorem{observation}[theorem]{Observation}  
\newtheorem{remark}[theorem]{Remark}

\newtheorem{summary}[theorem]{Summary}
\newtheorem{definition}[theorem]{Definition}
\newtheorem{problemDef}{Problem}

\newcommand{\identity}{\mathds{1}}
\newcommand{\pauliX}{X}
\newcommand{\pauliY}{Y}
\newcommand{\pauliZ}{Z}
\newcommand{\aConditionNumber}{\kappa}

\newcommand{\zero}[1]{\ket{0}^{\otimes #1}}
\newcommand{\brazero}[1]{\bra{0}^{\otimes #1}}
\renewcommand{\vec}[1]{\boldsymbol{#1}}

\newcommand{\ignore}[1]{}

\definecolor{darkgreen}{rgb}{0.0, 0.5, 0.0}
\definecolor{lightblue}{RGB}{108, 172, 228}

\newcommand{\santi}[1]{\todo[inline,caption={},color=Orange!20, size=\footnotesize]{{\bf Santi:} #1}} 

\newcommand{\All}{\textsc{ALL}}
\newcommand{\matrixPower}{\textsc{MM}}
\newcommand{\matrixInversion}{\textsc{MI}}
\newcommand{\chebyshevPolynomial}{\textsc{CP}}
\newcommand{\timeEvolution}{\textsc{TE}}
\newcommand{\sparse}{\textsc{Spar}}

\newcommand{\pauli}{\textsc{Pauli}}
\newcommand{\both}{\textsc{Both}}

\newcommand{\kappaMatrix}[1]{\kappa_{#1}}
\newcommand{\lambdaMatrix}[1]{\lambda_{#1}}
\newcommand{\inputVector}[1]{\boldsymbol{#1}}
\newcommand{\eigen}{\theta}

    \newcommand{\step}[1]{\ket{step_{#1}}}
    \newcommand{\brastep}[1]{\bra{step_{#1}}}
    \newcommand{\clockTransition}{\mathcal{T}}

\pdfpageattr{/Group <</S /Transparency /I true /CS /DeviceRGB>>}

\title{\textbf{Quantum computational complexity of matrix functions}}
\author[1]{Santiago Cifuentes}
\author[2]{Samson Wang}
\author[3]{Thais L. Silva}
\author[4,5]{Mario Berta}
\author[3]{Leandro Aolita}
\affil[1]{Instituto de Ciencias de la Computaci\'on, UBA-CONICET, Argentina}
\affil[2]{Institute for Quantum Information and Matter, Caltech, USA}
\affil[3]{Quantum Research Center, Technology Innovation Institute, Abu Dhabi, UAE}
\affil[4]{Institute for Quantum Information, RWTH Aachen University, Germany}
\affil[5]{Department of Computing, Imperial College London, UK}

\date{\vspace{-0.5cm}}

\begin{document}

\maketitle

\begin{abstract}
We investigate the dividing line between classical and quantum computational power in estimating properties of matrix functions. More precisely, we study the computational complexity of two primitive problems: given a function $f$ and a Hermitian matrix $A$, compute a matrix element of $f(A)$ or compute a local measurement on $f(A)\zero{n}$, with $\zero{n}$ an $n$-qubit reference state vector, in both cases up to additive approximation error.  We consider four functions\,---\,monomials, Chebyshev polynomials, the time evolution function, and the inverse function\,---\,and probe the complexity across a broad landscape covering different problem input regimes. Namely, we consider two types of matrix inputs (sparse and Pauli access), matrix properties (norm, sparsity), the approximation error, and function-specific parameters.

We identify \BQPcomplete{} forms of both problems for each function and then toggle the problem parameters to easier regimes to see where hardness remains, or where the problem becomes classically easy. As part of our results, we make concrete a hierarchy of hardness across the functions; in parameter regimes where we have classically efficient algorithms for monomials, all three other functions remain robustly \BQPhard{}, or hard under usual computational complexity assumptions. In identifying classically easy regimes, among others, we show that for any polynomial of degree $\poly{n}$ both problems can be efficiently classically simulated when $A$ has $\OC(\log n)$ non-zero coefficients in the Pauli basis. This contrasts with the fact that the problems are \BQPcomplete{} in the sparse access model even for constant row sparsity, whereas the stated Pauli access efficiently constructs sparse access with row sparsity $\OC(\log n)$. Our work provides a catalog of efficient quantum and classical algorithms for fundamental linear-algebra tasks.
\end{abstract}

{ \hypersetup{hidelinks} \tableofcontents } 

\renewcommand{\arraystretch}{1.6}

\section{Introduction}

\subsection{Setting and motivation}

In which problems
from matrix algebra can we expect strong 
(i.e.~super-polynomial) quantum speedups?
Recently, powerful abstract frameworks for approximating matrix functions via matrix polynomials have been established \cite{low2016HamSimQubitization, gilyen2018QSingValTransf}, with polylogarithmic runtime in the dimension of the matrix for a broad class of instances. This provides a unified perspective on synthesizing many quantum algorithms, as matrix polynomials can now be thought of as building blocks to construct other interesting matrix functions on a quantum computer  \cite{gilyen2018QSingValTransf, martyn2021GrandUnificationQAlgs}. However, not all efficient quantum algorithms for matrix polynomials lead to a super-polynomial quantum speedup; indeed, some have classically efficient counterparts. In this work, we aim to characterize this dividing line in terms of computational complexity. Apart from matrix polynomials, we also
discuss the hardness of two more concrete problems: matrix inversion and time evolution, which {themselves} can also be considered building blocks for synthesizing other matrix functions \cite{haah2018ProdDecPerFuncQSignPRoc, lin2022HeisenbergLimited, silva2022fourier, wang2023QAlgGSEE, wang2023fasterGSEE, an2023LinCombHamSimulation, an2023QAlgNonUnitary,low2024quantumeigenvalueprocessing}. 

We study two primitive problems for matrix algebra. First, given a Hermitian matrix $A$ and function $f$
, we ask for one matrix element $[f(A)]_{ij}$. We refer to this as the \textit{matrix element problem}, which can be considered as an elemental matrix algebra task. Second, we consider a task, which we refer to as the \textit{local measurement problem}, that at first sight may appear more native to quantum approaches than the previous one: Performing 
a local measurement on a state 
acted on by $f(A)$ without normalization. That is, we ask for the value $\brazero{n}f(A)^\dagger\left(\ketbra{0}{0}\otimes \identity_{N/2}\right)f(A)\zero{n}$, where 
$\ketbra{0}{0}$ is a projector 
acting on the first qubit and $\identity_{N/2}$ is the 
identity matrix on the remaining $n-1$ qubits, for a total system dimension $N=2^n$. Immediate questions arise: 
\textit{how 
do the complexities of these two problems relate to each other?} And \textit{how do they depend on the input models assumed}? For both these problems, we derive hardness results for several different settings and problem-parameter regimes (see Fig.~\ref{fig:ven_diagram}).

The complexity of matrix function problems can depend on various factors, which we tune individually to investigate their effect on the hardness of the problems. The concrete properties that we study are:

\begin{itemize}
\setlength\itemsep{0.1em}
    \item The access model to the matrix
    \item The matrix normalization
    \item The matrix sparsity
    \item The desired precision 
    \item The type of function
\end{itemize}

\noindent Let us now briefly discuss and motivate each of these.

\textit{Access model.} A standard matrix access model that we investigate is the so-called {\it sparse access} model (Def.~\ref{def:sparseAccess}), where the non-zero matrix entries of $A$ in the computational basis are accessible via an efficiently computable function. This is commonly studied both constructively for quantum algorithms as well as for computational complexity analysis.  
Without an efficient structure for matrix entries, a generic quantum data structure of size $\Omega(N)$ would be required, even for row- and column-sparse matrices. That is,  
$\Omega(N)$ 
gates are required to instantiate the access model for a generic $N\times N$ row- or column-sparse matrix \cite{zhang2023circuitCompleixtyEncodingClassicalData}. 
Sparse access should then be seen as one way of imposing efficient access to large matrices. By ``efficient," we mean that there is a polylogarithmic-sized circuit in both depth and width able to provide access to the entries. 

Another reasonable type of access one might ask for, particularly in physically-motivated problems, is classical  
access to the non-zero coefficients of $A$  
in the Pauli basis (Def.~\ref{def:pauliAccess}). 
This access model can be considered as a concrete special case of sparse access, and it can only be efficient generically if there are polylogarithmically many {non-zero} coefficients. This condition is commonly encountered in chemistry, material science, or many-body physics applications. Alternatively, the non-zero Pauli coefficients can often have some underlying structure enabling efficient description, as has been recently considered in a line of work on randomized quantum algorithms \cite{campbell2019randomCompiler, wan2021RandPhaseEst, wang2023qubitefflinalg, nakaji2023qswift}. The main question we ask concerning these access models is: does providing a matrix in Pauli access change the complexity of these problems compared to sparse 
access in the computational basis? Answering this question should 
shed light on whether hardness results for matrix algebra problems in 
one model hold any relevancy for hardness in 
the other one. For instance, if a problem's complexity is unchanged by the access model, this suggests a form of complexity robustness with respect to the matrix representation.

\textit{Sparsity.} To a certain extent, the sparser the matrix, the easier it is to perform matrix algebra --- both for deterministic and randomized classical algorithms. Moreover, as discussed above, it is common to ask for efficient sparse access models to allow for scalable quantum algorithms.
Thus, it is pertinent to ask whether quantum algorithms become classically simulatable at some sparsity level (both in the sparse and Pauli access models).
However, in the absence of further conditions, we cannot expect efficient classical algorithms for general matrix functions even for $\OC(1)$-row-sparse matrices (in the computational basis) without further constraints --- for example, the problem of computing a local measurement for $f(A) = A^{-1}$ is well-known to be \BQPcomplete{} for $\OC(1)$-sparse matrices \cite{harrow2009QLinSysSolver}.
In our work we start with the case of $\OC(\polylog{N})$ many non-zero coefficients per row or $\OC(\polylog{N})$ non-zero coefficients in the Pauli basis (the maximal amounts that still allow efficient representation of matrices with generic coefficients), and investigate the change in complexity when we tune the sparsity down with additional constraints.

\begin{figure}[t!]
\centering
\begin{tikzpicture}
\pgftransformscale{.7}


\begin{footnotesize}

\draw[shading=radial,outer color=violet!25,middle color =violet!20,inner color=violet!15, rounded corners=1.6cm] (-11cm,-4.8cm) rectangle (11cm,5.6cm);
\node[rotate=0] at (3,-4.5) {\color{violet!90!black} \small{\bfseries{\textsc{BQP}}}};

\draw[shading=radial,outer color=orange!30,middle color =orange!20,inner color=orange!15,rounded corners=1.3cm] (-11cm,-4.4cm) rectangle (0.7cm, 5.3cm);
\node[rotate=0] at (-5.2,-4) {\color{orange!90!black} \small{\bfseries{{\textsc{BPP}}}}};

\draw[shading=radial,outer color=red!25,middle color =red!20,inner color=red!15, rounded corners=1.2cm] (-11,-3.6) rectangle (-5cm, 4.8cm);
\node[rotate=0] at (-8,-3.2) {\color{red!90!black} \small{\bfseries{\textsc{PTIME}}}};

    \begin{scope}[shift={(0.3,-0.5)}]
        \draw[shading=radial,outer color=blue!25,middle color =blue!20,inner color=blue!15, rounded corners=1cm] (5cm,-3.1cm) rectangle (10.5cm, 5.6cm
        );
        \node[rotate=0] at (7.75,-2.6) {\color{blue!70!black} \small{\bfseries{\BQPcomplete{}}}};

        \node[draw, text width=2cm,align=center] at (7.75,4.5) {\matrixPower$^{\both}$ if $\opNorm{A} \leq 1$};

        \node[draw, text width=2cm,align=center] at (7.75,3) {\chebyshevPolynomial$^{\both}$ if $\opNorm{A} \leq 1$};

        \node[draw, text width=3cm,align=center] at (7.75,1.5) {\matrixInversion$^{\both}$ if $\norm{A}{1} = 
        {\OC(1/\polylog{N})}$};

        \node[draw, text width=3cm,align=center] at (7.75,-0) {\timeEvolution$^{\both}$ if $\norm{A}{1} = 
        {\OC(1/\polylog{N})}$};


    \end{scope}

    \begin{scope}[shift={(-0.5,0)}]
        \node[draw, text width=2.3cm,align=center] at (-7.55,3.8) {\matrixPower$^{\sparse}$ if $\opNorm{A} \leq 1 - 
        \eta$};

        \node[draw, text width=2.5cm,align=center] at (-7.55,2.4) {\All$^{\sparse}$ if $\#_{\text{nz}}=\polylog{N}$};

        \node[draw, text width=2.1cm,align=center] at (-7.55,1) {\All$^{\pauli}$ if $L=\log\log{N}$};

        \node[draw, text width=2cm,align=center] at (-7.55,-0.4) {\matrixInversion$^{\sparse}$ if $\kappaMatrix{A} = \OC(1)$};

        \node[draw, text width=2.5cm,align=center] at (-7.55,-2) {\chebyshevPolynomial$^{\sparse}$ if\\ $s = \OC(1)$ and $m=\OC(\log\log N))$};
    \end{scope}

    \begin{scope}[shift={(-0.8,0)}]
        \node[draw, text width=2.7cm,align=center] at (-1.5,4.4) {\matrixPower$^{\sparse}$ if $\norm{A}{1} \leq 1$};

        \node[draw, text width=2.8cm,align=center] at (-1.5,3.25) {\matrixPower$^{\pauli}$ if $\opNorm{A} \leq 1 - 
        \eta\,,\, \text{or}\; \lambdaMatrix{A} \leq 1$};

        \node[draw, text width=2cm,align=center] at (-1.5,1.85) {\matrixInversion$^{\pauli}$ if $\kappaMatrix{A} = \OC(1)$};

        \node[draw, text width=3.1cm,align=center] at (-1.5,0.45) {\timeEvolution$^{\sparse}$ if\\ $\norm{A}{1} t = \OC(\log \log N)$};

        \node[draw, text width=3.1cm,align=center] at (-1.5,-1) {\timeEvolution$^{\pauli}$ if\\ $\lambdaMatrix{A} t = \OC(\log \log N)$};

        \node[draw, text width=2.8cm,align=center] at (-1.5,-2.7) {\chebyshevPolynomial$^{\pauli}$ if\\ $\lambdaMatrix{A} = \OC(1)$ and $m=\OC(\log\log N))$};

    \end{scope}

\node[draw, text width=2.3cm,align=center] at (3,2) {\chebyshevPolynomial$^{\sparse}$ if 
$\norm{A}{1} = 
{\OC(1/\polylog{N})}$};

\node[draw, text width=2.3cm,align=center] at (3,-1) {\chebyshevPolynomial$^{\pauli}$ if 
$\lambdaMatrix{A} = 
{\OC(1/\polylog{N})}$};

\end{footnotesize}

\end{tikzpicture}

\caption{{\bf Diagram indicating the complexity of the studied problems}. We use the acronyms $\matrixPower$ (matrix monomial), $\chebyshevPolynomial$ (Chebyshev polynomials), $\matrixInversion$ (matrix inversion), $\timeEvolution$ (time evolution), and $\All$ (all the previous functions). The superscript indicates the access model: \sparse{} (Sparse), \pauli{} (Pauli), or \both{} (the problem belongs to the complexity class for both access models). We denote $\opNorm{A}$ as the operator (or spectral) norm of $A$, $\norm{A}{1}$ as its induced 1-norm, $\lambdaMatrix{A}$ the vector $\ell_1$-norm of its Pauli coefficients, $L$ as the number of Pauli terms, $\#_{\text{nz}}$ its number of non-zero elements in the computational basis, $\kappaMatrix{A}$ its condition number, $t$ as the evolution time, and $\eta>0$ is an $\OC(1)$ number.
Although the \textit{local measurement problem} appears, at first sight, a more natural task for quantum algorithms than the \textit{matrix element problem}, interestingly, we find an almost-complete match between the two problems for almost all settings studied (see Table \ref{tab:summary}). 
The only potentially discrepant cases are $\chebyshevPolynomial$ for inverse polynomially small matrix norms (purple region) and $\timeEvolution$ for constant time (blue region), for which our hardness proof works only for the matrix element problem. All other results sketched in the figure hold for both problems. When not indicated, it is assumed that $\|A\|\leq 1$ and all other problem parameters (sparsity, inverse precision, problem-specific parameters) scale polynomially in the input size (which is polylogarithmic in the matrix dimension).
Hence, for instance, $\matrixPower^{\sparse}$ indicates both the matrix element and local measurement problems for $f_m(A)=A^m$, for $m$ polynomial in input size , where $A$ is given through the sparse access model.
} 
\label{fig:ven_diagram}
\end{figure}

\textit{Matrix normalization.} For many problems, the complexity of the classical algorithms solving them depends on norms different from the operator norm \cite{montanaro2024quantumclassicalquery,tang2022dequantizing}, and thus considering different normalization conditions for 
$A$ can affect such dequantization results. Here, we consider 
normalization conditions based mostly on three different matrix norms: the operator norm $\|A\|$; the larger induced $1$-norm $\|A\|_1$ (for 
the sparse access model) and the vector $\ell_1$-norm of the Pauli coefficients (for 
the Pauli access model). We consider the classes of problems where each of these norms is upper-bounded by $1$ and also investigate how the complexity changes when stronger bounds are put on these norms.

\textit{Error parameter.} 
For each function of interest, we consider two regimes of additive precision for the desired estimations: $1/\varepsilon = \OC(
\polylog{N})$ and $1/\varepsilon = \OC(1)$. Recent work \cite{gharibian2022dequantizingSVT} shows that 
changing from the former to the latter can make certain $\BQPcomplete{}$ problems turn classically efficiently solvable, under access assumptions.
Moreover, understanding the impact of the precision parameter in a problem's complexity is a central question regarding the Quantum PCP Conjecture \cite{aharonov2013guest}.

\textit{Type of function.} We study two classes of polynomials. First, we study monomials. This is arguably 
the simplest instance of a polynomial, 
and thus potentially the most amenable one to classical approaches. Second, we study Chebyshev polynomials. 
This is a powerful class of polynomials in numerical approximation theory \cite{sachdeva2014FasterAlgsViaApxTheory} and has been widely studied in the context of quantum algorithms \cite{childs2015QLinSysExpPrec, low2016HamSimQubitization, gilyen2018QSingValTransf,tosta2023randomizedSemiQuantum, apers2024quantumWalksWaveEqn}.
Finally, we also consider time evolution, also known as 
Hamiltonian simulation (i.e. $f_t(A) = e^{-iAt}$), and the inverse function.
Moreover, for each function family, we characterize the complexity of sub-classes given by restricted regimes of relevant parameters (the polynomial degree, evolution time, or matrix condition number). 

The rest of our manuscript is structured as follows: we discuss related work in Section \ref{sec:related_work}, we present a condensed form of our results in Section \ref{sec:results} along with a summary in Table \ref{tab:summary} and Figure \ref{fig:ven_diagram}, and the remainder of the manuscript is dedicated to elucidating our results in full detail with proofs.

\begin{table}[t!]
    \centering
    \resizebox{\columnwidth}{!}{
    \begin{tabular}{|c|c|c|c|c|c|}
        \hline
         \multirow{2}{*}{$f(A)$} & \multirow{2}{*}{Access} & \multirow{2}{*}{$\opNorm{A} \leq c
         $} & \multirow{2}{*}{$\norm{A}{1}\, \text{or}\, \lambdaMatrix{A}\leq k$} & \multicolumn{2}{c|}{Additional classically efficient cases} \\ 
         \cline{5-6}
         & & & & Super sparse matrices & Problem-specific cases\\
         \hline
         \hline
         \multirow{3}{*}{$A^m$} & Sparse & \multirow{3}{*} {\shortstack{\BQPcomplete{}  \\ for $c=1$ \\(\ref{enum:monomials_hard} \textit{\&} \cite{janzing2007simpleBQP})}} & \multirow{3}{*}{\shortstack{\BPP{} for $k=1$ \\(\hyperref[enum:monomials_classic_sparse]{\bf M.2.3} \textit{\&} \cite{apers2022SimpleBetti}, \\ \hyperref[enum:monomials_classic_pauli]{\bf M.3.2} \textit{\&} \cite{wang2023qubitefflinalg})}} & \setlength\extrarowheight{-3pt}\begin{tabular}{@{}c@{}}  $\#_{\text{nz}}= \OC(\polylog{N})$ \\ (\hyperref[enum:monomials_super_sparse]{\bf M.4.1}) \end{tabular} & \begin{tabular}{@{}c@{}} \vphantom{!}\\ \vphantom{!} \end{tabular} \setlength\extrarowheight{-1pt} \begin{tabular}{@{}c@{}} $s = 1$,\;\ or $s=\OC(1)$ and \vspace{-4pt}\\ $m = \OC(\log \log N)$ (\hyperref[enum:monomials_classic_sparse]{\bf M.2.1} \textit{\&} \cite{gharibian2022dequantizingSVT})\,; \\ $\opNorm{A} \leq 1-\eta$, $s = \OC(1)$ (\hyperref[enum:monomials_classic_sparse]{\bf M.2.2}) \end{tabular} \begin{tabular}{@{}c@{}} \vphantom{!}\\ \vphantom{!} \end{tabular}\\ \cline{2-2} \cline{5-6}
         & Pauli &  & & \setlength\extrarowheight{-3pt}\begin{tabular}{@{}c@{}} $L = \OC(\log \log N)$ \\ (\hyperref[enum:monomials_super_sparse]{\bf M.4.2}) \end{tabular} & \begin{tabular}{p{0cm}} \vphantom{!}\\ \vphantom{!} \end{tabular} \setlength\extrarowheight{-3pt}\begin{tabular}{@{}c@{}}$\opNorm{A} \leq 1-\eta$, $\lambdaMatrix{A} = \OC(1)$  (\hyperref[enum:monomials_classic_pauli]{\bf M.3.2})\end{tabular} \begin{tabular}{p{0cm}} \vphantom{!}\\ \vphantom{!} \end{tabular}  \\ 
         \hline
         \multirow{3}{*}{$T_m(A)$} & Sparse & \multirow{4}{*}{\shortstack{\BQPcomplete{}$^{\dagger}$ \\ for $c=1$ \\ (\ref{enum:cheby_BQP_hard})}} & \multirow{4}{*}{\shortstack{Classically hard$^{\star}$ \\ for $k=\OC(1/\polylog{N})$
         \\ if $\BPP{} \neq \BQP{}$ \\ (\ref{enum:cheby_hard_unless}, entry \\ estimation only)}} &   \setlength\extrarowheight{-3pt}\begin{tabular}{@{}c@{}}
         $\#_{\text{nz}} = \OC(\polylog{N})$\\ (\hyperref[enum:cheby_easy_log_degree]{\bf C.3.2}) 
         \end{tabular} & \begin{tabular}{p{0cm}} \vphantom{!}\\ \vphantom{!} \end{tabular} \setlength\extrarowheight{-3pt}\begin{tabular}{@{}c@{}} $s\ \text{or}\ \|A\|_1 = \OC(1)\ \text{and}$ \\  $m = \OC(\log \log N)$ (\hyperref[enum:cheby_easy_log_degree]{\bf C.3.1}) \end{tabular} \begin{tabular}{p{0cm}} \vphantom{!}\\ \vphantom{!} \end{tabular}  \\ 
         \cline{2-2} \cline{5-6} 
         & Pauli & & &  \setlength\extrarowheight{-3pt}\begin{tabular}{@{}c@{}}  $L = \OC(\log \log N)$ \\ (\hyperref[enum:cheby_easy_log_degree]{\bf C.3.2})  \end{tabular}  &  \setlength\extrarowheight{0pt}\begin{tabular}{@{}c@{}} $m = \OC(\log \log N)$ and \vspace{-6pt} \\  $\lambdaMatrix{A} = \OC(1)$ (\hyperref[enum:cheby_easy_log_degree]{\bf C.3.1})\,; \vspace{-1pt }\\ $\lambdaMatrix{A} = \OC(1/({m^2}\log^{1.5}(N)))$ and \vspace{-6pt}\\ $1/\varepsilon=\OC(1)$ (\hyperref[enum:cheby_easy_log_degree]{\bf C.3.3}) 
         \end{tabular}   
         \\ 
         \hline
         \multirow{2}{*}{$A^{-1}$} & Sparse & \multicolumn{2}{c|}{\multirow{3}{*}{\shortstack{\BQPcomplete{}$^\dagger$ \\ for $c,k=\OC(1/\polylog{N})$ \\ (\hyperref[enum:inversion_BQP_hard]{\bf I.1.1-I.1.2} \textit{\&} \cite{harrow2009QLinSysSolver})}}} & \begin{tabular}{@{}c@{}}
          $\#_{\text{nz}}= \OC(\polylog{N})$ \vspace{-6pt}\\(\ref{enum:gen-func-sparse}) 
         \end{tabular} &   \setlength\extrarowheight{0pt}\begin{tabular}{@{}c@{}} $s$ or $\|A\|_1=\OC(1)$ and \vspace{-6pt} \\ $\kappaMatrix{A} = \OC(1)$ (\ref{enum:inversion_easy}) \end{tabular} \\ \cline{2-2} \cline{5-6}  
         & Pauli & \multicolumn{2}{l|}{} & \begin{tabular}{@{}c@{}}
         $L = \OC(\log\log N)$ \vspace{-6pt}\\(\ref{enum:gen-func-Pauli-sparse}) 
         \end{tabular} &   \begin{tabular}{@{}c@{}} $\lambda_A = \OC(1)$ and $\kappaMatrix{A} = \OC(1)$ (\ref{enum:inversion_easy})\,; \vspace{-1pt}\\ $\lambda_A = \widetilde{\OC}( 1/(\kappa_A^2 \log^{1.5}(N)))$ and \vspace{-6pt}\\  $1/\varepsilon=\OC(1)$ (\ref{enum:gen-func-reduced-norm}) \end{tabular}  \\ 
         \hline
         \multirow{2}{*}{$e^{-iAt}$} & Sparse & \multicolumn{2}{c|}{\multirow{2}{*}{\shortstack{\BQPcomplete{}$^\dagger$\, for $c,k=\OC(1/\polylog{N})$ \\   (\hyperref[enum:timeev_BQP_hard]{\bf T.1.1-T.1.3} \textit{\&} \cite{Feynman1985, Nagaj_2010}) } }}  
         & \begin{tabular}{@{}c@{}} \vphantom{!}\\ same as $A^{-1}$ \end{tabular} & \begin{tabular}{@{}c@{}} $\opNorm{A} t = \OC(\log \log N)$, $s= \OC(1)$ (\ref{enum:timeev_easy})\,; \\  $\|A\|_1 t = \OC(\log \log N)$ (\ref{enum:timeev_easy})\end{tabular} \\ \cline{2-2} \cline{6-6}  & Pauli & \multicolumn{2}{c|}{}  &  & \begin{tabular}{@{}c@{}} $\lambdaMatrix{A} t = \OC(\log \log N)$ (\ref{enum:timeev_easy}) \end{tabular} 
         \\ \hline
    \end{tabular}
    } 
    \caption{{\bf Results for the matrix element and local measurement problems.}
    Unless explicitly stated we assume that $\|A\|\leq 1$ and $1/\varepsilon, s, L, m,t, \kappaMatrix{A}, \lambdaMatrix{A} = \OC(\polylog{N})$ where applicable, which can be considered problem inputs. All results shown hold for both problems, except the one with a superscript $\star$, for which our proof only holds for the matrix element problem. All the hardness results 
    with the 
    superscript $\dagger$ hold even for the easier problem where precision $\varepsilon$ is fixed. 
    We 
    prove hardness and classical simulability results for four classes of matrix functions: monomials, Chebyshev polynomials $T_m(A)$, the complex exponential (time evolution), and the inverse function.
    In addition, certain results for classical simulability are generalized for general polynomials (see Summary \ref{sum:gen_poly}).
    The symbol $\#_{\text{nz}}$ stands for ``number of non-zero elements in the computational basis". The rest of the notation in the table is the same as in Fig.~\ref{fig:ven_diagram}.
    We remark that the stronger the condition on the norm (smaller norm), generally the easier the problem.}
    \label{tab:summary}
\end{table}

\subsection{Related work}
\label{sec:related_work}
In \cite{montanaro2024quantumclassicalquery}, Montanaro and Shao study the query complexity 
(i.e., number of queries to an oracle for $A$) for the matrix element problem for both classical and quantum algorithms for matrices of bounded operator norm $\|A\|\leq 1$. 
For classical algorithms, they give query complexity lower bounds that grow exponentially in the degree of the target polynomial, assuming the input matrix has bounded operator norm and the estimation error satisfies $1/\varepsilon = \OC(\polylog{N})$.\footnote{This also extends to functions approximated by a polynomial of said degree.} This contrasts with (essentially matching) upper and lower bounds for quantum algorithms, which are linear in the polynomial degree. Our work can be seen as a complementary investigation where, instead of query complexity, we tackle computational complexity. 
A central difference between these two notions is that, with
computational complexity, the access model must also be computationally efficient. In contrast, with query complexity, computational hardness can (in theory) be hidden inside the oracle to $A$.
Another distinction from \cite{montanaro2024quantumclassicalquery} is that there the central parameter probed is the degree of the target polynomial, whereas here we toggle various additional problem parameters including properties of the matrix.

Another related line is the significant body of literature on so-called quantum-inspired or dequantization algorithms \cite{tang2022dequantizing, tang2018QuantumInspiredRecommSys, tang2018QInspiredClassAlgPCA, gilyen2020ImprovedQInspiredAlgorithmForRegression, chia2019SampdSubLinLowRankFramework, Shao2022FQILSS, gharibian2022dequantizingSVT}. Here, classical algorithms are emboldened with a particular kind of 
access to vectors and matrices which mimics quantum query access, usually called ``sample and query access models". Whilst such work is hugely informative about the feasibility of generic superpolynomial speedup for matrix problems given oracular query access, we stress again, our setting differs in that we work with access models that are efficiently instantiable.


In turn, we highlight a recent work by Gharibian and Le Gall \cite{gharibian2022dequantizingSVT}, which provides a quantum-inspired algorithm for the Guided Local Hamiltonian problem\footnote{This problem consists of  finding the ground state energy of a local Hamiltonian given a ``guiding"
vector with $\Omega\left(1/\poly{N}\right)$ overlap with the corresponding eigenspace.} for constant precision. This is relevant to our discussions for two reasons.  First, the authors establish a core classical subroutine 
for sparse matrix monomials which has exponential-in-degree runtime in general but is efficient for $1$-sparse matrices or for monomials of logarithmic degree (see {Lemma \ref{lemma:power-tractable_fixed}}). 
We complement these findings by discovering other restricted settings that result in a classically efficiently solvable problem (see Table \ref{tab:summary}).
Second, their result provides an example of a matrix problem that is classically efficient when the precision is 
fixed (in the different setting of sample and query access to the guiding vector). Efficient classical constant-precision algorithms for ground state problems have also been found in other contexts \cite{bansal2009classical}. 
Yet, to our knowledge, the realm of constant precision remains mostly unexplored for more general matrix algebra tasks.

Finally, let us summarize a few additional previous results on the hardness of matrix functions that directly fall within our setting. These 
are contextualized in Table \ref{tab:summary} among our results. In the sparse access model, Janzing and Wocjan showed that estimating matrix elements of matrix monomials is \BQPcomplete{} for inverse error and 
monomial power scaling polylogarithmically with $N$, 
for $A$ normalized by its operator norm \cite{janzing2007simpleBQP}. However, the same problem was shown to be classically easy by Apers et al.~when the matrix normalization is 
strengthened 
to the induced 1-norm, i.e. when the norm is decreased from $\|A\|\leq 1$ to $\norm{A}{1}\leq 1$ \cite{apers2022SimpleBetti, montanaro2024quantumclassicalquery}.
We observe that the proof from \cite{janzing2007simpleBQP} shows that relaxing the normalization condition from $\norm{A}{1}\leq 1$ to $\norm{A}{1}\leq 2$ brings back the BQP-completeness of the problem. 
This reflects a sharp transition in hardness and evidences, thereby, a form of tightness of the algorithm from 
\cite{apers2022SimpleBetti}, since relaxing its hypothesis to $\norm{A}{1} = \OC(1)$ breaks it without any possibility of fixing it (assuming $\BPP{} \neq \BQP{}$). 
An efficient classical algorithm analogous to that of Apers et al.
for the Pauli access models (both 
 for deterministic and random sampling access) was presented by Wang et al.~\cite{wang2023qubitefflinalg}, when $A$ is normalized by its $\ell_1$-norm of the Pauli coefficients. 
 
BQP-completeness of the normalized-state version of our local measurement problem was shown for $f(A)=A^{-1}$ by Harrow, Hassadim, and Lloyd (HHL) in their seminal paper \cite{harrow2009QLinSysSolver}, for the sparse access model, operator-norm matrix normalization, and constant precision $1/\varepsilon = \OC(1)$.
Employing their construction, we can prove BQP-completeness 
also for the non-normalized case under the same assumption of $1/\varepsilon = \OC(1)$.
In addition, we characterize change in complexity of matrix inversion over the different settings we probe, i.e., we also look at the matrix element estimation problem (using a similar construction to that for monomials from  \cite{janzing2007simpleBQP}), sparse access versus Pauli access models, and different norms and sparsity regimes.

The universality of time evolution, meaning the capability of encoding any quantum circuit in a local Hamiltonian time evolution, is known from the original work by Feynman \cite{Feynman1985}. In particular, the circuit associated with any 
BQP problem can be mapped into a Hamiltonian and probabilistically implemented. Building upon Feynman's construction, Nagaj showed that a polynomial-sized circuit can be implemented using time evolution for a $3$-local (sparse) Hamiltonian with $t=\poly{n}$ \cite{Nagaj_2010}. This result naturally translates to the BQP-completeness of the local measurement problem with $f_{t}(A)=e^{-iAt}$. Here, we include a BQP-completeness proof using a slightly different Hamiltonian and extend the result to the matrix entry problem. 


\section{Results}\label{sec:results}
\subsection{Basic definitions}

In order to provide a precise summary of our results, we first need to define a few basic concepts. We denote $\opNorm{A} = \max_{\vec{x} : \opNorm{\vec{x}}_2 = 1} \opNorm{A\vec{x}}_2$ as the operator (or spectral) norm of $A$ and $\norm{A}{1} =\max_{\vec{x} : \opNorm{\vec{x}}_1 = 1} \opNorm{A\,\vec{x}}_1 = \max_{1\leq j \leq N} \sum_{i=1}^N A_{i,j}$ its induced 1-norm. Finally, for $A=\sum_{\ell} a_{\ell} P_{\ell}$ decomposed in the Pauli basis, we denote $\lambdaMatrix{A}=\sum_{\ell} |a_{\ell}|$ which we call the Pauli norm of $A$. We will use the symbol $s$ to denote sparsity, $L$ to denote number of Pauli terms, and $\varepsilon$ to denote precision. We start by recalling the notion of function of a matrix:

\begin{definition}[Function of a matrix (eigenvalue transformation)]\label{def:function_of_matrix}

    Let $A \in \C^{N \times N}$ be a Hermitian matrix diagonalized as $A = S \Lambda S^{-1}$, for $\Lambda$ real and diagonal and $S$ unitary, and $f:\R \to \C$ some univariate function. Then $f(A) = S f(\Lambda) S^{-1}$,   where $f(\Lambda)$ is obtained by applying $f$ to each diagonal element of $\Lambda$ while leaving the off-diagonal ones untouched.

\end{definition}

The two central tasks we study are computing entries of $f(A)$ or overlaps between a local measurement operator and a state vector transformed under $f(A)$, for different functions $f$. We now provide their definitions.\footnote{
The precise problems to which our computational hardness results will directly apply are actually the promise problem version of these estimation tasks, i.e., deciding whether the target quantity is above or below a value range, instead of directly estimating it. Both problems are intimately connected: for example, using an algorithm that decides given $g$ whether $A^m_{j,j} \geq g$ it is possible to approximate the value $A^m_{j,j}$ by doing binary search on the value $g$ in the range $[-\opNorm{A}^m, \opNorm{A}^m]$. The formal promise problems considered for hardness analysis are defined in Sec.~\ref{sec:detailed-statements} for each function.}

\begin{problemDef}[Matrix element problem]\label{def:general_problem}
    Let $\{f_m\}_{m \in \N}$ be a family of functions. Given 
    access to a Hermitian matrix $A \in \C^{N \times N}$ with bounded norm (such as $\opNorm{A} \leq 1$ or $\norm{A}{1} \leq 1$), two indices $i,j \in \until{N}$, a precision $\varepsilon > 0$ and a natural number $m \in \N$, compute an $\varepsilon$-approximation of $\bra{i} f_m(A) \ket{j}$.
\end{problemDef}

\begin{problemDef}[Local measurement problem]\label{def:general_problem-lm}
    Let $\{f_m\}_{m \in \N}$ be a family of functions. Given 
    access to a 
    Hermitian matrix $A \in \C^{N \times N}$ with bounded norm (such as $\opNorm{A} \leq 1$ or $\norm{A}{1} \leq 1$), a precision $\varepsilon > 0$ and a natural number $m \in \N$, compute an $\varepsilon$-approximation of $\brazero{n}
    f_m(A)^{\dag} (\ketbra{0}{0} \otimes \identity_{N/2}) f_m(A) 
    \zero{n}$, where $\ketbra{0}{0}$ is single-qubit rank-1 projector and $\identity_{N/2}$ is the $N/2 \times N/2$ identity matrix.
\end{problemDef}

We note that we could consider more general matrices than Hermitian matrices, and the more general Singular Value Transformation \cite{gilyen2018QSingValTransf, martyn2021GrandUnificationQAlgs} rather than the eigenvalue transform. Nevertheless, we restrict to this simpler setting because through a portion of our results we are interested in proving 
BQP-hardness. Hence, such hardness results will necessarily extend to more general cases if we restrict to
more constrained formulations for the problems. In fact, our proofs of hardness will even hold for the further restricted class of real symmetric matrices.


We study whether the matrix representation affects the difficulty of the problems. As discussed above, the two models we will consider are \textit{sparse access} and \textit{Pauli access}, defined as follows. 

\begin{definition}[Sparse access]\label{def:sparseAccess}
    A matrix $A \in \C^{N \times N}$ is a $s$-sparse matrix if it has at most $s$ non-zero entries per row and column. In addition, if $s=\OC(\polylog{N})$, we refer to $A$ simply as a sparse matrix.
    
    We say that we have classical sparse access to $A$ if $(i)$ we have
    efficiently-computable functions $h_r,h_c:\until{N} \times \until{s} \to \until{N}$ such that $h_r(i,k)$ is the index of the $k$-th non-zero entry of the $i$-th row of $A$ and $h_c(l,j)$ is the index of the $l$-th non-zero entry of the $j$-th column, and $(ii)$ given any $i,j\in\until{N}$, we can efficiently compute the entry $A_{i,j}$. 
    
    Meanwhile, we say we have quantum sparse access to $A$ if we have the following oracles.
    \begin{alignat}{2}
        O_{\mathrm{row}}&: & \ket{i} \ket{k}&\to  \ket{i} \ket{h_r(i,k)} \nonumber\\
        O_{\mathrm{col}}&: & \ket{l} \ket{j}&\to  \ket{h_c(l,j)} \ket{j} \\ 
        O_{A}&:\quad & \ket{i} \ket{j}\zero{b}&\to  \ket{i} \ket{j} \ket{A_{i,j}} \nonumber
    \end{alignat}
    for $i,j \in \until{N}, k,l \in \until{s}$.
    
\end{definition}

We further note that efficient classical sparse access automatically implies efficient quantum sparse access. The quantum access can be granted via efficient boolean circuits for arithmetic operations \cite{reif1986logarithmic}, which can be mapped to reversible circuits \cite{bennett1973logical}, or via more modern quantum arithmetic circuits \cite{shpilka2010arithmetic, wang2024comprehensive}. The soundness of this model was recently highlighted in \cite{Zhang_2024}, where the sparse access oracles are explicitly constructed for physically relevant Hamiltonian matrices.

\begin{definition}[Pauli query access and Pauli-sparseness]\label{def:pauliAccess}
   Consider the decomposition of $A \in \C^{N \times N}$  as
\begin{align}\label{eq:pauli_representation}
        A = \sum_{\ell=1}^L a_\ell P_\ell
    \end{align}
    where each $P_\ell$ is a multi-qubit Pauli matrix (tensor product of single-qubit Paulis) and $a_\ell \in \C$. Then, Pauli query access consists of efficient classical access to the 
    coefficients $\{a_\ell\}_{\ell\in\until{L}}$ and the \textit{Pauli norm} $\lambda_A = \sum_{\ell=1}^L |a_\ell|$. Moreover, whenever $L, \lambdaMatrix{A} = \OC(\polylog{N})$ we say that the matrix is \textit{Pauli-sparse}.
\end{definition}

Importantly, we stress again that Pauli access is a special case of classical sparse access, which can be natural for many problems.

\subsection{Summary of results}
\label{subsec:summary}
We now summarize our contributions. In this section we give informal versions of our results, quoted alongside prior results in the literature which also fall into our setting. We also point the reader to Table \ref{tab:summary} which provides a further condensed visual summary. Full, formal statements along with proofs can be found in Sec.~\ref{sec:detailed-statements}.

As mentioned in Sec.~\ref{sec:related_work}, Problem~\ref{def:general_problem} has already been considered for monomials 
$f_m(x) = x^m$ in \cite{janzing2006BQPmixing}, where it was shown that\, if $A$ is given through sparse access and satisfies $\opNorm{A} \leq 1$, the problem 
is \BQPcomplete{}.
We show 
an analogous result  
for the Pauli Access model, and even when $A$ 
is assumed to be Pauli-sparse. 
We also show that these two results hold for the local measurement problem too. 
{In turn,} 
when the norm condition $\opNorm{A} \leq 1$ is 
strengthened to $\opNorm{A} \leq 1 -\eta$, for any fixed $\eta>0$,
both problems become classically easy 
for both access models given a sparsity assumption. The same has been known to be true in the absence of an additional sparsity assumption for the alternative norm assumptions $\norm{A}{1} \leq 1$ or $\lambdaMatrix{A} \leq 1$, which was shown in \cite{apers2022SimpleBetti} and \cite{wang2023qubitefflinalg} respectively.
Finally, we find the problem is classically easy to solve exactly if the matrix is made sparse enough in the Pauli basis. 

\begin{summary}[Results for matrix monomials]\label{sum:monomials}
    Instantiate Problems~\ref{def:general_problem} and~\ref{def:general_problem-lm} with $f_m(x) = x^m$. Unless explicitly stated, set $m$, $1/\varepsilon = \OC(\polylog{N})$, $\|A\|\leq 1$, and let $A$ be either sparse or Pauli-sparse (i.e., either $s$ or $L,\lambdaMatrix{A} = \OC(\polylog{N})$) depending on the access model. Then:

    \begin{enumerate}[label=\textbf{M.\arabic*}]
        \item \label{enum:monomials_hard} Problem \ref{def:general_problem} is \BQPcomplete{} when the input matrix $A$ is given through either the sparse access model (Thm.~\ref{teo:matrixPowerSparseAccessBQPcomplete} \cite{janzing2007simpleBQP}) or the Pauli access model (Prop.~\ref{prop:matrixPowerPauliAccessBQPcomplete}), as is Problem \ref{def:general_problem-lm} in both access models (Prop.~\ref{prop:matrixPowerSparseAccessBQPcomplete-lm}). For the sparse access model this holds even for a choice of $s = \OC(1)$. Moreover, in both cases, the result holds even if we add the condition that $\norm{A}{1} \leq 2$ and $A$ is 5-local.
        
        \item \label{enum:monomials_classic_sparse}Both problems can be solved efficiently classically in 
        time $\OC(\polylog{N})$ whenever 
        $A$ is given through the sparse access model and satisfies either 
        
        \begin{enumerate}
            \item[(1)] $s = 1$, or $s= \OC(1)$ and $m = \OC(\log \log N)$ (exact algorithm, Lemma \ref{lemma:power-tractable_fixed} \cite{gharibian2022dequantizingSVT}), or
            \item[(2)] $\opNorm{A} \leq 1-
            \eta$ with $
            \eta = \Omega(1)$, $s=\mathcal{O}(1)$  (Props.~\ref{teo:classically_easy_bounded_precision_and_norm} and~\ref{prop:matrix_power_bounded_norm_pauli_tractable}), or
            \item[(3)] $\norm{A}{1} \leq 1$ (Prop.~\ref{prop:one_norm_classically_easy} \cite{apers2022SimpleBetti}). 
        \end{enumerate}
        
        \item \label{enum:monomials_classic_pauli} Both problems can be solved efficiently classically in time $\OC(\polylog{N})$ 
        whenever $A$ is given through the Pauli access model and satisfies either:

        \begin{enumerate}
            \item[(1)] $\opNorm{A} \leq 1-\eta$, $\lambdaMatrix{A} = \OC(1)$, and $\eta =\Omega(1)$
            (Props.~\ref{teo:classically_easy_bounded_precision_and_norm} and~\ref{prop:matrix_power_bounded_norm_pauli_tractable}), or
            \item[(2)] $\lambdaMatrix{A} \leq 1$(Prop.~\ref{prop:matrix_power_pauli_normalized_classically_easy} \cite{wang2023qubitefflinalg})
        \end{enumerate}

        \item \label{enum:monomials_super_sparse} 
        Both problems can be solved efficiently classically in time $\OC(\polylog{N})$ if $A$ either
        
        \begin{enumerate}
            \item[(1)] Has at most $\OC(\polylog{N})$ non-zero entries in the computational basis 
            (Prop.~\ref{prop:sparse_power_super_sparse}), or

            \item[(2)] Is given through the Pauli access model and $L = \OC(\log \log N)$ (Thm.~\ref{teo:pauli_power_super_sparse}).
        \end{enumerate}
    \end{enumerate}
\end{summary}

The fact that a strengthening of the norm condition $\opNorm{A}\leq 1$ to $\opNorm{A} \leq 1 - \eta$ allows one to develop efficient classical algorithms suggests that the quantum advantage is related to handling the higher end of the spectrum, i.e. the eigenvectors with eigenvalues close to one. 
We note that the condition $\opNorm{A} \leq 1 - \eta$ was also studied in \cite{montanaro2024quantumclassicalquery} as a sufficient condition to construct a non-normalized quantum block encoding starting from sparse query access.\footnote{The standard way \cite[Lemma 48]{gilyen2018QSingValTransf} of preparing a block encoding of an $s$-sparse matrix $A$ with largest matrix entry magnitude $\leq 1$ given through the sparse access model returns a subnormalized block encoding of $A/s$. To the authors' knowledge, there is not currently a known method to generically obtain a non-normalized block encoding when starting from sparse access.} This indicates that for certain problems, specific use of a block encoding could shroud superpolynomial quantum advantage. 

We observe that the classical algorithm from \hyperref[enum:monomials_classic_sparse]{\bf M.2.2} cannot be extended to matrices with $\|A\|_1=\OC(1)$ since the hardness results holds even when $\norm{A}{1} \leq 2$ and $s=4$ (Thm.~\ref{teo:matrixPowerSparseAccessBQPcomplete} and Ref.~\cite{janzing2007simpleBQP}). This reflects a sharp transition in hardness with respect to the norm parameter.

A first look at \ref{enum:monomials_hard} indicates that for a certain variation of the problem, the \BQP-\textit{hardness} is robust with respect to the choice of access model, namely for the parameter settings of \cite{janzing2006BQPmixing}. However, result \ref{enum:monomials_super_sparse} shows a concrete difference in complexity between the Pauli and sparse access models in the regime when the sparsity is $o(\polylog{N})$. As discussed, in the sparse access model constant row sparsity is in general \BQPhard{}. In contrast, here we see that if $A$ has structure in the Pauli basis with $L=\OC(\log \log(N))$, implying super-constant sparsity $s=\OC(\log \log(N))$ (and total number of non-zero computational basis elements $\OC(LN)$), the problem is classically easy. This demonstrates that Pauli structure makes problems easier than other efficiently-computable instantiations of sparse accesses. The difference regarding the complexities between the access models for our classical algorithms (we demand $\polylog{N}$ non-zero entries for the sparse access model but $\OC(\log \log N)$ for the Pauli one) is due to the fact that products between canonical projectors $\ket{i}\bra{j}$ mostly vanish, while those between Pauli matrices do not.

Result \ref{enum:monomials_hard} is obtained by observing that the reduction from 
\cite{janzing2006BQPmixing} builds a matrix that is sparse also in the Pauli representation. It can also be seen that the matrix built in this reduction has 1-norm upper bounded by 2. We generalize these techniques to demonstrate a simple criterion for BQP-hardness for any function: any class of matrix functions for entry estimation parameterized by parameter $m$ is \BQPhard{} for inverse error $1/\varepsilon = \OC(1/k)$ if its odd component $f^{o}_m$  satisfies the condition
\begin{equation}\label{eq:jantzig_condition}
    \frac{1}{M}\bigg|f^{o}_m(1)+2\sum_{l=1}^{\frac{M-1}{2}}f^{o}_m\big(\cos({2\pi \ell}/{M})\big)\bigg| \geq k\,,
\end{equation}
for any $M = \OC(\polylog{N})$ (see proof of Theorem \ref{teo:matrixPowerSparseAccessBQPcomplete} and  Eq.~\eqref{eq:jantzig_condition2} for more detailed discussion). 

Meanwhile, for \hyperref[enum:monomials_classic_sparse]{\bf M.2.2}  and \hyperref[enum:monomials_classic_pauli]{\bf M.3.1}, we develop an efficient classical
algorithm under the condition $\opNorm{A} \leq 1-\eta$ by observing 
that 
$A^m$ tends to zero fast as $m$ increases. 
In turn, to prove results \hyperref[enum:monomials_classic_sparse]{\bf M.2.3} and \hyperref[enum:monomials_classic_pauli]{\bf M.3.2}, respectively based on the conditions $\norm{A}{1} \leq 1$ and $\lambdaMatrix{A} \leq 1$, we use classical algorithms that are particular cases of algorithms from \cite{montanaro2024quantumclassicalquery} and \cite{wang2023qubitefflinalg}, 
which rely on Monte Carlo sampling and Markov-chain Monte Carlo. 
Finally, the results in \ref{enum:monomials_super_sparse} are proven by observing that, under the given sparsity conditions, 
$A^m$ belongs to a low-dimensional subspace and  can thus be efficiently represented explicitly for any $m$.

\medskip

For Chebyshev polynomials we prove the following:
\begin{summary}[Results for matrix Chebyshev polynomials]\label{sum:chebyshev}
    Instantiate Problems~\ref{def:general_problem} and \ref{def:general_problem-lm} with $f_m(x) = T_m(x)$, with $T_m$ the Chebyshev polynomial of the first kind and degree $m\in \N$. Unless explicitly stated set $m,1/\varepsilon = \OC(\polylog{N}), \|A\|\leq 1$, and let $A$ be sparse or Pauli-sparse depending on the access model (that is, $s$ or $L,\lambdaMatrix{A} = \OC(\polylog{N})$. Then,

    \begin{enumerate}[label=\textbf{C.\arabic*}]
        \item \label{enum:cheby_BQP_hard}  The problems are \BQPcomplete{} whenever the matrix $A$ is given through either access model (Thm.~\ref{teo:chebychev-sparse-alternative-BQPcomplete} and Prop.~\ref{prop:chebyshev_lm_hard}). 
        For the sparse access model this holds even for a choice of $s = \OC(1)$. Moreover, in both cases, the result holds even if we add the conditions that $1 / \varepsilon = \Omega(1)$, $\norm{A}{1} \leq 2$ and $A$ is 5-local.

        \item \label{enum:cheby_hard_unless} If $\BPP{} \neq \BQP{}$ Problem \ref{def:general_problem} 
        cannot be solved classically in polynomial time
        when the input matrix $A$ is either
        
        \begin{enumerate}
            \item[(1)] Given through the sparse access model and satisfies $\norm{A}{1} =\OC(1/\polylog{N})$ (Prop.~\ref{teo:hardness_chebyshev_1_norm}). 
            \item[(2)] Given through the Pauli access model, is Pauli-sparse and satisfies $\lambdaMatrix{A} =\OC(1/\polylog{N})$ (Prop.~\ref{teo:hardness_chebyshev_1_norm}). 
        \end{enumerate}

        \item \label{enum:cheby_easy_log_degree} Both problems can be solved classically in polynomial time via
        \begin{enumerate}
            \item[(1)] An exact algorithm in the sparse access model whenever $s=\OC(1)$ and $m=\OC(\log \log N)$, a randomized algorithm in the sparse access model whenever $\norm{A}{1} = \OC(1)$ and $m = \OC(\log \log N)$ or a randomized algorithm for the Pauli access model if $\lambdaMatrix{A}=\OC(1)$ and $m = \OC(\log \log N)$  (Prop.~\ref{prop:cheby_sparse_low_degree}). 
            \item[(2)] An exact algorithm in the sparse access model when $A$ has at most $\polylog{N}$ non-zero entries in the computational basis (consequence of Prop.~\ref{prop:sparse_power_super_sparse}) or an exact algorithm in the Pauli access model whenever $L= \OC(\log \log N)$ (consequence of Thm.~\ref{teo:pauli_power_super_sparse}). 
            \item[(3)] A randomized algorithm in the Pauli access model when $\lambda_A = \OC(1/(m^2 \log^{1.5} (N)))$ and the required precision is constant (Thm.~\ref{teo:pauli-sketch}).
        \end{enumerate}
    \end{enumerate}

\end{summary}

The results from Summary~\ref{sum:chebyshev} show a sharp contrast between monomials and the Chebyshev polynomials. We first see indication of this as for the latter the problem remains \BQPhard{} even if we ask for a fixed precision $\varepsilon \leq \frac{1}{3}$. Furthermore, strengthening the normalization condition beyond $\opNorm{A} \leq 1$ does not allow for polynomial-time classical algorithms (under the common assumption that $\BQP{} \neq \BPP{}$), which is a concrete separation from monomials for which we showed an efficient algorithm. 

Result \ref{enum:cheby_BQP_hard} is obtained again by using the clock construction of \cite{janzing2006BQPmixing}, but noting that the Chebyshev polynomials enhance the argument from that reduction. 
Specifically, denoting the clock matrix as $H$, the remarkably convenient fact that the eigenvalues of $H$ happen to coincide precisely with the extrema of $T_M$ allows us to show that the problem is \BQPhard{} to \textit{constant} error. 

The intractability results \ref{enum:cheby_hard_unless} are obtained by exploiting the relationship between Hamiltonian Simulation and the Chebyshev polynomials expressed through the Anger-Jacobi expansion (see Def.~\ref{def:chebyshev_polynomials}). This result is achieved through a polynomial-time Turing reduction (rather than a Karp reduction), and thus we do not claim \BQP-\textit{completeness} for the problem\footnote{The notion of \BQP{}-\textit{completeness} is based on Karp-reductions, see Def.~\ref{def:bqp_hard_and_complete} for definitions. Meanwhile, a polynomial-time Turing reduction from problem $A$ to $B$ consists of an algorithm that is able to decide $A$ in polynomial-time using an oracle able to solve problem $B$.} under the constraints $\norm{A}{1}\leq 1$ or $\lambdaMatrix{A} \leq 1$. 

The classical algorithm of Result~\hyperref[enum:cheby_easy_log_degree]{\bf C.3.1} is obtained by observing that computing monomials is easy under the given hypothesis, and the error propagation due to the coefficients of the Chebyshev polynomials can be controlled because they can be bounded as $\OC(4^m)$ for the $m$-th polynomial. Finally, the results in~\hyperref[enum:cheby_easy_log_degree]{\bf C.3.2} are direct consequences of Prop.~\ref{prop:sparse_power_super_sparse} and Thm.~\ref{teo:pauli_power_super_sparse} (which we further generalize in the next Summary), while \hyperref[enum:cheby_easy_log_degree]{\bf C.3.3} is due to a general algorithm based on importance-sampling sketching on the Pauli coefficients, which we also elucidate in the next summary.

\medskip

The hardness results for monomials (and Chebyshev polynomials) are of course inherited by general polynomials. As for classical feasibility, for general polynomials, we can extend the ideas from Result~\ref{enum:monomials_super_sparse} and obtain the following classical algorithms:

\begin{summary}[Classical eigenvalue transform]\label{sum:gen_poly}
    Instantiate Problems~\ref{def:general_problem} and~\ref{def:general_problem-lm} with any degree-$d$ polynomial. 
    \begin{enumerate}[label=\textbf{P.\arabic*}]
        \item \label{enum:gen-func-sparse} If $A$ has $\OC(\polylog{N})$ non-zero matrix elements 
        in the
        computational basis both problems can be classically solved exactly in $\OC(d^2\cdot \polylog{N})$ time (Prop.~\ref{prop:sparse_power_super_sparse}). 
        \item \label{enum:gen-func-Pauli-sparse} If $A$ has $L=\OC(\log \log(N))$ non-zero 
        coefficients in the Pauli basis (implying that $A$ is 
        $\OC(\log \log(N))$-
        sparse
        ) both problems can be classically solved exactly in $\OC(d^2\cdot \polylog{N})$ time. (Thm.~\ref{teo:pauli_power_super_sparse})
        \item \label{enum:gen-func-reduced-norm} Further suppose the polynomial has magnitude at most $1$ on the domain $[-1,1]$ and $\|A\|\leq 1$. If $A$ has $L=\OC(\polylog{N})$ non-zero 
        coefficients in the Pauli basis (or efficiently-representable structure in the Pauli basis), both problems can be classically solved to constant error in $\OC(d^2\cdot \polylog{N})$ time, for $\lambda_A = \OC(1/(d^2 \log^{1.5}(N)))$ (Thm.~\ref{teo:pauli-sketch}).
    \end{enumerate}
\end{summary}

Results \ref{enum:gen-func-sparse} and \ref{enum:gen-func-Pauli-sparse} are direct extensions of Result~\ref{enum:monomials_super_sparse}. They provide classical algorithms for the eigenvalue transform for any polynomial, with runtimes that are polynomially equivalent to those of their corresponding quantum algorithms.
Importantly, these techniques are also 
straightforwardly transferable to the singular value transform for even functions. 
They demonstrate that 
no superpolynomial quantum advantage in the problem size is possible 
 when considering matrices that are super 
sparse, particularly when given in the Pauli basis. Observe that these results can be easily extended to general functions that can be efficiently approximated by polynomials, including the time evolution function and inverse function.

Result \ref{enum:gen-func-reduced-norm} comes from a two-stage algorithm. First, we perform importance sampling on Pauli coefficients to obtain a ``sketched" description of $A$ in the Pauli basis. We then apply the algorithm of \ref{enum:monomials_super_sparse}. We see that under the standard condition $d = \polylog{N}$, it would suffice to ask for $\lambda_A \leq c$ for some $c=\OC(1/\polylog{N})$ to guarantee an efficient classical algorithm. We saw in Result \ref{enum:cheby_hard_unless} that we should not expect much more from a classical algorithm, as solving the Chebyshev problem to non-constant precision for $\lambda_A \leq \OC(1/\polylog{N})$ is hard under complexity-theoretic assumptions. We will see in the following two summaries that for the inverse and time evolution functions this setting can even be shown to be \BQPcomplete{}.
{We remark that there is still a regime on which efficient classical algorithms for computing the Chebyshev polynomials may still be possible: namely, whenever $\norm{A}{1} \leq 1$ and the precision is constant.}

\medskip

We now consider the paradigmatic problem of time evolution, defined through the complex exponential function. For this, we prove the following:
\begin{summary}[Results for time evolution]\label{sum:time-evolution}
    Instantiate Problems~\ref{def:general_problem} and~\ref{def:general_problem-lm} with $f_t(x) = \exp(ixt)$. Unless explicitly stated, set $t,1/\varepsilon = \OC(\polylog{N}), \|A\|\leq 1$, and let $A$ be sparse or Pauli-sparse depending on the access model (that is, $s$ or $L = \OC(\polylog{N})$. Then, it holds that

    \begin{enumerate}[label=\textbf{T.\arabic*}]
        \item \label{enum:timeev_BQP_hard} The problems are $\BQPcomplete{}$ for some fixed precision $1/\varepsilon=\mathcal{O}(1)$ (Props.~\ref{prop:bqp_complete_ham_sim} and \ref{prop:bqp_complete_ham_sim_local}) whenever the matrix $A$: 

        \begin{enumerate}
            \item[(1)] Is given through the sparse access model and satisfies $\norm{A}{1} \leq 1$, or
            \item[(2)] Is given through the Pauli Access model, satisfies $\lambdaMatrix{A} \leq 1$ and is Pauli-sparse. 
            \item[(3)] Has even more strongly contracted norm $\norm{A}{1} =\OC(1/\polylog{N})$ (hence
            $\opNorm{A} =\OC(1/\polylog{N})$ also) in sparse access or $\lambdaMatrix{A} =\OC(1/\polylog{N})$ in Pauli access (Prop.~\ref{prop:bqp_complete_ham_sim_small-norm}).
        \end{enumerate}

        \item \label{enum:timeev_easy} The problems can be solved efficiently classically in $\OC(\polylog{N})$ time via
        \begin{enumerate}
            \item[(1)] A randomized algorithm whenever $\|A\|_1 t = \OC(\log \log N)$ in the sparse access model or $\lambdaMatrix{A} t = \OC(\log \log N)$ in the Pauli access model (Prop.~\ref{prop:ham-sim-classical}).

            \item[(2)] {A deterministic algorithm whenever $\opNorm{A} t = \OC(\log \log N)$ and $A$ is $\OC(1)$-sparse (Prop.~\ref{prop:constant_time_evolution})}.
            
            \item[(3)] A deterministic algorithm in the Pauli access model whenever $L= \OC(\log \log N)$ (Thm.~\ref{teo:function-super-sparse}).
        \end{enumerate}
        
    \end{enumerate}

\end{summary}

Similar to Chebyshev polynomials, here hardness holds robustly in a variety of settings, even when matrix norms are small and precision is a constant. 
As discussed previously, measurement of time-evolved states can be considered the canonical \BQPcomplete{} problem, dating back to the landmark proposal of Feynman \cite{feynman1982SimQPhysWithComputers}. To show BQP-completeness for the entry estimation variant of the problem for \ref{enum:timeev_BQP_hard}, we use a Hamiltonian clock construction of \cite{peres1985reversibleLogic}. Result \ref{enum:timeev_easy} arises due to an application of the classical algorithms of \cite{apers2022SimpleBetti} and \cite{wang2023qubitefflinalg} and our generalizations thereof to the local measurement problem (\hyperref[enum:monomials_classic_sparse]{\textbf{M.2.1}} \textit{and}  \ref{enum:monomials_classic_pauli}) applied to the truncated Taylor expansion strategy of \cite{berry2014HamSimTaylor}. 

\medskip

We also consider the inverse function. To frame it into the format 
of Problems~\ref{def:general_problem} and \ref{def:general_problem-lm}, we consider the family of functions given by
\newcommand{\inv}{\texttt{\upshape{inv}}}
\begin{align}
        \inv_m(x) = \begin{cases}
            \frac{1}{x} & \text{if } |x| \geq \frac{1}{m}\,,\\
            0 & \text{otherwise}\,.
    \end{cases}
\end{align}
The family $\inv_m$ coincides with the inverse function for any $A$ whose condition number $\kappaMatrix{A} = \|A\|\|A^{-1}\|$ is upper-bounded by $m$, and thus we will assume that the input matrix $A$ satisfies $\kappaMatrix{A} \leq m$.

\begin{summary}[Results for matrix inversion]\label{sum:inversion}
    Instantiate Problems~\ref{def:general_problem} and~\ref{def:general_problem-lm} with $f_m(x) = \inv_m(x)$. Unless explicitly stated, set $m,1/\varepsilon = \OC(\polylog{N}), \|A\|\leq 1$, and let $A$ be sparse or Pauli-sparse depending on the access model (that is, $s$ or $L = \OC(\polylog{N})$. Then, it holds that

    \begin{enumerate}[label=\textbf{I.\arabic*}]
        \item \label{enum:inversion_BQP_hard} The problems are $\BQPcomplete{}$ (Thm.~\ref{teo:inversion_bqp_complete} and Prop.~\ref{prop:inverse-BQP-complete-lm} \cite{harrow2009QLinSysSolver}) for a fixed precision whenever the matrix $A$ is either

        \begin{enumerate}
            \item[(1)] Given through the sparse access model and satisfies $\norm{A}{1} =\OC(1/\polylog{N})$ (hence
            $\opNorm{A} =\OC(1/\polylog{N})$ also), or 
            \item[(2)] Given through the Pauli access model, satisfies $\lambdaMatrix{A} =\OC(1/\polylog{N})$ and is Pauli-sparse. 
        \end{enumerate}
        
        \item \label{enum:inversion_easy} The problems can be solved classically in polynomial time  
        \begin{enumerate}
            \item For the sparse access model whenever $s,m = \OC(1)$ (exactly) or $\norm{A}{1}\leq1,\,m=\OC(1)$ (approximately) (Thm.~\ref{teo:matrix_inversion_fixed_kappa_sparse})
            \item For the Pauli access model whenever $\lambdaMatrix{A}, m = \OC(1)$ (Thm.~\ref{teo:matrix_inversion_fixed_kappa_sparse})
        \end{enumerate}
    \end{enumerate}
        
\end{summary}

The hardness part of Results \ref{enum:inversion_BQP_hard} are proven, once again, by exploiting the construction from \cite{janzing2006BQPmixing} in the case of entry estimation, and by using the construction of \cite{harrow2009QLinSysSolver} for the local measurement problem with a small tweak to hone the hardness result to Hermitian (real symmetric) matrices. 
These results are analogous to \ref{enum:cheby_BQP_hard} and \ref{enum:timeev_BQP_hard}, since it shows that even a strengthening of the condition $\opNorm{A} \leq 1$ does not lead to efficient classical algorithms for the case of the inverse function. This intractability result holds even for a constant precision level, but, unlike the case of Chebyshev polynomials, this is less surprising since we are considering additive precision and the inverse function commutes with scalar multiplication.\footnote{More precisely, given an algorithm that computes a $(k\varepsilon)$-approximation of $\bra{i} A^{-1} \ket{j}$, 
we can obtain an $\varepsilon$-approximation of the same value by using the same algorithm to compute a $(k\varepsilon)$-approximation of $\bra{i} \left( A/k\right)^{-1}\ket{j} = k \bra{i} A^{-1} \ket{j}$ and then dividing that approximation by $k$.} 

Result \ref{enum:inversion_easy} is deduced by employing a low-degree polynomial approximation of the inverse function in the range $\left[-1, -\frac{1}{m}\right] \cup \left[\frac{1}{m}, 1\right]$, which was developed in \cite{gilyen2018QSingValTransf}, and by approximating that polynomimal with the classical algorithms in Results \ref{enum:monomials_classic_sparse} and \ref{enum:monomials_classic_pauli}. We note that the classically solvable case is not immediately evident because, even when $m$ is fixed, the precision is still a parameter of the problem. Therefore, we cannot pick a fixed polynomial approximation of the inverse function (we need to pick an $\varepsilon$-approximation of the inverse function). Meanwhile, for the case of monomials and the Chebyshev polynomials, whenever $m$ is fixed the resulting function is a polynomial, and those can be computed exactly and straightforwardly classically if the matrix is sparse or Pauli-sparse (Lemma~\ref{lemma:power-tractable_fixed}).


\subsection{Discussion}
The ultimate goal of quantum algorithm development is to identify 
problems of practical interest for which there is a large quantum advantage in terms of computational resources, if possible quantified in an end-to-end fashion, accounting for all  necessary steps and not just specific sub-routines. Focusing on rigorous worst-case estimates, we have aimed at identifying what mathematical structure needs to be present in the input data and problem parameters for any significant quantum advantage to be 
{\it in principle} available. 
Concretely, we looked at the complexity of matrix functions for the matrix element problem (Problem \ref{def:general_problem})
and the local measurement problem (Problem \ref{def:general_problem-lm}) --- which gives a basic but, in our opinion, comprehensive framework for analyzing the potential of proposed quantum algorithms in different regimes of interest. Analyzing the impact of tuning a variety 
of problem-relevant parameters, as summarized in Fig.~\ref{fig:ven_diagram} and Table \ref{tab:summary}, we emphasize the following takeaways around classical simulability and hardness.

First, besides being useful as building blocks for approximating more general functions, as we argued before, monomials and Chebyshev polynomials also directly arise in classical/quantum walks \cite{Venegas_Andraca_2012} as well as in quantum Krylov methods \cite{Seki2021,Bespalova_2021, kirby2023exact, oleary2024partitioned}. For monomials and Chebyshev polynomials especially, previous results \cite{janzing2007simpleBQP} along with ours (see Results \ref{enum:monomials_hard}, \ref{enum:cheby_BQP_hard}, \ref{enum:cheby_hard_unless}) show that it is unlikely that certain quantum Krylov methods are classically simulable in desired parameter regimes.

Second, we identified a natural special case of sparse access which can render a subclass of otherwise BQP-complete problems classically easy:
Pauli access (\ref{enum:gen-func-Pauli-sparse} / Thms.~\ref{teo:pauli_power_super_sparse} and \ref{teo:function-super-sparse}). This  
highlights the importance of specifying precisely the access model when pondering quantum advantage, since a more specific model can be used to develop efficient classical algorithms under regimes for which the same problem might be hard in worst-case in other models.  

Third, we made concrete a hierarchy of hardness for matrix functions. Specifically, in the sparse access model, we show hardness for Chebyshev polynomials, time evolution, and matrix inversion for $\|A\|\leq 1/\polylog{N}$ and 
$s=\OC(1)$ (e.g.~see \ref{enum:cheby_hard_unless} / Thm.~\ref{teo:hardness_chebyshev_1_norm} for Chebyshev polynomials). Meanwhile, this regime is classically easy for monomials up to even constant-suppressed norm (\hyperref[enum:monomials_classic_sparse]{\bf M.2.2} / Thm.~\ref{teo:classically_easy_bounded_precision_and_norm}); and analogous statements hold true in the Pauli access model (see \hyperref[enum:monomials_classic_pauli]{\bf M.3.1} / Prop.~\ref{prop:matrix_power_bounded_norm_pauli_tractable} for classical algorithm). 
Similarly, we have also proven a concrete complexity separation in the regime $s=\OC(\polylog{N})$ as we actually show hardness for the same three functions under the even stronger normalization condition on the induced 1-norm $\|A\|_1\leq 1/\polylog{N}$, which contrasts with classically easy randomized algorithms for monomials that efficient process $\OC(\polylog{N})$-sparse matrices whenever $\|A\|_1\leq 1$ (\hyperref[enum:monomials_classic_sparse]{\bf M.2.3} / Prop.~\ref{prop:one_norm_classically_easy} \cite{apers2022SimpleBetti}); and, again, with analogous statement for Pauli access (\hyperref[enum:monomials_classic_pauli]{\bf M.3.2} / Prop.~\ref{prop:matrix_power_pauli_normalized_classically_easy} \cite{wang2023qubitefflinalg}).

Fourth, we stress that, in the sparse access model, the results from \cite{janzing2007simpleBQP} tell us that our known capabilities for randomized classical algorithms for monomials are (essentially) tight in dependence on $\|A\|_1$, in that the problem under the condition $\|A\|_1\leq1$ is 
in BPP but whenever $\|A\|_1\leq2$ is \BQPcomplete{}. We extend this also to 
the local measurement problem (\ref{enum:monomials_hard} / Prop.~\ref{prop:matrixPowerSparseAccessBQPcomplete-lm}).

Fifth, so far we find no concrete evidence in our settings of a difference in complexity between local measurement and matrix entry estimation.  
One thing we showed  is that standard techniques yield \BQP-hardness for monomials for \textit{constant} error for the \textit{normalized} local measurement problem (Prop.~\ref{prop:matrixPowerSparseAccessBQPcomplete-lm-normalized}, Appendix), whereas we only presently show \BQP-hardness for inverse polynomial error for the unnormalized version (\ref{enum:monomials_hard} /  Prop.~\ref{prop:matrixPowerPauliAccessBQPcomplete}).  

Last, we note that all of our classical algorithms for matrices given in Pauli basis actually allow for processing of arbitrary matrices with rank super-polylogarithmic in ${N}$. Thus, in this sense, ``high rank" does not guarantee quantum advantage, particularly when there is enough structure in Pauli basis. However, we also often exploit this structure for quantum algorithms, so care should be taken. Overall, we reinforce the idea that very specific matrix structure is needed for potential significant quantum speed-up. However, there are also still plentiful and significant gaps in the complexity landscape to further explore. Namely, some open questions motivated by our work are as follows: Can we find concrete differences in complexity between the local measurement and the entry estimation problems? One might then also study the normalized versions of the former. Further, are our classical Pauli algorithms for monomials tight in the way that the algorithms in the sparse access model are in the regime of constant sparsity? That is, can efficient classical algorithms be extended to $\lambdaMatrix{A} \geq 1$; or can the BQP-completeness result be strengthened to $\lambdaMatrix{A} = \OC(1)$ from $\lambdaMatrix{A} = \OC(\polylog N)$? More generally, we have seen a recurring theme through our results: that structure in the Pauli basis can be a powerful tool for classical algorithms. We leave it open whether there can be other interesting methods to exploit Pauli structure for classical algorithms.

\medskip

\textit{Note added.} In a recent arXiv update, Ref.~\cite{montanaro2024quantumclassicalquery} discusses a new idea to prove BQP-hardness of the entry estimation problem for matrix polynomials with $\|A\|\leq 1$ in the sparse access model. At present, there is a gap in the proof strategy, as only a proof of existence of a circuit-to-matrix mapping is known. It remains to be seen if there is an \textit{efficient} mapping for their approach, which is needed to complete a proof of BQP-hardness. 



\section{Technical background}

In this section, we present some standard definitions and results from the literature that will be useful for establishing our results in the next section. We also elucidate some generalizations thereof, which will be needed to consider the local measurement problem.

\paragraph{BQP and computational complexity.}

For completeness, we start by recalling the definitions of \BQP{}, \BQPhard{}, and \BQPcomplete{}. Then, function approximation results and classical algorithms are recalled.

\begin{definition}[Promise problems and \BQP{}]
    A promise problem $\Pi$ is given by two disjoint subsets of the set of binary strings $\Pi_{\mathrm{yes}}, \Pi_{\mathrm{no}} \subseteq \{0,1\}^\star$ that represent the set of positive and negative instances of $\Pi$, respectively.\footnote{The ``non-promise'' problems, known commonly as decision problems, consist of the subset of promise problems such that $\Pi_{yes} \cup \Pi_{no} = \{0,1\}^\star$.} 
    
    \BQP{} is the class of promise problems that can be solved by a uniform family of poly-sized quantum circuits.\footnote{Observe that we are defining \BQP{} as a class of promise problems. Formally, one might denote this class as \prBQP{}, while reserving \BQP{} for the class of non-promise problems. Nonetheless, using the term \BQP{} to refer directly to the promise class is common in the literature, and we follow this convention.} More precisely, a promise problem $\Pi = (\Pi_{\mathrm{yes}}, \Pi_{\mathrm{no}})$ is in \BQP{} whenever there is a family of circuits $\{C_{n}\}_{n \in \N}$ (where $C_n$ act on $r(n) = \poly{n}$ qubits and has $\poly{n}$ gates), a classical algorithm able to compute a description for $C_n$ in time $\OC(\poly{n})$ and these circuits satisfy that given $\boldsymbol{x} \in \{0,1\}^n$ it is the case that
    \begin{align}\label{eq:def_BQP}
        C_n \ket{\boldsymbol{x}} \zero{r(n) - n} = \alpha_{\boldsymbol{x},0} \ket{0} \ket{\psi_{\boldsymbol{x},0}} + \alpha_{\boldsymbol{x},1} \ket{1} \ket{\psi_{\boldsymbol{x},1}}\,,
    \end{align}
    where $\ket{\psi_{\boldsymbol{x},0}}, \ket{\psi_{\boldsymbol{x},1}}$ are $r(n)-1$ qubit states and
    \begin{enumerate}
        \item If $\boldsymbol{x} \in \Pi_{\mathrm{yes}}$ it holds that $|\alpha_{\boldsymbol{x},1}|^2 \geq \frac{2}{3}$.
        \item If $\boldsymbol{x} \in \Pi_{\mathrm{no}}$ it holds that $|\alpha_{\boldsymbol{x},1}|^2 \leq \frac{1}{3}$.
    \end{enumerate}
\end{definition}
\noindent In other words, a promise problem is within the class BQP if there exists a family of efficient quantum algorithms that solves the problem with high probability (at least $2/3$ for all problem instances). This probabilistic aspect is important to accommodate the intrinsic random nature of quantum measurements.

We will also need the notions of reductions, hardness, and completeness.

\begin{definition}[Karp reductions, \BQPhard{} and \BQPcomplete{}]\label{def:bqp_hard_and_complete}
    A promise problem $\Pi = (\Pi_{\mathrm{yes}}, \Pi_{\mathrm{no}})$ is Karp-reducible to the promise problem $\Pi' = (\Pi_{\mathrm{yes}}', \Pi_{\mathrm{no}}')$ whenever there is a classical polynomial-time algorithm $\mathcal{A}$ such that
    \begin{enumerate}
        \item If $\boldsymbol{x} \in \Pi_{\mathrm{yes}}$ it holds that $\mathcal{A}(\boldsymbol{x}) \in \Pi_{\mathrm{yes}}'$.
        \item If $\boldsymbol{x} \in \Pi_{\mathrm{no}}$ it holds that $\mathcal{A}(\boldsymbol{x}) \in \Pi_{\mathrm{no}}'$.
    \end{enumerate}
    A promise problem $\Pi'$ is \BQPhard{} whenever all promise problems $\Pi \in \BQP$ are Karp-reducible to $\Pi'$. If a promise problem $\Pi$ is both \BQP{} and \BQPhard{} then it is \BQPcomplete{}.
\end{definition}

With these, it is straightforward to prove that there is at least one \BQPcomplete{} problem.

\begin{observation}\label{obs:circuit_simulation_bqp_complete}
    The following problem is (naturally) \BQPcomplete{}.\newline 
    \noindent \problem{
    \problemBQPCircuitSimulation
    }{
    An $n$-bit string $\boldsymbol{x}$ and a circuit $C = U_T \ldots U_1$, with $T=\poly{n}$ and each $U_i$ acting non-trivially on at most 3 qubits (and therefore admiting a $\OC(1)$ classical description), that acts on $r=\poly{n} \geq n$ qubits as $C \ket{\boldsymbol{x}} \zero{r-n} = \alpha_{\boldsymbol{x},0} \ket{0} \ket{\psi_{\boldsymbol{x},0}} + \alpha_{\boldsymbol{x},1} \ket{1} \ket{\psi_{\boldsymbol{x},1}}$, where $\ket{\psi_{\boldsymbol{x},0}}$, $\ket{\psi_{\boldsymbol{x},1}}$, $\alpha_{\boldsymbol{x},1}$, and $\alpha_{\boldsymbol{x},0}$ are all unknown except for the promise that either $|\alpha_{\boldsymbol{x},1}|^2 \geq \frac{2}{3}$ or $|\alpha_{\boldsymbol{x},1}|^2 \leq \frac{1}{3}$.
    }{
    Empty promise
    }{
    Decide whether $|\alpha_{\boldsymbol{x},1}|^2 \geq \frac{2}{3}$ or rather $|\alpha_{\boldsymbol{x},1}|^2 \leq \frac{1}{3}$.}
\end{observation}

\paragraph{Hamiltonian simulation for sparse and Pauli access.} 
We will also invoke known connections between our quantum access models and another important oracle for quantum algorithms. Specifically, {sparse} and {Pauli} access models both can be employed to simulate the time evolution $e^{-iAt}$.

\begin{lemma}[Hamiltonian simulation for Sparse and Pauli access]\label{lemma:simulation_sparse_pauli}

    Given sparse access to s-sparse $A \in \C^{N \times N}$ it is possible to construct a circuit $U$ such that $\opNorm{e^{-iAt}-U}\leq \alpha$ in time polynomial in ($\frac{1}{\alpha}$, $t$, $s$, $\|A\|$, $\log N$) \cite{aharononv2007AdiabaticQStateGeneration, berry2005EffQAlgSimmSparseHam}.
    
    Similarly, given Pauli access to $A$ with Pauli norm $\lambdaMatrix{A}$ and number of Pauli matrices $L$ it is possible to construct a circuit $U$ such that $\opNorm{e^{-iAt} - U} \leq \alpha$ in time polynomial in ($\log\frac{1}{\alpha}$, $t$, $\log N$, $L$, $\lambdaMatrix{A}$) \cite{berry2014HamSimTaylor, low2016HamSimQubitization}.
    
\end{lemma}

\bigskip

\paragraph{Polynomial approximations.} 
Now, we introduce some general and well-known techniques for approximating the inverse function and the complex exponential with polynomials. These techniques are relevant because most of our classically simulable results will be based on approximating these functions by low-degree polynomials and using classically efficient algorithms to compute monomials or Chebyshev polynomials.

\begin{lemma}[Polynomial approximation of $\frac{1}{x}$ (\cite{childs2015QLinSysExpPrec}, Lemmas 17 and 19)]\label{lemma:approximate_inverse_function}

    The function
    \begin{align*}
        g(x) = \frac{1 - (1-x^2)^b}{x}
    \end{align*}
    $\varepsilon$-approximates the function\footnote{We say that a function $g$ $\varepsilon$-approximates a function $f$ in domain $\mathcal{D}$ when for all $x \in \mathcal{D}$ it holds that $|f(x) - g(x)| \leq \varepsilon$.} $f(x)=\frac{1}{x}$ in the domain $\left[-1,-\frac{1}{\kappa}\right] \cup \left[\frac{1}{\kappa}, 1\right]$ for any $b \geq \kappa^2 \log \left( \frac{\kappa}{\varepsilon} \right)$. Moreover, a polynomial of degree $\OC(\kappa\log(\kappa^2/\epsilon))$ that $\epsilon$-approximates $f(x)$ can be obtained from $g(x)$. 
    
\end{lemma}

\begin{lemma}[Polynomial approximation of $e^{ixt}$ (\cite{gilyen2018QSingValTransf}, Lemmas 57 and 59)]\label{lemma:anger_jacobi}
    Let $t \in \R \setminus \{0\}$ and $\varepsilon \in \left( 0, \frac{1}{e} \right)$.
    Then, the polynomial $J_0(t)+2\sum_{k=1}^Ri^k\,J_k(t)\,T_k(x)$ of degree $R = \Theta\left( t + \frac{\log(1/\varepsilon)}{\log(e+\log (1 / \varepsilon)/t)}\right)$ is a $2\varepsilon$-approximation to the function $e^{itx}$, where $J_k(t)$ are Bessel functions and $T_k(x)$ Chebyshev polynomials of the first kind.
\end{lemma}

\paragraph{Quantum algorithms for matrix function estimation.} 
In \cite{janzing2006BQPmixing}, the authors present an algorithm to compute $\bra{j} f(A) \ket{j}$ based on phase estimation. Briefly, computing $\bra{j} f(A) \ket{j}$ amounts to computing $\sum_{l} f(\eigen_l) |\langle j \ket{u_i}|^2$ where the state $\ket{u_i}$ represents the eigenvector of $A$ with corresponding eigenvalue $\eigen_l$. Phase estimation allows sampling of the eigenvalue $\eigen_l$ with probability $|\langle j \ket{u_i}|^2$, and therefore we obtain an estimator of $\bra{j} f(A) \ket{j}$ as long as we can control: (1) the error propagation that occurs when applying $f$ to the approximated eigenvalue computed by phase estimation and (2) the maximum error that may happen with some probability if phase estimation completely fails. Condition (1) can be achieved by bounding the Lipschitz constant $K_f$ of $f$, while (2) is obtained by bounding $\norm{f}{\infty}$. This strategy is formalized below:

\begin{lemma}[Quantum algorithm for entry estimation. (Janzing/Wocjan \cite{janzing2006BQPmixing}, Lemma 2)]\label{lemma:approximate_lipschitz_functions} 
    Let $A \in \C^{N \times N}$ be a Hermitian matrix such that $\opNorm{A} \leq 1$ and let $f:I\subseteq \R \rightarrow \R$ be a function satisfying $|f(x) - f(y)| \leq K_f |x-y|$ for all $x,y \in I$, where $K_f$ is a constant. Let a circuit $U$ be given such that $\opNorm{U - \exp (iA)} \leq \alpha$ using resources that scale polynomially in $\log(N)$ and $\sfrac{1}{\alpha}$. Then, given a state $\ket{\psi}$ whose decomposition into $A$-eigenvectors contains eigenvalues only in the interval $I$ we can estimate $\bra{\psi} f(A) \ket{\psi}$ up to error $\varepsilon (\norm{f}{\infty} + K_f)$ with probability at least $1-\delta$ with time and space resources polynomial in $\log(N)$, $\sfrac{1}{\varepsilon}$ and $\log \left( \sfrac{1}{\delta} \right)$. 
    
    
\end{lemma}
\noindent For matrices given by the access models of Defs.~\ref{def:sparseAccess} and \ref{def:pauliAccess}, Lemma~\ref{lemma:simulation_sparse_pauli} gives efficient ways of constructing the operator $U$. Therefore, in these cases, for any function satisfying $\norm{f}{\infty},K_f=\OC(\polylog{N})$, this algorithm runs in polynomial time if also $1/\varepsilon=\OC(\polylog{N})$. For instance, when $f_m(x) = x^m$ it is the case that $K_{f_m}= m$ and $\norm{f_m}{\infty} = 1$, and the \BQP{} algorithm follows.
 Since this algorithm computes $\bra{\psi} f(A) \ket{\psi}$ for any $\ket{\psi}$, it is easy to estimate any $\bra{i} f(A) \ket{j}$ observing that these non-diagonal terms can be expressed as a sum of  ``diagonal'' ones in different bases (Lemma~\ref{lemma:non-diagonal-hermitian}).

We demonstrate a quantum algorithm for the local measurement version of the problem.

\begin{lemma}[Quantum algorithm for normalized local measurement]\label{lemma:approximate_lipschitz_functions-lm}
    Consider a function $f:I \rightarrow [-f_{max}, f_{max}]$ which satisfies $|f(x) - f(y)| \leq K_f |x-y|$ for all $x,y \in I$, and where the smallest discontinuity in $I$ is of size $b$.  
    The normalized local measurement $\bra{f(A)}(\ketbra{0}{0} \otimes \identity_{N/2}) \ket{f(A)}$ (denoting $\ket{f(A)} = f(A)\ket{0}/\|f(A)\ket{0}\|$) can be solved for matrix $A$ with spectrum contained in $I$, given in sparse or Pauli access, to additive error $\varepsilon \leq \|f(A)\ket{0}\| b/2$ with cost polynomial in ($K_f$,$1/\|f(A)\ket{0}\|$, $\log{N}$, $1/\varepsilon$).
\end{lemma}

\begin{proof}
    Denote the eigendecomposition of $A$ as $A=\sum_i \eigen_i \ketbra{u_i}{u_i}$ and write the zero state in this basis as $\ket{0}=\sum_{i}\beta_{i} \ket{u_i}$. Let us now consider the following sequence of operations:
    \begin{align}
        \ket{0}\ket{0}\ket{0} = \sum_{i}\beta_{i} \ket{u_i} \ket{0}\ket{0} \xrightarrow{\textsf{QPE}} \sum_{i}\beta_{i} \ket{u_i}\ket{\tilde{\eigen}_i}\ket{0} &\xrightarrow{\textsf{C-R, QPE$^{-1}$}} \sum_{l}\beta_{i} \ket{u_i} \left( \frac{f(\tilde{\eigen}_i)}{f_{max}} \ket{0} + \sqrt{1-\frac{f^2(\tilde{\eigen}_i)}{f^2_{max}}} \ket{1} \right) \\
        &\xrightarrow{\textsf{AA}} \sum_{i}\frac{f(\tilde{\eigen}_i)}{\|f(A) \ket{0}\|}  \beta_{i}\ket{u_i} \approx \frac{1}{\|f(A) \ket{0}\|}f(A)\ket{0} \,, 
    \end{align}
    where $\textsf{QPE}$ denotes quantum phase estimation to an error $|\widetilde{{\theta_i}} - {\theta_i}| \leq  \varepsilon'$ for all $i$ (with cost $\OC(1/\varepsilon')$ \cite{kitaev1995QuantumMeasurement, nielsen2002QCQI, lin2022LectureNotes}), $\textsf{C-R}$ denotes a rotation of the third register controlled on the second register, and $\textsf{AA}$ denotes amplitude amplification \cite{brassard2002AmpAndEst} by factor ${f_{max}}/{\|f(A) \ket{0}\|}$. 
    We remark that $\tilde{\eigen}_l$ may fall outside of $I$ by $\varepsilon'$ --- this can be resolved by extending $f$ to be a total function whose value outside of $x \notin I$ corresponds to $f(y)$ of $y \in I$ closest to $x$. We then should only consider $\varepsilon' \leq b/(2K_f)$.  The Lipschitz condition ensures that the output state is an approximation of $f(A)\ket{0}/\|f(A) \ket{0}\|$ with additive error $\varepsilon'K_f/\|f(A) \ket{0}\|$ in $\ell_2$-norm. The total cost in cumulative Hamiltonian simulation time of $\textsf{QPE} + \textsf{AA}$ is ${f_{max}}/(\|f(A) \ket{0}\|\varepsilon')$, which can be simulated with linear cost using the algorithm of~\cite{low2016HamSimQSignProc}. The stated result can be checked to follow by choice of $\varepsilon'= \varepsilon\|f(A) \ket{0}\|/(2K_f)$.  
\end{proof}

\paragraph{Classical algorithms for matrix functions.} We start by stating a lemma on sparse matrix multiplication. This same idea was used to classically compute matrix polynomials in Ref.~\cite{gharibian2022dequantizingSVT}.

\begin{lemma}[Classical algorithms for matrix powers]\label{lemma:power-tractable_fixed}
    Given sparse access to a $N \times N$ $s$-sparse matrix $A$, for any indices $(i,j)$ it is possible to compute $\left[A^m\right]_{i,j}$ exactly in time $\OC(s^m)$ classically. Similarly, $\bra{i} A^{m_1} \pi A^{m_2} \ket{j}$ can be computed exactly in time $\OC(s^{m_1 + m_2})$ classically. 
\end{lemma}

\begin{proof}
    The first statement is possible by using a matrix multiplication algorithm recursively:

    \begin{enumerate}
        \item If $m=1$ then return $A_{i,j}$.

        \item Else, find the $\ell$ non-zero entries of the $i$-th row of $A$, and name their positions $k_1,\ldots,k_\ell$, where $\ell \leq s$. Then compute recursively the entries $\{(k_1, j), \ldots, (k_\ell, j) \}$ of $A^{m-1}$, and return $A_{i,k_1} A^{m-1}_{k_1, j} + \ldots + A_{i,k_\ell} A^{m-1}_{k_\ell, j}$.
    \end{enumerate}

    Let $R(m)$ be the runtime for this algorithm given the degree $m$. Then
    $R(m) = s R(m-1) + s$ and therefore $R(m) = \OC(s^m)$. We can also use this same idea for $\bra{i} A^{m_1} \pi A^{m_2} \ket{j}$, namely:

    \begin{enumerate}
        \item If $m_1=m_2=0$ we return $\bra{i} \pi \ket{j} = \bra{i} (\ket{0}\bra{0} \otimes \identity_{N/2}) \ket{j}$ which can be computed straightforwardly.

        \item If $m_2 > 0$ we note that
        \begin{align}
            \left[A^{m_1} \pi A^{m_2} \right]_{i,j} = \left[A^{m_1} \pi A^{m_2-1} \right]_{i,t_1} A_{t_1,j} + \ldots + \left[A^{m_1} \pi A^{m_2-1} \right]_{i,t_p} A_{t_p,j}   
        \end{align}

        where $t_1,\ldots, t_p$ with $p\leq s$ are the row indices of the non-zero entries of column $j$ of $A$.

        \item If $m_2 = 0$ and $m_1 > 0$ we apply an analogous strategy.
    \end{enumerate}

    This algorithm has complexity $\OC(s^{m_1 + m_2})$.
\end{proof}

The inspiration to consider different normalization factors comes from recent classical algorithms that allow computing $p(A)$ for any polynomial $p$ and whose complexity depends on $\norm{A}{1}$, the Pauli norm $\lambdaMatrix{A}$ and the coefficients of the polynomial $p$. Roughly, these algorithms work by considering the matrix $A$ as a description of a Markov chain, and thus, Monte Carlo techniques can estimate any entry of $A^m$.

\begin{lemma}[Classical sampling algorithm for classical sparse access (\cite{montanaro2024quantumclassicalquery} Proposition 5.5, \cite{apers2022SimpleBetti} Lemma 3.4)]\label{lemma:montanaro_classical}
    Let $f(x) = \sum_{r=0}^m \alpha_r x^r$. Then, there is an algorithm that, given sparse access to an $s$-sparse matrix $A$ and two indices $i,j$, computes an $\varepsilon$-approximation of $\bra{i} f(A) \ket{j}$ with probability at least $1-\delta$ in time 
    \begin{align*}
        \OC \left( \frac{m \,s}{\varepsilon^2} \norm{f(\norm{A}{1} x)}{\ell_1}^2 \log \left(\frac{1}{\delta} \right)\right)\,,
    \end{align*}
    where we denote $\norm{f(b x)}{\ell_1} = \sum_{r=0}^m |\alpha_r b^r|$.    
\end{lemma}

\begin{lemma}[Classical sampling algorithm for Pauli access (adapted from \cite{wang2023qubitefflinalg}, Proposition 3)]\label{lemma:wang_classical}

    Let $f(x) = \sum_{r=0}^m \alpha_r x^r$.
    and assume $\ell_1$-sampling access to the Pauli coefficients of $A = \sum_{\ell} a_\ell P_\ell$. That is, suppose $\lambda_A = \sum_{\ell=1}^L |a_\ell|$ is known and there is an efficient sampler who returns the tuple $(\ell, \operatorname{sign}(a_\ell))$ with probability $\frac{|a_{\ell}|}{\lambda_A}$. Then there is an algorithm that, given two indices $(i,j)$, computes an $\varepsilon$-additive approximation of $\bra{i} f(A) \ket{j}$ with probability at least $1-\delta$ in time
    \begin{align*}
        \OC\left( \frac{m \log(N)}{\varepsilon^2}\norm{f(\lambdaMatrix{A} x)}{\ell_1}^2 \log \left( \frac{1}{\delta}\right)\right)\,,
    \end{align*}
    where we denote $\norm{f(b x)}{\ell_1} = \sum_{r=0}^m |\alpha_r b^r|$. If no sampling access is available, it can be provided in $\OC(L)$ time as a preprocessing step starting from Pauli access.
\end{lemma}

The above two algorithms can be adapted to the local measurement problem: 

\begin{lemma}[Classical sampling algorithm for both access models and local measurement]\label{lemma:randomized_classical_lm}
    Given $f(x) = \sum_{r=0}^m \alpha_r x^r$, there exist classical algorithms that yield $\varepsilon$-additive approximations to $\bra{i}f(A) \pi f(A) \ket{i}$ for $\pi = \ketbra{0}{0}\otimes \identity_{N/2}$ with probability at least $1-\delta$ in time  
\begin{equation}
    \widetilde{\OC} \left( \frac{m s}{\varepsilon^2} \norm{f(\norm{A}{1} x)}{\ell_1}^4 \log \left( \frac{1}{\delta} \right)\right)\quad \text{and} \quad \OC\left( \frac{m \log(N)}{\varepsilon^2}\norm{f(\lambdaMatrix{A}\, x)}{\ell_1}^4 \log \left( \frac{1}{\delta} \right)\right)\,,
\end{equation}
for $A$ given in classical sparse access or Pauli access, respectively.
\end{lemma}

\begin{proof}[Proof sketch]
    The central primitive we need is a method to compute quantities of the form $\bra{i} A^a\, \pi\, A^b \ket{i}$ --- with this, the local measurement of any $f(A)$ which is a probabilistic combination of monomials can be efficiently returned by sampling over such quantities. This idea can then be extended to general linear combinations. For example, let us discuss obtaining $\bra{i} A^a\, \pi\, A^b \ket{i}$ in Pauli access. We note that $\pi$ is the convex sum of two Pauli strings, and thus we observe that $\bra{i} A^a\, \pi\, A^b \ket{i}$ can be written as a linear combination of terms of the form $\bra{i} P_{\ell_1}\cdots P_{\ell_a}\, P_\pi\, P'_{\ell'_1}\cdots P_{\ell'_b}' \ket{i}$, where each $P_{\chi}$ denotes some Pauli string. Considering the linear combination as a normalized probability distribution, we can sample from these terms (along with any accompanying phase) with variance upper bounded by $(\lambda_A^a\lambda_A^b)^2$. Each such term costs $\OC((a+b)\log N) = \OC(m\log N)$ time to explicitly and exactly evaluate. We provide a full proof in Appendix \ref{appdx:additional-results}. 
\end{proof}

Together, Lemmas \ref{lemma:montanaro_classical}, \ref{lemma:wang_classical}, \ref{lemma:randomized_classical_lm} immediately imply (as a special case) efficient classical algorithms for monomials when problem parameters scale polylogarithmically ($m, 1/\varepsilon, s \ \text{or}\ L = \OC(\polylog{N})$), when $\|A\|_1 \leq 1$ or $\lambdaMatrix{A} \leq 1$ for sparse access or Pauli $\ell_1$-sampling access (obtainable in $\OC(L) $ preprocessing time), respectively.

\section{Detailed statements and main proofs}\label{sec:detailed-statements}

In this section, we study the complexity of different instantiations of Problems~\ref{def:general_problem} and \ref{def:general_problem-lm}. Each matrix function studied is presented in a separate subsection that begins by defining a {promise problem form} of the respective problem.

For brevity, we will use the following macros to refer to the different access models for a given matrix $A$:

\begin{itemize}
    \item \sparseAccess{}: sparse access to $A$, as per Def.~\ref{def:sparseAccess}, assuming that $A$ is $s$-sparse.
    \item \PauliSparse: Pauli query access to $A$ as per Def.~\ref{def:pauliAccess}, assuming that $A$ is Pauli-sparse. 
    \item \accessModel{}: a placeholder for any access model when defining the problems.
\end{itemize}

We will use the symbol $b(A)$ as placeholder notation for any norm of $A$.

\subsection{Monomials}

In this section we present in detail our hardness results and classical algorithms for Problems~\ref{def:general_problem} and~\ref{def:general_problem-lm} when $f_m(x) = x^m$, i.e.~the family of matrix powers or monomials. Let us start with hardness results. To this end we adapt the formal definition of the promise problem from \cite[Def.~4.1]{janzing2007simpleBQP}, where the authors introduce the \problemDiagonalEntryEstimation{} problem. It consists in computing the value $[A^m]_{j,j}$, given a sparse and real symmetric matrix $A$ satisfying $\opNorm{A} \leq 1$, an integer $j$ and a power $m$. They prove this problem is \BQPcomplete{} when $A$ is given in sparse access. We define an analogous problem, which will also allow us to analyze Pauli access. 
\newline

\noindent
\problem{
\problemMatrixPower{b(A)}{\accessModel{}}
}{
A $N \times N$ Hermitian matrix $A$ with norm $b(A) \leq 1$ and accessible through \accessModel{}, a positive integer $m$, index $j$, a precision $\varepsilon$ and a threshold $g$, such that $m, \sfrac{1}{\varepsilon} = \OC(\polylog{N}), g = \OC(1)$.
}{$b$ is an upper bound on $\opNorm{A}$, $m$ and $\frac{1}{\varepsilon}$ are $\polylog{N}$.
}{\textsc{YES} if $[A^m]_{j,j} \geq g + \varepsilon$, \textsc{NO} if $[A^m]_{j,j} \leq g - \varepsilon$.
}

We also define the local measurement version of the problem.\newline 

\noindent
\problem{
\problemMatrixPowerLM{b(A)}{\accessModel{}}
}{
A $N \times N$ Hermitian matrix $A$ with norm $b(A) \leq 1$ and accessible through \accessModel{}, a positive integer $m$, a precision $\varepsilon$ and a threshold $g$, such that $m, \sfrac{1}{\varepsilon} = \OC(\polylog{N}), g = \OC(1)$.
}{
$b$ is an upper bound on $\opNorm{A}$, $m$ and $\frac{1}{\varepsilon}$ are $\polylog{N}$.
}{
Let $\pi = \ket{0}\bra{0} \otimes \identity_{N/2}$ and $r = \bra{0}A^m \pi A^m\ket{0}$. Then, answer \textsc{YES} if $r \geq g + \varepsilon$ and \textsc{NO} if $r \leq g - \varepsilon$.
}


Note that, as in \cite{janzing2007simpleBQP}, we consider a more restricted class of matrices for the promise problems to show hardness: real symmetric matrices, rather than Hermitian matrices.  First, we recount and adapt the main result and proof of  \cite{janzing2007simpleBQP} in a way that will be relevant for our discussions.

\begin{theorem}[]\label{teo:matrixPowerSparseAccessBQPcomplete}

    The problem \problemMatrixPower{\opNorm{A}}{\sparseAccess} is \BQPcomplete{}. The hardness result holds even under the condition that the matrix $A$ {is 5-local and} satisfies $\norm{A}{1} \leq 2$, and under the restriction that $A$ is real symmetric.
    
\end{theorem}

\newcommand{\overlap}{P}

\newcommand{\expression}{E}

\begin{proof}
    Regarding membership in \BQP{}, the algorithm described in Lemma~\ref{lemma:approximate_lipschitz_functions} can be employed to compute {$\bra{j} f_m(A) \ket{j}$ for any $j \in \until{N}$} when $f_m(x) = x^m$ in polynomial time using Lemma~\ref{lemma:simulation_sparse_pauli} and observing that $\norm{f_m}{\infty}^{[-1,1]} = 1$ and $K_{f_m} \leq m$. Thus, \problemMatrixPower{\opNorm{A}}{\sparseAccess{}} is in $\BQP{}$. Moreover, any non-diagonal entry can be expressed as a sum of diagonal terms on a different basis (see Lemma~\ref{lemma:non-diagonal-hermitian}).

To prove hardness, consider $C = U_{T} \ldots U_1$ as the $r$-qubit input circuit to \problemBQPCircuitSimulation{} (Observation~\ref{obs:circuit_simulation_bqp_complete}).    
The reduction we are going to show defines a sparse Hermitian matrix, for which it is possible to construct sparse access. Moreover, there will be a diagonal element of a monomial of that matrix such that the circuit accepts $\ket{\inputVector{x}}$ if that entry contains a value above some threshold. We are free to pick the gate set from which $C$ is composed; it will turn out that if the gates $U_1,\ldots,U_{T}$ are assumed to be either Hadamard or Toffoli gates (both having real entries and forming a universal set of gates), then the reduction defines a real symmetric matrix, which is a special case of a Hermitian matrix. From here on, and in all proofs to follow, we assume this decomposition. 
It will be useful to consider a new circuit $C'$ obtained from $C$, defined as $C' = U_1^{\dagger} \ldots U_{T}^\dagger (Z \otimes \identity^{r-1}) U_{T} \ldots U_1 =: V_{M-1} \ldots V_0$ with $M = 2T+1$ (see also Figure~\ref{fig:circuit-extension}).

    {
    We will use a unary clock to keep track of the computation steps, where state $\step{k} = \ket{0}^{\otimes k} \ket{1} \ket{0}^{\otimes M-k-1}$ denotes the $k$th computation step. Thus, the $k$th clock transition ($k$ to $k+1$) can be described by the operator $\clockTransition_k = \identity^{\otimes{k}} \otimes \ket{01}\bra{10} \otimes \identity^{M-k-2}$ for $k<M-1$, and $\clockTransition_{M-1} = \ket{1}\bra{0} \otimes \identity{}^{M-2} \otimes \ket{0} \bra{1}$. It holds that
    \begin{align}
        \clockTransition_\ell\, \step{k} = \begin{cases}
            \step{k+1} & k=l\,,\\
            0 & \text{otherwise}\,,
        \end{cases}
    \end{align}
    (where $+$ is understood modulo $M$) and that, for any $d$,
    \begin{align}\label{eq:relation_transition_operator}
        \prod_{k = 0}^{M-1} \clockTransition_{k+d} = \step{d} \brastep{d} \,,
    \end{align}
    where the addition $+$ is understood modulo $M$.

    Let $\ket{s_{\inputVector{x}}} = \step{0} \ket{\inputVector{x}} \ket{0}^{r-n}$ with $\ket{\inputVector{x}} \ket{0}^{r-n}$ the input bitstring to \problemBQPCircuitSimulation{}. The circuit $C'$ operates on $\ket{\inputVector{x}} \ket{0}^{r-n}$ as the identity $\identity$ if $|\alpha_{\inputVector{x}, 0}|^2 =1$, while if $|\alpha_{\inputVector{x}, 1}|^2 = 1$ it behaves as $-\identity$. We define
    \begin{align}\label{eq:def_w}
        W = \sum_{\ell=0}^{M-1} \clockTransition_\ell \otimes V_\ell\,.
    \end{align}

    One can check that, due to Eq~\eqref{eq:relation_transition_operator},
    \begin{align}\label{eq:W_M}
        W^M = \sum_{\ell = 0}^{M-1}  \step{\ell} \brastep{\ell} \otimes V_{(\ell + M + 1)} \ldots V_\ell\,,
    \end{align}
    and, because $(C')^2 = \identity$, we have that $(W^M)^2 = \sum_{\ell=0}^{M-1} \step{\ell} \brastep{\ell} \otimes \identity^{\otimes r}$. This implies that $(W^M)^2$ behaves as the identity in the subspace spanned by the set $\{\step{\ell} \brastep{\ell} \otimes \identity^{\otimes r}\}_{0 \leq \ell \leq M-1}$, and from now on we restrict all our analysis to that subspace. Note that $W^M$ only has $\pm 1$ eigenvalues there, and therefore let $\mathcal{S}^\pm$ be the $W^M$-invariant subspaces associated with the projectors $Q^{\pm} = \frac{1}{2} (\identity \pm W^M)$.
    }
    
    We observe that the action of $W$ over $\mathcal{S}^\pm$ is isomorphic to a cyclic shift (with a phase shift in the case of $-$). Starting from $Q^{\pm} \ket{s_{\inputVector{x}}}$, the cycle travels across the vectors $W^\ell Q^{\pm} \ket{s_{\inputVector{x}}}$ for $\ell =0,\ldots, M-1$. Using this property, we show in Lemma \ref{lemma:W_spectrum} that the eigenvalues of $W$ take values 
    $e^{i2\pi \ell / M}$ (for eigenstates in $\mathcal{S}^+$) and  $e^{i \pi (2\ell + 1) / M}$ (for eigenstates in $\mathcal{S}^-$). Moreover, we can evaluate the overlap 
    \begin{equation}\label{eq:overlap}
       \omega_\ell^+= \bra{s_{\inputVector{x}}}P^+_\ell\ket{s_{\inputVector{x}}}=\frac{|\alpha_{\inputVector{x},0}|^2}{M},
    \end{equation}
    where $P^+_\ell$ is the projector onto $\mathcal{P}_\ell^+$ -- the eigenspace of $W$ corresponding to eigenvalue $e^{i2\pi \ell / M}$. Similarly, $\omega_\ell^-=\bra{s_{\inputVector{x}}}P^-_\ell\ket{s_{\inputVector{x}}}=\frac{|\alpha_{\inputVector{x},1}|^2}{M}$, where $P^-_\ell$ is the projector onto $\mathcal{P}_\ell^-$ -- the eigenspace of $W$ corresponding to eigenvalue $e^{i \pi (2\ell + 1) / M}$. 
    %

\begin{figure}[th!]
        \centering
        \begin{tikzpicture}[y=-1cm,scale=1.5]
            \draw (5,-0.1) -- (0,-0.1) node[left] {$\lvert x_1\rangle$};
            \draw (5,0.2) -- (0,0.2) node[left] {$\lvert x_2\rangle$};
            \node[align=center] at (0.1,0.45) {$\vdots$};
            \draw (5,0.8) -- (0,0.8) node[left] {$\lvert x_n\rangle$};
            \draw (5,1.1) -- (0,1.1) node[left] {$\zero{r-n}$};

            \draw[fill=white] (0.6,-0.25) rectangle node {$C$} (1.4,1.25);

            \draw[fill=white] (2.3,-0.25) rectangle node {$Z$} (2.7,0.05);

            \draw[fill=white] (3.6,-0.25) rectangle node {$C^{\dagger}$} (4.4,1.25);
            
        \end{tikzpicture}
        \caption{Circuit extension $C'$ of the input circuit $C$ used in \cite{janzing2006BQPmixing}. Note that the amplitude $\bra{\inputVector{x}}\bra{0} C' \ket{\inputVector{x}}\ket{0}$ is linearly related to the measurement probability that decides the problem. In turn, the estimation of such amplitude reduces to estimating an element of a monomial of a sparse matrix $A$, defined in Eq.~\eqref{eq:A}.}
        \label{fig:circuit-extension}
    \end{figure}
    
    Now, consider the Hermitian (real symmetric) matrix 
    \begin{align}
    \label{eq:A}
        A = \frac{W + W^\dagger}{2}\,,
    \end{align}
    where $W$ is defined in Eq.~\eqref{eq:def_w} we note that $A$ is sparse with sparsity $s=4$ due to the fact that each $V_\ell$ is a either a Hadamard or Toffoli gate, and also $5$-local because the clock transitions are 2-local and each circuit gate is at most 3-local. We can also check that $\norm{A}{1} \leq 2.$  Note also that each eigenvector $\ket{\psi_\ell^{+}}\in \mathcal{S}^+$ of $W$ with eigenvalue $e^{i2\pi \ell/M}$ is also an eigenvector of $W^\dagger$, but with eigenvalue $e^{-i2\pi \ell/M}$, and the same happens for the eigenvectors $\ket{\psi_{\ell}^{-}}\in \mathcal{S}^-$. Therefore, %
    $A\ket{\psi_\ell^{+}}  = \frac{e^{i2\pi \ell/M}\ket{\psi_\ell^{+}} + e^{-i2\pi\ell/M}\ket{\psi_\ell^{+}}}{2} = \cos \left( \frac{2\pi \ell}{M}\right)\ket{\psi_\ell^{+}}$ and similarly $A\ket{\psi_\ell^{-}}  = \cos \left( \frac{\pi(2 \ell +1)}{M}\right)\ket{\psi_\ell^{+}}$.
    We denote the eigenvalues of $A$ as $\eigen_\ell^{+}=\cos \left( \frac{2 \pi \ell}{M}\right)$ (corresponding to $\mathcal{P}_\ell^+$) and $\eigen_\ell^{-}=\cos \left( \frac{\pi (2\ell + 1)}{M} \right)$ (corresponding to $\mathcal{P}_\ell^-$) for $\ell = 0,\ldots, M-1$. Some properties of these eigenvalues will be useful. 
    First, $\eigen_0^+=1$ and $\eigen_\ell^+$ and $\eigen_{M-\ell}^+$ coincide for $\ell = 1,\ldots, \frac{M-1}{2}$ since $\cos \left( \frac{2 \pi \ell}{M}\right) = \cos \left( \frac{2 \pi (M-\ell)}{M}\right)$. Second, $\eigen_{\frac{M-1}{2}}^-=-1$ and $\eigen_\ell^-=\eigen_{M-\ell-1}^-$ for $\ell = 0,\ldots, \frac{M-1}{2}-1$ since $\cos \left( \frac{\pi (2\ell + 1)}{M}\right) = \cos \left( \frac{\pi (2(M-\ell-1) + 1)}{M}\right)$. Lastly, we also observe that $\eigen_{\frac{M-1}{2}-\ell}^{-}=-\eigen_\ell^+$ for $\ell = 0,\ldots, \frac{M-1}{2}$.  These properties yield the spectral decomposition
    \begin{equation}\label{eq:spectral_A}
         A = P_0^+ - P_{\frac{M-1}{2}}^- + \sum_{l=1}^{\frac{M-1}{2}} \eigen_\ell^+ \, (P_\ell^+ +P_{M-\ell}^+ - P_{\frac{M-1}{2}-\ell}^- - P_{\frac{M-1}{2}+\ell}^-).
    \end{equation}
    Denote 
    $\ket{j} = \ket{s_{\inputVector{x}}}$
    , with $j$ the integer whose binary representation is $s_{\inputVector{x}}$. Then, for a function $f_m(A)$
    \begin{align}\label{eq:value_of_diagonal_entry}
        [f_m(A)]_{j,j} &= f_m(1)\,\omega_0^++f_m(-1)\omega_{\frac{M-1}{2}}^-+\sum_{l=1}^{\frac{M-1}{2}}f_m(\eigen_\ell^+)\,\left(\omega_\ell^+ +\omega_{M-\ell}^+\right) +f_m(-\eigen_\ell^+)\,\left(\omega_{\frac{M-1}{2}-\ell}^- + \omega_{\frac{M-1}{2}+\ell}^-\right)\nonumber \\
         &=\frac{|\alpha_{\inputVector{x},0}|^2}{M}\left[f_m(1)+2\sum_{l=1}^{\frac{M-1}{2}}f_m(\eigen_\ell^+)\right]+\frac{|\alpha_{\inputVector{x},1}|^2}{M}\left[f_m(-1)+2\sum_{l=1}^{\frac{M-1}{2}}f_m(-\eigen_\ell^+)\right],
    \end{align}
    where we employed Eqs.~\eqref{eq:spectral_A} and \eqref{eq:overlap}, and the fact that the projectors are orthogonal. Thus, for $f_m(A) = A^m$, with $m$ odd, we can explicitly write 
    \begin{align}\label{eq:diagonal_in_terms_of_EO}
        [A^m]_{j,j} = \frac{(1- 2|\alpha_{\inputVector{x},1}|^2)}{M} \left[1+2\sum_{l=1}^{\frac{M-1}{2}}\left(\eigen_\ell^+\right)^m\right]:=(1- 2|\alpha_{\inputVector{x},1}|^2)\,E_0,
    \end{align}
    where we have used the fact that $1-|\alpha_{\inputVector{x},1}|^2 = |\alpha_{\inputVector{x},0}|^2$, and denoted $E_0 = \frac{1}{M}(1+2\sum \left(\eigen_\ell^+\right)^m )$.  
    
    Note that if $|E_0|$ is not too small, then by computing $[A^m]_{j,j}$ one can recover the acceptance probability $|\alpha_{\inputVector{x},1}|^2$ of the original circuit $C$.
    Precisely, observe that
    \begin{align}\label{eq:preliminary_bound}
        \expression_0 \geq \frac{1}{M}\left(1+2\cdot\frac{M-1}{2}\left(\eigen_{\frac{M-1}{2}}^+\right)^m\right)\geq \frac{1}{M} + \left( \eigen_{\frac{M-1}{2}}^{+} \right)^m,
    \end{align}
    where the first inequality follows by observing that the eigenvalues are enumerated in decreasing order. 
    Since, $\eigen_{\frac{M-1}{2}}^{+} = \cos \left( \frac{\pi (M-1)}{M} \right)<0$ and $m$ is odd, we need to take $m$ big enough for the right-hand side of Eq. \eqref{eq:preliminary_bound} to be sufficiently positive.
    One can check that by picking $m=M^3$ (which is odd) 
    we can ensure that $\expression_0 > \frac{3}{4M}$. 
    Finally, using Eq.~\eqref{eq:diagonal_in_terms_of_EO} this implies that if $|\alpha_{\inputVector{x},1}|^2 \leq \frac{1}{3}$ then $[A^m]_{j,j} \geq \frac{\expression_0}{3} > \frac{1}{4M}$, and whenever $|\alpha_{\inputVector{x},1}|^2 \geq \frac{2}{3}$ then $[A^m]_{j,j} \leq -\frac{\expression_0}{3} < -\frac{1}{4M}$. Thus, we can decide which of the two cases holds for $|\alpha_{\inputVector{x},1}|^2$ by computing a $\frac{1}{4M}$-approximation of $[A^m]_{j,j}$, where we recall that $M= \OC(\polylog{N})$. 
    
    The Karp mapping goes as
    \begin{align}
        \left(C = U_{T}\ldots U_1, \inputVector{x}\right) \to \left(-A, \ket{j} = {\step{0}}  \ket{\inputVector{x}} \zero{r-n}, m = (2T+1)^3, g=0, \varepsilon = \frac{1}{4(2T+1)}\right)\,.
    \end{align}
    This can be computed in polynomial time, and, according to our previous arguments, it is correct, i.e., it maps positive (negative) instances of the first problem to positive (negative) instances of the second one).

\end{proof}

We kept the presentation of the proof general up to Eq.~\eqref{eq:value_of_diagonal_entry}, which reads for any function $f$ that
\begin{align}
    [f_m(A)]_{j,j} 
    =\frac{1}{M}\left[f_m(1)+2\sum_{l=1}^{\frac{M-1}{2}}f_m(\eigen_\ell^+)\right]-\frac{|\alpha_{\inputVector{x},1}|^2}{M}\left[f_m(1)-f_m(-1)+2\sum_{l=1}^{\frac{M-1}{2}}(f_m(\eigen_\ell^+) - f_m(-\eigen_\ell^+))\right].
\end{align}
The first term is independent of the quantum circuit being simulated, and can be evaluated exactly efficiently.
Thus, we see that this proof strategy for hardness actually shows that any family of matrix functions $f_m$ for entry estimation is \BQPhard{} for inverse error $1/\varepsilon = \OC(1/k)$ if the condition
\begin{equation}\label{eq:jantzig_condition2}
    \frac{1}{M}\left|f^{o}_m(1)+2\sum_{l=1}^{\frac{M-1}{2}}f^{o}_m(\eigen_\ell^+)\right| \geq k
\end{equation}
is satisfied for some $m$ where $f^o_m$ denotes the odd contribution of $f_m$, and for any $M = \OC(\polylog{N})$ (the even contribution cancels). Roughly, this inequality states that for entry estimation to be \BQPhard{} for $f_m$ it is sufficient that $f^o_m$ varies fast enough in some subinterval of $[-1, 1]$. If it is not true then each term $f^{o}_m(\cos (2\pi \ell / M))$ with $0 \leq \ell \leq \frac{M-1}{4}$ would cancel out its ``almost'' opposite term $f^{o}_m(\cos(2\pi (\frac{M-1}{2} - \ell) / M)) = f^{o}_m(-\cos(\pi - 2\pi (\frac{M-1}{2} - \ell) / M)) = -f^{o}_m(\cos(2\pi(\ell+1)/M))$. The same condition applies to the Pauli access model, as we will see in Proposition \ref{prop:matrixPowerPauliAccessBQPcomplete}.

We remark that while quantum algorithms can also find off-diagonal matrix entries (for instance, use Lemma \ref{lemma:non-diagonal-hermitian} in the Appendix to write off-diagonal entries as a linear combination of diagonal entries in some other basis), and our classical algorithms will also be able to manage general matrix entries, the above proof shows hardness even for the restricted problem of computing a diagonal entry. Moreover, it turns out the same hardness result can be also shown when restricting the problem to strictly off-diagonal entries, and this idea was also shown in \cite{janzing2007simpleBQP}. Similar tricks can be used in the forthcoming BQP-completeness results for Chebyshev polynomials and the inverse function.

\begin{remark}[Hardness for off-diagonal entries]
    The variant of the problem \problemMatrixPower{}{} that computes a strictly off-diagonal matrix entry $[A^m]_{i,j}$ for $i\neq j$ can also be shown to be BQP-complete. Hardness can be shown by tensoring the matrix in \eqref{eq:A} with an idempotent matrix 
    $B = A \otimes \left(\begin{smallmatrix} 1/2 & 1/2 \\ 1/2 & 1/2  \end{smallmatrix}\right)$ which satisfies the property $B^m = A^m \otimes \left(\begin{smallmatrix} 1/2 & 1/2 \\ 1/2 & 1/2  \end{smallmatrix}\right)$. Thus, diagonal entries of $A^m$ are encoded into off-diagonal entries of $B^m$.
\end{remark}

The arguments from the proof of Theorem~\ref{teo:matrixPowerSparseAccessBQPcomplete} can be adapted to also work for the Pauli query access model.

\begin{proposition}\label{prop:matrixPowerPauliAccessBQPcomplete}

    The problem \problemMatrixPower{\opNorm{A}}{\PauliSparse{}} is \BQPcomplete{}. As in Theorem~\ref{teo:matrixPowerSparseAccessBQPcomplete}, the hardness results holds {even if $A$ is 5-local}, $\norm{A}{1} \leq 2$, and for real symmetric $A$.
    
\end{proposition}

\begin{proof}
The \BQP{} membership of \problemMatrixPower{\opNorm{A}}{\PauliSparse{}} again follows from a direct application of Lemma~\ref{lemma:approximate_lipschitz_functions} alongside Lemma~\ref{lemma:simulation_sparse_pauli}, noting that Lemma \ref{lemma:simulation_sparse_pauli} is also applicable for Pauli access when $L, \lambda= \OC(\polylog{N})$.

{

The \BQP{}-\textit{hardness} is shown by providing Pauli-sparse query access to the matrix $W = \sum_{\ell=0}^{M-1} \clockTransition{}_\ell \otimes V_\ell$ from the proof of Theorem~\ref{teo:matrixPowerSparseAccessBQPcomplete} (from which one can build the Pauli access to $A = \frac{W + W^\dagger}{2}$). By Lemma~\ref{lemma:pw-clock-construction}, we know that each term $\clockTransition{}_\ell$ can be written with $\OC(1)$ Pauli terms and Pauli norm $1$. Also, by Lemma~\ref{lemma:pw-unitary} each $V_\ell$ can be decomposed in $\OC(1)$ Pauli matrices with Pauli norm $\OC(1)$. Finally, by Lemma~\ref{lemma:pw-multiplicative} we conclude that each term $\clockTransition{}_\ell \otimes V_\ell$ can be written with $\OC(1)$ Pauli terms  and $\OC(1)$ Pauli norm, and thus the Pauli decomposition of $A$ has norm $\lambda_A=\OC(M)$ and $L=\OC(M)$ terms, which are $\OC(\polylog{N})$ for polynomial-sized circuits.
}

\end{proof}

We also obtain an analogous result for the local measurement version of the problem.

\begin{proposition}\label{prop:matrixPowerSparseAccessBQPcomplete-lm}
    The problems \problemMatrixPowerLM{\opNorm{A}}{\sparseAccess{}} and \problemMatrixPowerLM{\opNorm{A}}{\PauliSparse{}} are \BQPcomplete{}. The hardness results hold even if {the matrix is $5$-local,} $\norm{A}{1} \leq 2$, and for real symmetric $A$.
\end{proposition}

\begin{proof}
    To show inclusion for both access models we use a two stage algorithm. First, we use Lemma \ref{lemma:approximate_lipschitz_functions} to evaluate $\|A^{m}\ket{0}\|^2 = |\bra{0} A^{2m} \ket{0}| $ to additive error $\varepsilon/3$. If this value is $\leq g+\varepsilon/2$ then we output \textsc{NO} as we can guarantee that $\bra{0}A^{m}\mathsf{\pi}A^{m}\ket{0} \leq g+5\varepsilon/6$ by H\"older's tracial matrix inequality. If the value is otherwise $> g+\varepsilon/2$ then we can guarantee that $\|A^{m}\ket{0}\|^2 > g+\varepsilon/6 = \Omega(\varepsilon)$ and we use Lemma \ref{lemma:approximate_lipschitz_functions-lm} to evaluate $\bra{0}A^{m}\mathsf{\pi}A^{m}\ket{0}/\|A^{m}\ket{0}\|^2$ to additive error $\varepsilon/3$, with runtime $\OC(m/\varepsilon\|A^m\ket{0}\|^2) = \OC(m/\varepsilon^2) = \OC(\polylog{N})$, where we have used the fact that the Lipschitz constant of $x^{m}$ is $m$. 
    Multiplying the two outputs (additive estimates for $\|A^{m}\ket{0}\|^2$ and $\bra{0}A^{m}\mathsf{\pi}A^{m}\ket{0}/\|A^{m}\ket{0}\|^2$) together gives the desired quantity to additive error $2\varepsilon/3 + \varepsilon^2/3 \leq \varepsilon$ (assuming $\varepsilon\leq 1$, else we can rescale the errors appropriately).

    To show hardness, we will simulate any circuit with $T$ gates on $r$ qubits $C = U_{T}...U_1$. As in previous proofs, we assume this circuit is built wholly from Hadamard and Toffoli gates.
    Recall we would like to approximate $|\alpha_{\inputVector{x},1}|^2$ in Eq.~\eqref{eq:def_BQP}. We will consider the measurement $\pi = \ketbra{1}{1} \otimes \identity$ for convenience, which is an arbitrary choice by adding a final gate in the circuit. Consider the sequence of unitaries $C'= V_{M} ... V_1 = U_{1}^{\dag}...U_{T}^{\dag}\, C_{\textsc{NOT}}\, \identity ... \identity\, C_{\textsc{NOT}}\, U_{T}...U_1$, where we have padded the sequence with $(T+1)$ identities and 2 CNOT gates, so that $M=3(T+1)$. $C'$ acts on $r+1$ qubits, with $U_{T}...U_1$ acting on the lower $r$ registers, and  $C_{\textsc{NOT}}$ acting on the first two registers, controlled on the second one. We note that $C_{\textsc{NOT}}(\ketbra{1}{1} \otimes \identity)C_{\textsc{NOT}} = \ketbra{10}{10} + \ketbra{01}{01}$.
    Given the first register being in the zero state, the $C_{\textsc{NOT}}$ gate toggles the measurement to be on/off on the second register. This, in turn, implies that for any input state $\ket{\inputVector{x}} \in \C^{2^{r}}$, we have 
    \begin{equation}\label{eq:clock-flag-stepfn}
        \bra{0,\inputVector{x}}V^{\dag}_1...V^{\dag}_{\ell} \pi^{(r+1)} V_{\ell}...V_1 \ket{0,\inputVector{x}} = \begin{cases}
          \bra{\inputVector{x}}C \pi^{(r)} C \ket{\inputVector{x}} = |\alpha_{\inputVector{x},1}|^2 & \text{if}\;\  T+1\leq {\ell}\leq 2T+2
            \\
            0  & \text{otherwise} \,,
        \end{cases}
    \end{equation}
    where we denote $\pi^{(m)}= \ketbra{1}{1} \otimes \identity^{(m)}$ where $\identity^{(m)}$ is the $m$-qubit identity matrix 
    
    {
    Let us now consider the operator
    \begin{equation}\label{eq:monomial-walk-operator}
        A = \frac{1}{2} \sum_{{\ell}=0}^{M-1} \left(\clockTransition{}_\ell \otimes V_{{\ell}+1} + \clockTransition{}_\ell^\dagger \otimes V^{\dag}_{{\ell}+1}\right),
    \end{equation}
    
    Similar to Eq.~\eqref{eq:A}, we note that $A$ in Eq.~\eqref{eq:monomial-walk-operator} is also a sparse matrix with sparsity $s=4=\mathcal{O}(1)$ (the gates we simulate have sparsity at most $2$) and 5-local.
    When successive powers of $A$ are applied to the inital state $\step{0} \otimes \ket{0,\inputVector{x}}$, 
    where the first register is the ``clock", a classical random walk is performed over the 
    the following $M$ quantum states 
    \begin{align}
    \label{eq:M_states}
        \step{\ell} \otimes V^{({\ell})} \ket{0,\inputVector{x}} = \begin{cases}
            \step{\ell} \otimes  \ket{0,\inputVector{x}} & \text{for}\;\ {\ell}=0  \\
            \step{\ell} \otimes U_{\ell} ...U_1 \ket{0,\inputVector{x}} & \text{for}\;\ 1\leq {\ell} \leq T \\
            \step{\ell} \otimes C_{\mathrm{NOT}} U_{T} ...U_1 \ket{0,\inputVector{x}} & \text{for}\;\ T+1 \leq {\ell} \leq 2T+2 \\
            \step{\ell} \otimes  U_{M-\ell}...U_1 \ket{0,\inputVector{x}} & \text{for}\;\ 2T+3 \leq {\ell} \leq M\,,
        \end{cases}
    \end{align}
    where we have denoted $V^{({\ell})} = V_{\ell}...V_1$.
    We stress from this equation that the state in the latter $(n+1)$ registers is wholly determined by the state of the 
    clock register; it is agnostic to the path taken. 
    }
    

    From Eq.~\eqref{eq:clock-flag-stepfn} we see that only amplitudes corresponding to $T+1\leq \ell \leq 2T+2$ contribute a non-zero measurement 
    probability, and for the choice of input state $\ket{\inputVector{x}}=\zero{r}$ the measurement 
    probability is exactly $|\alpha_{\inputVector{0},1}|^2$ as desired.  Considering $\inputVector{x}=\inputVector{0}$ is sufficient since any other input state can be prepared with $r$ additional gates, which we absorb into our definition of $C$. All that remains is to evaluate the amplitude corresponding to $T+1\leq \ell \leq 2T+2$ for a given value of matrix power $t$.

    Random walks on a $1$D chain are well-studied and known to be rapidly mixing. The random walk in question is represented by the $M$-component probability distribution $\inputVector{p}_m$ (here, the power $m$ of $A$ labels the random-walk iteration), which  approaches the uniform distribution $\inputVector{u}=\{\frac{1}{M},...,\frac{1}{M}\}$ as 
    \begin{equation}\label{eq:rapid-mixing}
        \left\| \inputVector{p}_{m} - \inputVector{u} \right\|_1 \leq \frac{1}{2} \exp \left( - \frac{\pi^2}{2} \frac{m}{M^2}\right)\,,
    \end{equation}
    for any 
    iteration $m\geq M^2 \geq 49$ (e.g.~see \cite[Theorem 2.3]{berestycki2016mixing}). Explicit evaluation thus gives 
    \begin{equation}
        \label{eq:sum_prob}
        \bra{0}A^m\pi^{(r)}A^m\ket{0} = |\alpha_{\vec{0},1}|^2 \!\sum_{T+1\leq \ell \leq 2T+2} p^2_{m}(\ell)\,.  
    \end{equation}
    For $p_{\infty}=u$, the last sum equals $ \frac{T+2}{M^2} \geq \frac{1}{3M}$. Our final step will be to show that {for finite but} large enough $m$ the ratio is still $\Omega(1/M)$, and thus any problem in BQP can be 
    decided by solving \problemMatrixPowerLM{\opNorm{A}}{\sparse} for precision $\varepsilon=\mathcal{O}(1/M)$.
    
    For an arbitrary distribution $\inputVector{p}_m$ such that $\left\| \inputVector{p}_m - \inputVector{u} \right\|_1 = \varepsilon$, the sum in Eq.~\eqref{eq:sum_prob}  is minimized when $p_{m}$ takes uniform values $\frac{1}{M}-\frac{\varepsilon}{2(T+2)}$ across all $T+1\leq \ell \leq 2T+2$. Thus, we can bound the sum as 
    \begin{align}
        \sum_{T+1\leq \ell \leq 2T+2} p^2_{m}(\ell) &\geq \sum_{T+1\leq \ell \leq 2T+2} \left(\frac{1}{M}- \frac{\varepsilon}{2(T+2)}  \right)^2 \\
        &\geq \frac{(T+2)(2-\frac{\varepsilon}{T+2})^2}{4M^2} \\
        &> \frac{1}{6M}\,,
    \end{align}
    where the last inequality is true for any $\varepsilon\leq 1$. From  Eq.~\eqref{eq:rapid-mixing}, we see that it is thus sufficient to take ${m}= \OC(M^2)$
    , which scales as $\OC(\polylog{N})$ because so does the target circuit size $T$. This ensures that one can estimate $|\alpha_{\inputVector{0},1}|^2$ up to constant precision via Eq.~\eqref{eq:sum_prob} by estimating the amplitude $\bra{0}A^m\pi^{(r)}A^m\ket{0}$ up to precision $\varepsilon=\OC(1/\polylog{N})$. 
    
    Note that the matrix $A$ consists of $\OC(M) = \OC(\polylog{N})$ non-zero entries and can be instantiated via sparse access efficiently.  Moreover, it is Pauli-sparse as computational basis entries and constant-dimension unitaries have efficient Pauli decompositions (Lemmas \ref{lemma:pw-unitary} and \ref{lemma:pw-clock-construction} respectively), and Pauli norms are multiplicative (Lemma \ref{lemma:pw-multiplicative}). Thus, hardness holds in the Pauli access model as well.
\end{proof}

We may also ask whether \problemMatrixPowerLM{\opNorm{A}}{\accessModel{}} becomes harder when we consider its ``normalized form'' (that is the quantity expected in the output is a normalized measurement result of quantum states $A^m\ket{0}/\|A^m\ket{0}\|$). We provide some indication this could be the case via Proposition \ref{prop:matrixPowerSparseAccessBQPcomplete-lm-normalized} in Appendix \ref{appdx:additional-results}: using exactly the same construction, one can show that the normalized problem is \BQPhard{} even for \textit{constant} error.


\bigskip 

Now let us think about classical algorithms.  Theorem~\ref{teo:matrixPowerSparseAccessBQPcomplete} and Proposition~\ref{prop:matrixPowerSparseAccessBQPcomplete-lm} imply that both problems \problemMatrixPower{\opNorm{A}}{\sparseAccess{}} and \problemMatrixPowerLM{\opNorm{A}}{\sparseAccess{}} are \BQPcomplete{} if the inverse precision scales polynomially with the input size and the matrix $A$ satisfies $\opNorm{A} \leq 1$. If we strengthen the second condition then they become classically solvable:

\begin{theorem}\label{teo:classically_easy_bounded_precision_and_norm}
    Let $\eta : \N \to \R$ and assume that for $A \in \C^{N \times N}$ it holds that $\opNorm{A} \leq 1 - \eta(N)$.\footnote{From now on, we omit the dependence on $N$ of $\eta$ and simply write $\eta$ to denote $\eta(N)$.}  Then, the problem \problemMatrixPower{\opNorm{A}}{\sparseAccess{}} 
    can be solved classically in time $\OC\left( \left(\frac{1}{\varepsilon}\right)^{-\log (s) / \log \left(1-\eta \right)} \right)$ for any value of $m$. Whenever $s=\OC(1)$ and $\eta=\Omega(1)$ this algorithm works in polynomial time. 
    \problemMatrixPowerLM{\opNorm{A}}{\sparseAccess{}} can be solved using similar ideas with polynomially equivalent complexity. 
\end{theorem}

\begin{proof}
    Observe that if $\opNorm{A} \leq 1-\eta$ it holds that

    \begin{align}
        [A^m]_{j,j} = |\bra{j}A^m\ket{j}| &= \left|\sum_\lambda \lambda^m |\bra{j} \lambda \rangle|^2 \right| \leq \sum_{\lambda} |\lambda|^m |\bra{j} \lambda \rangle|^2 \nonumber\\
        &\leq \sum_{\lambda} \left(1-\eta\right)^m |\bra{j} \lambda \rangle|^2 = \left(1-\eta\right)^m \,.
    \end{align}

    Therefore, whenever $m > \frac{\log \varepsilon}{\log \left(1-\eta\right)}$ is the case that 0 is an $\varepsilon$-approximation of $\bra{i}A^m \ket{i}$. Meanwhile, if $m \leq \frac{\log \varepsilon}{\log \left(1-\eta\right)}$ we can use the algorithm from Lemma~\ref{lemma:power-tractable_fixed} to compute the answer in time $\OC(s^{m}) = \OC\left(s^{\log (\varepsilon) / \log \left(1-\eta\right)}\right) = \OC\left( \left( \frac{1}{\varepsilon} \right)^{- \log (s) / \log \left(1-\eta\right)} \right)$.

    Regarding $\bra{0} A^m \pi A^m \ket{0}$, we can use Lemma~\ref{lemma:power-tractable_fixed}  whenever $m < \frac{\log \varepsilon}{2\log \left( 1 - \eta \right)}$. Otherwise, $0$ is a sufficient $\varepsilon$-approximation.
\end{proof}

We recall that the condition $s = \OC(1)$ alone should still yield hard problems for classical algorithms 
(observe that the \BQPhard{} proofs from Theorem~\ref{teo:matrixPowerSparseAccessBQPcomplete} and Proposition~\ref{prop:matrixPowerSparseAccessBQPcomplete-lm} rely only on $4$-sparse matrices). Thus, it is the additional condition on the norm of $A$ which allow for classical algorithms. We note that similar conditions on the norm are applied to construct quantum block-encodings on matrices \cite[Theorem 30]{gilyen2018QSingValTransf}, \cite[Lemma 4.5]{montanaro2024quantumclassicalquery}, though here $\eta = 1/\polylog{N}$ may be tolerated at a cost of only $\polylog{N}$ gate overhead. 

We can obtain an analogous result for the Pauli access model.

\begin{proposition}\label{prop:matrix_power_bounded_norm_pauli_tractable}
    Assume the the matrix $A \in \C^{N\times N}$ satisfies the condition $\opNorm{A} \leq 1 - \eta$. Then, the problem \problemMatrixPower{\opNorm{A}}{\PauliSparse{}{}} 
    can be solved classically to precision $\varepsilon$ with success probability at least $1-\delta$ in time complexity 
    \begin{equation}
        \OC\left( \frac{\log \varepsilon}{\log(1 - \eta)} \frac{\log N}{\varepsilon^2} \left( \frac{1}{\varepsilon} \right)^{- 2\log (\lambdaMatrix{A}) / \log \left(1-\eta\right)} \log\left( \frac{1}{\delta} \right) \right)\,.
    \end{equation}
    Whenever $\lambdaMatrix{A} = \OC(1),\eta = \Omega(1)$ this algorithm works in polynomial time. Similarly, \problemMatrixPowerLM{\opNorm{A}}{\PauliSparse{}} 
    can also be solved classically with polynomially equivalent complexity in Pauli access.
\end{proposition}

\begin{proof}
    Regarding Problem \ref{def:general_problem}, 
    as in the proof of Theorem~\ref{teo:classically_easy_bounded_precision_and_norm}, whenever $m > \frac{\log \varepsilon}{\log \left(1-\eta\right)}$ it is the case that 0 is an $\varepsilon$-approximation of $\bra{i} A^m \ket{j}$. If $m \leq \frac{\log \varepsilon}{\log \left(1-\eta\right)}$ use the algorithm from Lemma~\ref{lemma:wang_classical}. An analogous reasoning can be applied to the Problem \ref{def:general_problem-lm} 
    employing Lemma~\ref{lemma:randomized_classical_lm} when $m$ is small.
\end{proof}

As a matter of completeness, we now briefly list the consequences of direct application of the randomized algorithms in Lemmas \ref{lemma:montanaro_classical}, \ref{lemma:wang_classical} and \ref{lemma:randomized_classical_lm}. Both problems can be solved classically for the sparse access model if $\norm{A}{1} \leq 1$ using the techniques developed in \cite{montanaro2024quantumclassicalquery}. Similarly, for the Pauli case, whenever $\lambdaMatrix{A} \leq 1$ we can employ techniques from \cite{wang2023qubitefflinalg}.

\begin{proposition}\label{prop:one_norm_classically_easy}
    The problems \problemMatrixPower{\opNorm{A}}{\sparseAccess{}} {and \problemMatrixPowerLM{\opNorm{A}}{\sparseAccess{}}} can be solved classically with probability at least $1-\delta$ in time $\widetilde{\OC}(\frac{sm}{\varepsilon^2}\norm{A}{1}^{2m} \log \left( \frac{1}{\delta}) \right)$ {and $\widetilde{\OC}(\frac{sm}{\varepsilon^2}\norm{A}{1}^{4m} \log \left( \frac{1}{\delta} \right))$, respectively. Hence, i}n particular, \problemMatrixPower{\norm{A}{1}}{\sparseAccess{}} {and \problemMatrixPowerLM{\norm{A}{1}}{\sparseAccess{}}} can be solved classically in polynomial time.
\end{proposition}

\begin{proof}
    {For the matrix element problem, c}onsider the algorithm from Lemma~\ref{lemma:montanaro_classical}. In this case $f_m(x) = x^m$ and it holds that $\norm{f_m(\norm{A}{1}x)}{l_1} = \norm{A}{1}^m$. In turn, for the local measurement problem, use the algorithm from Lemma~\ref{lemma:randomized_classical_lm} {in a similar way}.
\end{proof}





\begin{proposition}\label{prop:matrix_power_pauli_normalized_classically_easy}
    The problems \problemMatrixPower{\opNorm{A}}{\PauliSparse{}} and \problemMatrixPowerLM{\opNorm{A}}{\PauliSparse{}} can be solved classically in time $\OC\left(m \log (N) \frac{\lambdaMatrix{A}^{2m}}{\varepsilon^2} \log \left( \frac{1}{\delta} \right)\right)$ and $\OC\left(m \log (N) \frac{\lambdaMatrix{A}^{4m}}{\varepsilon^2} \log \left( \frac{1}{\delta} \right)\right)$ respectively both with success probability at least $1-\delta$. In particular, \problemMatrixPower{\lambdaMatrix{A}}{\PauliSparse{}} and \problemMatrixPowerLM{\lambdaMatrix{A}}{\PauliSparse{}} can be solved classically in polynomial time.
\end{proposition}

\begin{proof}
    Use the algorithms from Lemmas \ref{lemma:wang_classical} and \ref{lemma:randomized_classical_lm}.
\end{proof}

Finally, we introduce a notion of \textit{super sparsity}, under which we can solve both problems in both access models. Roughly, a super-sparse matrix in the sparse access model has $\OC(\polylog{N})$ non-zero entries, while a super-sparse Pauli query access only has $\OC(\log \log N)$ non-zero coefficients.

\begin{proposition}\label{prop:sparse_power_super_sparse}
    
     Consider an access model which lists all $k$ non-zero entries in the computational basis via triples $\{a_{ij},i,j\}_{(i,j)\in S}$ such that $|S|=k$ and $A = \sum_{(i,j) \in S} a_{ij} \ket{i}\bra{j}$. We say that we have classical \textsc{Super-sparse} access to $A$ when when $k= \polylog{N}$. 
    Under this access, both Problems \ref{def:general_problem} and \ref{def:general_problem-lm} can be solved exactly for a monomial of power $m$ in time complexity $\OC(m\,k^3)$. Thus, \problemMatrixPower{\opNorm{A}}{\textsc{Super-sparse}} and \problemMatrixPowerLM{\opNorm{A}}{\textsc{Super-sparse}} can be solved in $\polylog{N}$ time.
\end{proposition}

\begin{proof}
    For any $d \in \N$, 
    $A^d$ contains at most $k^2$ different projectors of the form $\ket{i_{\ell_1}} \bra{j_{\ell_2}}$, with $\ell_1,\ell_2 \in \until{k}$. Thus, we can directly compute the coefficient $a_{i_{\ell_1},j_{\ell_2}}^{(d+1)}$ associated to each 
    projector $\ket{i_{\ell_1}} \bra{j_{\ell_2}}$ of $A^{d+1}$ 
    given the coefficients for $A^d$, as

    \begin{align}
        a_{i_{\ell_1},j_{\ell_2}}^{(d+1)} =  \sum_{i'_\ell} a_{i_{\ell_1}, i_\ell'} a_{i_{\ell}', j_{\ell_2}}^{(d)}
    \end{align}
    Hence,
    $A^m$ can be computed with 
    $\OC(m\,k^3)$ elementary operations. Finally, with the resulting 
    explicit description of $A^m$ both problems can be easily solved.
\end{proof}

\begin{theorem}\label{teo:pauli_power_super_sparse}

    There is an exact 
    classical 
    algorithm that solves \problemMatrixPower{\opNorm{A}}{\PauliSparse} in time $\OC(m\, L\, 2^L \log N)$. 
    Similarly, there is a similar algorithm that solves \problemMatrixPowerLM{\opNorm{A}}{\PauliSparse{}} in time complexity 
    $\OC((m\,L\,2^L + 2^{2L})\log N)$. 
    In particular, both
    algorithms run in
    time 
    $\polylog{N}$ whenever $A$ is super Pauli-sparse, i.e., whenever $L = \OC(\log \log(N))$. 
\end{theorem}

\begin{proof}
    Given $A = \sum_{\ell=1}^L a_\ell P_{i_{\ell}}$, let $\mathcal{G} = \{P_{i_{\ell}}\}_{\ell\in \until{L}}$
    and $\langle\mathcal{G}\rangle$ be the Pauli sub-group generated by $G$ without considering global phases. Then, b
    y  Lemma~\ref{lemma:size_generalized_generated} it holds that $|\langle \mathcal{G} \rangle| \leq 2^{L+1}$. Thus, the Pauli decomposition of $A^k$ involves at most $2^{L+1}$ terms, for any $k \in \N$, and we can compute $A^m$ in a bottom-up manner as follows: Given $A^k$ we can compute $A^{k+1} = A A^{k}$ by computing all the $\OC(L\,2^{L})$ products between the non-zero terms of $A$ and non-zero terms of $A^k$ (each product taking $\log N$ time to evaluate). To compute $A^m$ we need to perform this operation $m$ times, and each step costs $\OC(L\,2^L\log N)$. 
    
    Finally, given the explicit Pauli representation of $A^m$, computing $\bra{i} A^m \ket{j}$ is straightforward by explicit sparse-matrix multiplication and takes time $\OC(2^L \log N)$. Meanwhile, to compute $\bra{0} A^m \pi A^m \ket{0}$, one can expand $A^m$ and compute each term: Assuming $A^m = \sum_{q=1}^R b_q P_{i_q}$, it follows that 
    \begin{equation}
        \bra{0} A^m \pi A^m \ket{0} = \sum_{q_1,q_2=0}^R b_{q_1} b_{q_2} \bra{0} P_{i_{q_1}}
    \pi P_{i_{q_2}}
    \ket{0}\,.
    \end{equation}
    There are $R^2=\OC(2^{2L})$ terms, and each one can be computed in time $\OC(\log N)$.
\end{proof}

In Observation \ref{obs:pauli_power_super_sparse} we remark that Theorem \ref{teo:pauli_power_super_sparse} automatically allows exact and efficient computation of any polynomial (such as Chebyshev polynomials) whenever $L= \OC(\log\log N)$.

\subsection{Chebyshev polynomials}\label{sec:chebyshev_sparse_access}

We consider the following problems: \newline

\noindent
\problem{
\problemMatrixChebyshevPolynomial{b(A)}{\accessModel{}}
}{
A $N \times N$ Hermitian matrix $A$ with norm $b(A) \leq 1$ and accessible through \accessModel{}, a positive integer $m$, index $j\in \until{N}$, a precision $\varepsilon$ and a threshold $g$, such that $m$, $1/\varepsilon,g = \OC(\polylog{N})$.
}{
Empty promise
}{
\textsc{YES} if $[T_m(A)]_{j,j} \geq g + \varepsilon$. \textsc{NO} if $[T_m(A)]_{j,j} \leq g - \varepsilon$.

}

\noindent
\problem{
\problemMatrixChebyshevPolynomialLM{b(A)}{\accessModel{}}
}{
A $N \times N$ Hermitian matrix $A$ with norm $b(A) \leq 1$ and accessible through \accessModel{}, a positive integer $m$, a precision $\varepsilon$ and a threshold $g$, such that $m$, $1/\varepsilon,g = \OC(\polylog{N})$.
}{
Empty promise
}{
Let $\pi = \ket{0}\bra{0} \otimes \identity_{N/2}$ and $r = \bra{0} T_{m}(A) \pi T_{m}(A)\ket{0}$. Then, answer \textsc{YES} if $r \geq g + \varepsilon$ and \textsc{NO} if $r \leq g - \varepsilon$.
}

The problems can be solved through the phase estimation-based algorithm from Lemma~\ref{lemma:approximate_lipschitz_functions}. Moreover, we can prove \BQP-\textit{completeness} when the operator norm is used as the bound condition.

\begin{theorem}\label{teo:chebychev-sparse-alternative-BQPcomplete}
    The problems \problemMatrixChebyshevPolynomial{\opNorm{A}}{\sparseAccess{}} and \problemMatrixChebyshevPolynomial{\opNorm{A}}{\PauliSparse{}} are \BQPcomplete{}. The hardness results hold even under the hypothesis of constant precision $1/\varepsilon = \Omega(1)$, $\norm{A}{1} \leq 2$, {A being 5-local} and for real symmetric $A$
\end{theorem}

\begin{proof}
     First, we show that these two problems are in BQP. By Lemmas~\ref{lemma:approximate_lipschitz_functions} and \ref{lemma:simulation_sparse_pauli} we are able to compute $[T_m\left(A\right)]_{j,j}$ efficiently if we can bound $\norm{T_m}{\infty}^{[-1,1]}$ and its Lipschitz constant $K_{T_m}$ over $[-1,1]$. Clearly $\norm{T_m}{\infty}^{[-1,1]} = 1$, and we can bound $K_{T_m}$ using the Mean Value Theorem as
    \begin{align}\label{eq:lipschitz-chebyshev}
        |T_m(x) - T_m(y)| \leq |T_m'(c)| |x-y| = m |U_{m-1}(c)| |x-y| \leq m (m+1) |x-y| \,, 
    \end{align}
    where we have used the well-known properties detailed in Def.~\ref{def:chebyshev_polynomials}.

    Now we provide the \BQP{}-\textit{hardness} proofs. Consider the matrix $A = \frac{W + W^\dagger}{2}$ from the hardness result of Thm.~\ref{teo:matrixPowerSparseAccessBQPcomplete}, which is 4-sparse{, $5$-local}, Hermitian (real symmetric if we use the Hadamard + Toffoli gate set) and has operator norm bounded by 1. To prove hardness for constant precision it is enough to find a value for $m$ such that $T_m$ is odd and Eq.~\eqref{eq:jantzig_condition2} is satisfied for $k = \Omega(1)$. From now on, we use the notation from the proof of Thm.~\ref{teo:matrixPowerSparseAccessBQPcomplete}.

    Consider $m = M$. Then, $T_m$ is an odd function ($M = 2T+1$ is odd, and the $\ell$-th Chebyshev polynomial is odd if $\ell \in \N$ is odd), and since $T_m\left( \cos \left( \frac{\pi j}{m} \right) \right) = (-1)^j$ for any $j \in \N$ it holds that $T_m(\eigen_\ell^+) = 1$, and consequently

    \begin{equation}
         \frac{1}{M}\left(T_m(1)+2\sum_{l=1}^{\frac{M-1}{2}}T_m(\eigen_\ell^+)\right) = 1\,.
    \end{equation}

    It follows that we can distinguish between acceptance and rejection by picking $\varepsilon = \frac{1}{3}$. The precise Karp mapping we are considering is

    \begin{align*}
        (C=U_{T} \ldots U_1, \ket{\inputVector{x}}) \to \left(-A, m=2T+1,\  \ket{j} = {\step{0}}  \ket{\inputVector{x}} \zero{r-n},\ g = 0,\ \varepsilon = \frac{1}{3}\right)\,.
    \end{align*}


    The matrix constructed in the reduction is the same as the one from Theorem~\ref{teo:matrixPowerSparseAccessBQPcomplete}, which is Pauli sparse. Thus, from the same reasoning, the Pauli access version of the problem is also \BQPhard{}.
\end{proof}

We get equivalent results for the local measurement version of these problems.

\begin{proposition}\label{prop:chebyshev_lm_hard}
    The problems \problemMatrixChebyshevPolynomialLM{\opNorm{A}}{\sparseAccess{}} and \problemMatrixChebyshevPolynomialLM{\opNorm{A}}{\PauliSparse{}} are \BQPcomplete{}, even under the hypothesis of constant precision, $\norm{A}{1} \leq 2$, and for real symmetric {and $5$-local} $A$.
\end{proposition}
    
\begin{proof}
    To show inclusion we use a two stage algorithm in the same spirit of the proof of Proposition \ref{prop:matrixPowerSparseAccessBQPcomplete-lm}. First, we use Lemma \ref{lemma:approximate_lipschitz_functions} to evaluate $\|T_{m}(A)\ket{0}\|^2$ to additive error $\varepsilon/3$. If this value is $\leq g+\varepsilon/2$ then we output \textsc{NO} as we can guarantee that $\bra{0}T_{m}(A)\mathsf{\pi}T_{m}(A)\ket{0} \leq g+5\varepsilon/6$ by H\"older's tracial matrix inequality. If the value is $> g+\varepsilon/2$ then we can guarantee that $\|T_{m}(A)\ket{0}\|^2 \geq g+\varepsilon/6 = \Omega(\varepsilon)$ and we use Lemma \ref{lemma:approximate_lipschitz_functions-lm} to evaluate $\bra{0}T_{m}(A)\mathsf{\pi}T_{m}(A)\ket{0}$ to additive error $\varepsilon/3$, with runtime $\OC(\poly{1/\varepsilon\|T_{m}(A)\ket{0}\|^2, m}) = \OC(\poly{1/\varepsilon^2, m}) = \OC(\polylog{N})$, where we have used the fact that the Lipschitz constant of $T^m$ is $\OC(m^2)$ (see proof of Theorem \ref{teo:chebychev-sparse-alternative-BQPcomplete}). Multiplying the two outputs together gives the desired quantity to additive error $2\varepsilon/3 + \varepsilon^2/3 \leq \varepsilon$ (assuming $\varepsilon\leq 1$, else we can rescale the errors appropriately).


    To show hardness, we use the following property of Chebyshev polynomials:
    
    \begin{equation}\label{eq:chebshev-ballistic}
        T_m\left(\frac{x + x^{-1}}{2} \right) = \frac{x^m + x^{-m}}{2}\,.
    \end{equation}

    This conveys the ``ballistic" property of walks performed by Chebyshev operators (see \cite{apers2024quantumWalksWaveEqn} for further discussion).

    As in the proof of Prop.~\ref{prop:matrixPowerSparseAccessBQPcomplete-lm} for monomials we adopt a walk operator of the form
    \begin{equation}\label{eq:chebyshev-clock-operator}
        A = \frac{1}{2}(W + W^{\dag})\,,
    \end{equation}
    with  $W = \sum_{{\ell}=0}^{M-1} \clockTransition{}_\ell \otimes V_{{\ell}+1}$,
    %
    where we define the $(r+1)$-qubit circuit $V_{M} ... V_1 = U_{1}^{\dag}...U_{T}^{\dag}\, C_{\textsc{NOT}}\, C_{\textsc{NOT}}\, U_{T}...U_1$ where now it will not be necessary to pad the sequence with identities. Recall that $U_{T}...U_1$ is the $r$-qubit circuit we are simulating, and similar to the proof of the monomials we presume that the input state to the $\BQP$ problem $\ket{\vec{x}}$ is encoded in the first gates of the circuit, so we can take $\ket{\vec{0}}$ as input. Recall also that $C_{\textsc{NOT}}$ denotes a CNOT gate targeted on an otherwise untouched ancillary register which we place in the first non-clock register --- this will be the register we measure. Thus, we have $M =2T+2$ clock register states to simulate a $T$-gate circuit $U_{T}...U_1$. It can be checked that $W^\dag=W^{-1}$ and so Eq.~\eqref{eq:chebshev-ballistic} implies that $T_m(\frac{1}{2}(W + W^{\dag})) = \frac{1}{2}(W^m + (W^{\dag})^m)$. Similar to before we explicitly write the orbit of states for successive powers of $W$ or $W^{\dag}$, for any $r$-qubit input state $\ket{\phi}$:

    {
    \begin{align}
        \step{\ell} \otimes V^{({\ell})} \ket{0,\phi} = \begin{cases}
            \step{0} \otimes  \ket{0,\phi} & \text{for}\;\ {\ell}=0  \\
            \step{\ell} \otimes U_{\ell} ...U_1 \ket{0,\phi} & \text{for}\;\ 1\leq {\ell} \leq T \\
            \step{\ell} \otimes C_{\mathrm{NOT}} U_{T} ...U_1 \ket{0,\phi} & \text{for}\;\ {\ell} = T+1 \\
            \step{\ell} \otimes  U_{M-\ell}...U_1 \ket{0,\phi} & \text{for}\;\ T+2 \leq {\ell} \leq M-1\,.
        \end{cases}
    \end{align}
    We can check that 
    \begin{align}
        W^{T+1} \step{0}\ket{0,\phi}= (W^\dag)^{T+1} \step{0}\ket{0,\phi} = \step{\ell} \otimes C_{\mathrm{NOT}} U_{T} ...U_1 \ket{0,\phi} \,.
    \end{align}
    Using Eq.~\eqref{eq:chebshev-ballistic} with the choice $m=T+1$ we have
    \begin{align}
        T_{T+1}(A)\step{0}\ket{0,\phi} &= \frac{1}{2}\left(W^{T+1} + (W^\dag)^{T+1}\right)\step{0}\ket{0,\phi} \\
        &= \step{\ell} \otimes C_{\mathrm{NOT}} U_{T} ...U_1 \ket{0,\phi}\,, \label{eq:chebyshev-state}
    \end{align}
    which is a normalized state for which we know that measurement of the first (non-clock) register yields
    \begin{equation}
            \bra{0,\vec{0}}U^{\dag}_1...U^{\dag}_{T}C_{\mathrm{NOT}} \pi^{(r+1)} C_{\mathrm{NOT}} U_{T} ...U_1 \ket{0,\vec{0}} = \bra{\vec{0}}C \pi^{(r)} C \ket{\vec{0}}   = |\alpha_{\vec{0},1}|^2\,,
    \end{equation}
    where we have denoted $\pi^{(m)} = \ketbra{1}{1}\otimes \identity^{\otimes (m-1)}$ and set $\ket{\phi}=\ket{\inputVector{0}}$. Thus, any BQP problem with $T$ gates can be reduced to the Chebyshev local measurement problem with constant error and $(T+1)$-th Chebyshev polynomial. 
    As Eq.~\eqref{eq:chebyshev-state} is already normalized, this demonstrates hardness both for the normalized and unnormalized problems, assuming efficient access.
    }

    Finally, we can check efficient access: the walk operator $A$ in Eq.~\eqref{eq:chebyshev-clock-operator} is 4-sparse {and 5-local} and lends itself to efficient sparse access. It also has a Pauli decomposition of $\poly{M}$ Pauli operators, with $\OC(M)$ Pauli norm (see Lemmas \ref{lemma:pw-clock-construction}, \ref{lemma:decompositionOfUniversalGates} and \ref{lemma:pw-multiplicative}).
\end{proof}

We cannot prove \BQP{}-\textit{completeness} for neither the case when $\norm{A}{1} \leq 1$ nor $\lambdaMatrix{A} \leq 1$, but still can argue that these problems should not be solvable efficiently classically, since that would imply that $\BPP{} = \BQP{}$:

\newcommand{\thealgorithm}{\mathcal{A}}
\begin{proposition}\label{teo:hardness_chebyshev_1_norm}
    Suppose there is a classical probabilistic algorithm $\thealgorithm$ that;  given sparse (or Pauli) access to a matrix $A \in \C^{N \times N}$ satisfying $\norm{A}{1} \leq 1$ (or $\lambdaMatrix{A} \leq 1$), two indices $i,j \in \until{N}$, an integer $m$, a precision $\varepsilon$ satisfying $m,\frac{1}{\varepsilon}, = \OC(\polylog{N})$ and a number $\delta \in (0,1)$; outputs with probability at least $1-\delta$ an $\varepsilon$-approximation of $\bra{i} T_m \ket{j}$ in time $\OC\left(\poly{\log N, m, \frac{1}{\varepsilon}, }\right)$. If so, then $\BPP{} = \BQP{}$. 

    The result holds even if the algorithm $\mathcal{A}$ can only work on matrices satisfying $\norm{A}{1} = \OC(1/\polylog{N})$ or $\lambdaMatrix{A} = \OC(1/\polylog{N})$.
\end{proposition}

\begin{proof}
    We describe the algorithm for the sparse access case, but the treatment of the Pauli access model is equivalent. Consider the problem of computing, given a matrix $A$ through sparse access with norm condition $\norm{A}{1} \leq 1$, a $\varepsilon$-approximation of $\bra{i} e^{iAm} \ket{j}$, which is \BQPcomplete{} for choice of $m, 1/\varepsilon = \OC(\polylog{N})$  (see Prop.~\ref{prop:bqp_complete_ham_sim}).\footnote{Technically, the decision version of this problem is \BQPcomplete{}} We are going to show how to solve it efficiently classically using $\thealgorithm$.

    To $\varepsilon$-approximate $\bra{i} e^{iAm} \ket{j}$ we can consider a $\frac{\varepsilon}{2}$-approximation of the function $f(x) = e^{ixm}$ given by the Anger-Jacobi expansions (see Lemma~\ref{lemma:anger_jacobi}). It holds that

    \begin{align*}
        \bra{i} e^{iAm} \ket{j} &\approx \bra{i} J_0(m) \identity + 2 \sum_{k=1}^R (-1)^k J_{2k}(m) T_{2k}(A) + 2 i \sum_{k=0}^R (-1)^k J_{2k+1}(m) T_{2k+1}(A) \ket{j}\\
        &= \bra{i}J_0(m) \identity \ket{j} + 2 \sum_{k=1}^R (-1)^k J_{2k}(m) \bra{i} T_{2k}(A) \ket{j} + 2 i \sum_{k=0}^R (-1)^k J_{2k+1}(m) \bra{i} T_{2k+1} (A) \ket{j}\,,
    \end{align*}
    where $R = \OC\left(m + \log \left( \frac{2}{\varepsilon}\right) \right)$ and $J_j$ is the Bessel function of the first kind of order $j$. To approximate this expression within error $\frac{\varepsilon}{2}$ it is enough to approximate each term $J_{k}(m) \bra{i} T_{k}(A) \ket{j}$ with precision $\frac{\varepsilon}{8R}$ and $J_0(m)$ with precision $\frac{\varepsilon}{4}$. 
    
    Note that $|J_k(m)|\leq 1$ in general. In addition, since both $\|A\|\leq\norm{A}{1}$ and $\|A\|\leq\lambdaMatrix{A}$ hold, the proposition's assumptions imply that  $|\bra{i} T_{k}(A) \ket{j}| \leq 1$. Hence, 
    we can $\frac{\varepsilon}{8R}$-approximate $J_k(m) \bra{i} T_{k}(A) \ket{j}$ using a $\frac{\varepsilon}{24R}$ approximation of both $J_k(m)$ and $\bra{i} T_{k}(A) \ket{j}$.\footnote{This follows straightforwardly from error propagation in multiplication.} Thanks to the fact that $R = \OC\left(m + \log\left( \frac{2}{\varepsilon}\right)\right)$, the precision $\frac{\varepsilon}{24R}$ is $\OC\left( \frac{1}{\poly{m \varepsilon}}\right)$ and we can use $\thealgorithm$ to compute in polynomial time the $\frac{\varepsilon}{24R}$-approximation of $\bra{i} T_k(A)\ket{j}$, while for the term $J_k(m)$ we employ folklore techniques. Finally, by picking the probability of success as $\delta = \OC\left(\frac{1}{R}\right)$ and using the Union Bound we ensure that with constant probability all approximations are correct.

    {The second statement of the Proposition follows straightforwardly by observing that $\bra{i} e^{iAm} \ket{j} = \bra{i} e^{i(A/b)bm} \ket{j}$ for any $b \in \R \setminus \{0\}$ and applying the same strategy as before.}
\end{proof}

Thms.~\ref{teo:chebychev-sparse-alternative-BQPcomplete} and \ref{teo:hardness_chebyshev_1_norm} show that the problem of computing matrix Chebyshev polynomials is harder than the problem of computing monomials. In particular, \problemMatrixChebyshevPolynomial{\opNorm{A}}{\sparseAccess{}} remains \BQPcomplete{} under the restriction that $\varepsilon^{-1}=\OC(1)$, and we can argue that the problem cannot be solved classically even when considering $\norm{A}{1} \leq 1$ or $\opNorm{A} \leq 1 -\eta$.\footnote{We did not explicitly prove this fact, but it is a consequence of the proof from Theorem~\ref{teo:hardness_chebyshev_1_norm}.}  
Moreover, Thm.~\ref{teo:hardness_chebyshev_1_norm} shows that under the Pauli representation computing matrix Chebyshev polynomials is hard classically, even when considering that the Pauli norm is bounded by 1. 

Nonetheless, we identify a classically tractable case when the requested Chebyshev polynomial has a sufficiently low degree:

\begin{proposition}\label{prop:cheby_sparse_low_degree}
    The problem \problemMatrixChebyshevPolynomial{\opNorm{A}}{\sparseAccess{}} can be solved exactly in time $\OC(ms^m)$, which is polynomial time whenever $s = \OC(1),m = \OC(\log \log N)$. It can also be solved approximately classically with probability $1-\delta$ in time
    \begin{equation}
        \OC\left(m^3 s\,\frac{4^{2m}}{\varepsilon^2}  \max(1, \norm{A}{1}^{2m}) \log \left( 
    \frac{1}{\delta} \right)\right)\,,
    \end{equation}
    which is polynomial time if $\norm{A}{1} = \OC(1), m = \OC(\log \log N)$. Meanwhile, the problem \problemMatrixChebyshevPolynomial{\opNorm{A}}{\PauliSparse{}} can be solved with probability $1-\delta$ in time 
    \begin{equation}
        \OC\left(m^3 \log(N) \frac{4^{2m}}{\varepsilon^2} \max(1, \lambdaMatrix{A}^{2m}) \log\left( \frac{1}{\delta} \right)\right)\,,
    \end{equation}
    which is polynomial time whenever $\lambdaMatrix{A} = \OC(1),m = \OC(\log \log N)$.
\end{proposition}

\begin{proof}
    For the case of sparse access, using the algorithm from Lemma~\ref{lemma:power-tractable_fixed} we can compute all the monomials $\bra{i} A^k \ket{j}$ for $k=0,\ldots, m$ in time $\OC(ms^m)$. Then, multiplying by the coefficients of $T_m$ and summing up the results takes additional time $\OC(m)$. 

    For the other two algorithms we use Lemmas~\ref{lemma:montanaro_classical} and~\ref{lemma:wang_classical}. Let's compute the value of $\norm{T_m(\lambdaMatrix{A} x)}{\ell_1}$. Since each coefficient of $T_m$ is upper bounded by $4^m$ (Lemma~\ref{lemma:bound_tcheby_coeff}), it follows that

    \begin{equation}
        \norm{T_m(\lambdaMatrix{A}x)}{\ell_1} \leq \sum_{k=0}^m 4^m \lambdaMatrix{A}^k = 4^m \sum_{k=0}^m \lambdaMatrix{A}^k\,,
    \end{equation}

    If $\lambdaMatrix{A} = 1$ then $\norm{T_m(\lambdaMatrix{A}x)}{\ell_1} \leq m 4^m$. Meanwhile, if $\lambdaMatrix{A} < 1$ we get $\norm{T_m(\lambdaMatrix{A}x)}{\ell_1} = \OC(4^m)$. Finally, if $\lambdaMatrix{A} > 1$ we can bound the sum by $m \lambda^{m}$ and we get $\norm{T_m(\lambdaMatrix{A}x)}{\ell_1} \leq m (4\lambdaMatrix{A})^m $.

    


    
\end{proof}

\begin{proposition}\label{prop:cheby_sparse_low_degree_lm}
    The problems \problemMatrixChebyshevPolynomialLM{\opNorm{A}}{\sparseAccess{}} and \problemMatrixChebyshevPolynomialLM{\opNorm{A}}{\PauliAccess{}} can be solved in time complexities polynomially equivalent to those for the entry estimation version of the problem in 
    Prop.~\ref{prop:cheby_sparse_low_degree}.
\end{proposition}

\begin{proof}
    Use Lemma~\ref{lemma:randomized_classical_lm} and the computations from Prop.~\ref{prop:cheby_sparse_low_degree}.
    
\end{proof}

\subsection{Matrix inversion}

We formalize the problems as follows:\newline

\noindent
\problem{
\problemMatrixInversion{b(A)}{\accessModel{}}
}{
A $N\times N$ Hermitian matrix $A$ 
with condition number $\kappa_A$, norm $b(A) \leq 1$, and accessible through \accessModel{}, index $j \in \until{N}$, a precision $\varepsilon$ and a threshold $g$, such that $\kappa_A, 1 / \varepsilon, g = \OC(\polylog{N})$.
}{
No promise
}{
\textsc{YES} if $[A^{-1}]_{j,j} \geq g + \varepsilon$ and \textsc{NO} if $[A^{-1}]_{j,j} \leq g - \varepsilon$.}

\noindent
\problem{
\problemMatrixInversionLM{b(A)}{\accessModel{}}
}{
A $N\times N$ Hermitian matrix $A$ with condition number $\kappa_{A}$, norm $b(A) \leq 1$, and accessible through \accessModel{}, a precision $\varepsilon$ and a threshold $g$, such that $\kappa$, $1/\varepsilon,g = \OC(\polylog{N})$.
}{
No promise
}{
Let $\pi = \ket{0}\bra{0} \otimes \identity_{N/2}$ and $r = 
\brazero{n}A^{-1}\, \pi\, A^{-1}
\zero{n}$. Then, answer \textsc{YES} if $r \geq g + \varepsilon$ and \textsc{NO} if $r \leq g - \varepsilon$.
}

We note that our formalization for the inverse function does not coincide with the one from the HHL paper \cite{harrow2009QLinSysSolver} since they consider the normalized version (i.e. computing $\bra{0} A^{-1}\, \pi\, A^{-1} \ket{0} 
 / \opNorm{A^{-1} \ket{0}}$).

We prove the hardness of \problemMatrixInversion{\norm{A}{1}}{\accessModel{}} for both access models employing again the construction from Theorem~\ref{teo:matrixPowerSparseAccessBQPcomplete} and Eq.~\eqref{eq:jantzig_condition2}. In particular, we obtain the hardness result even for constant precision and under the condition that $\norm{A}{1} \leq  1/\polylog{N}
$ for the sparse access model and the equivalent one $\lambdaMatrix{A} \leq 1/\polylog{N}
$ for the Pauli model.

\begin{theorem}\label{teo:inversion_bqp_complete}
    The problems \problemMatrixInversion{\norm{A}{1}}{\sparseAccess{}} and \problemMatrixInversion{\lambdaMatrix{A}}{\PauliSparse{}} are \BQPcomplete{}. Both hardness results hold even under the hypothesis of constant precision, and for real symmetric $A$. The results hold also under the more stringent norm conditions $\norm{A}{1} \leq 1/\polylog{N}$ and $\lambdaMatrix{A} \leq 1/\polylog{N}$ respectively. 
\end{theorem}

\begin{proof}[Proof sketch]
    For both 
    BQP membership, use the algorithm from Lemma~\ref{lemma:approximate_lipschitz_functions} considering the function
    \begin{align}
        f_\kappa(x) = \begin{cases}
            \frac{1}{x} &\text{for } x \in [-1,1] \setminus [-\frac{1}{\kappa}, \frac{1}{\kappa}]\,, \\
            0  &\text{for } x \in (-\frac{1}{\kappa}, \frac{1}{\kappa})\,,
        \end{cases}
    \end{align}
    that satisfies $K_f \leq \kappa^2$ and $||f||_\infty^{[-1,-1]} \leq \kappa$. 

    The \BQPhard{} proof for \problemMatrixInversion{\norm{A}{1}}{\sparseAccess{}} can be obtained again in a similar fashion to the proof of BQP-hardness of monomials. More precisely, one uses 
    Eq.~\eqref{eq:value_of_diagonal_entry} 
    considering the matrix $\frac{A}{2}$, for $A$ given by Eq. \eqref{eq:A}, which has the same eigenvalues  but divided by two. The proof for the hardness of \problemMatrixInversion{\lambdaMatrix{A}}{\PauliSparse{}} is identical but dividing $A$ by $\lambdaMatrix{A}$. We leave the details of proving hardness to Theorem~\ref{teo:matrix_inversion_bqp_hard_full_proof} in the Appendix.   
\end{proof}

We can prove an analogous result for the local measurement versions.

\begin{proposition}\label{prop:inverse-BQP-complete-lm}
    The problems \problemMatrixInversionLM{\norm{A}{1}}{\sparseAccess{}} and \problemMatrixPowerLM{\lambdaMatrix{A}}{\PauliSparse{}} are \BQPcomplete{}. Both hardness results hold under the hypothesis of constant precision, and for real symmetric $A$.
\end{proposition}

\begin{proof}
    Regarding inclusion in \BQP{}, the normalized version is well known to be in \BQP{} for both access models due to the algorithm from \cite{harrow2009QLinSysSolver} which employs a Hamiltonian simulation oracle we can instantiate efficiently. The unnormalized version can also be solved employing that algorithm, since we can compute $\frac{\bra{0}A^{-1} \pi A^{-1} \ket{0}}{\opNorm{A^{-1}} \ket{0}}$  with precision $\frac{\varepsilon}{\opNorm{A^{-1}\ket{0}}^2}$ and then multiply this value by $\opNorm{A^{-1}\ket{0}}^2$ (which we can compute using the algorithms developed for matrix power). This will work in polynomial time because $\opNorm{A^{-1}\ket{0}} = \OC(\kappa) = \OC(\polylog{N})$.

    To prove hardness, we employ the construction from~\cite{harrow2009QLinSysSolver} with some small tweaks to ensure the hypothesis condition of real symmetric matrices. We leave the details to the Appendix, Theorem~\ref{teo:matrix_inversion_lm_bqp_hard_full_proof}, where we recount the ideas of the proof from \cite{harrow2009QLinSysSolver} for completeness.
\end{proof}

The previous results also imply the \BQP{}-\textit{completeness} of the problem \problemMatrixInversion{\opNorm{A}}{\sparseAccess{}}: the \BQP{}-\textit{hardness} follows from a direct reduction from \problemMatrixInversion{\norm{A}{1}}{\sparseAccess{}}, while the \BQP{} algorithm is the same one based on Lemma~\ref{lemma:approximate_lipschitz_functions}.

Theorem~\ref{teo:inversion_bqp_complete} establishes that the problem \problemMatrixInversion{\norm{A}{1}}{\sparseAccess{}} is \BQPhard{} even for a constant precision level. Moreover, by inspecting the proof it is clear that it will remain hard if we ask $\opNorm{A} \leq 1 - \eta$ for some fixed $\eta$, as we did in Theorem~\ref{teo:classically_easy_bounded_precision_and_norm}. Nonetheless, we can obtain a classically solvable restriction if we upper bound $\kappa$:

\begin{theorem}[Classical algorithm for matrix inversion with suppressed condition number]\label{teo:matrix_inversion_fixed_kappa_sparse}
    Let $\beta(\kappaMatrix{A}, \varepsilon) = 2 \kappaMatrix{A} \log \left( \frac{\kappaMatrix{A}^2}{\varepsilon}\right)$. Then, the problem \problemMatrixInversion{\opNorm{A}}{\sparseAccess{}} can be solved exactly classically in time $\OC\left(\beta(\kappaMatrix{A}, \varepsilon) s^{\beta(\kappaMatrix{A}, \varepsilon)}\right)$, which is $\polylog{N}$ whenever $\kappaMatrix{A}, s = \OC(1)$. It can also be solved approximately classically with probability at least $1-\delta$ in time 
    $$\OC\left(\beta(\kappaMatrix{A}, \varepsilon)^3 s \frac{2^{2 \beta(\kappaMatrix{A}, \varepsilon)}}{\varepsilon^2} \max(1, \norm{A}{1}^{4\beta(\kappaMatrix{A}, \varepsilon)}) \log \left( 
    \frac{1}{\delta} \right) \right),
    $$
     which is $\polylog{N}$ if $\norm{A}{1} \leq 1, \kappaMatrix{A} = \OC(1)$.
    
    Meanwhile \problemMatrixInversion{\opNorm{A}}{\PauliSparse{}} can be solved with probability at least $1-\delta$ in time 
    $$
    \OC\left(\beta(\kappaMatrix{A}, \varepsilon)^3 \log(N) \frac{2^{2\beta(\kappaMatrix{A}, \varepsilon)}}{\varepsilon^2} \max(1, \lambdaMatrix{A}^{4\beta(\kappaMatrix{A}, \varepsilon)}) \log \left(\frac{1}{\delta} \right)\right),
    $$
    which is $\polylog{N}$ whenever $\kappaMatrix{A}, \lambdaMatrix{A} = \OC(1)$.

    Finally, Problems \problemMatrixInversionLM{\opNorm{A}}{\sparseAccess{}} and \problemMatrixInversionLM{\lambdaMatrix{A}}{\PauliSparse{}} can be solved in times polynomially equivalent to the two time complexities just mentioned above.
\end{theorem}

\begin{proof}
    We begin with the proof of the statements about the matrix element problem. Given $\varepsilon$ is it possible to efficiently build an $\varepsilon$-approximation $P(x)$ of the function $\frac{1}{x}$ in the range $[-1,-\frac{1}{\kappaMatrix{A}}] \cup [\frac{1}{\kappaMatrix{A}}, 1]$ with degree upper-bounded by $2 \kappaMatrix{A} \log \left( \frac{\kappaMatrix{A}^2}{\varepsilon} \right) = \beta(\kappaMatrix{A}, \varepsilon)$ using Lemma~\ref{lemma:approximate_inverse_function}. Then, for the sparse access model, we can compute $P(A)$ exactly using the algorithm from Lemma~\ref{lemma:power-tractable_fixed} for each monomial. The final complexity is $\OC(\beta(\kappaMatrix{A}, \varepsilon) s^{2\beta(\kappaMatrix{A}, \varepsilon)})$. 
    
    Regarding the two probabilistic algorithms, we can use Lemmas~\ref{lemma:montanaro_classical} and~\ref{lemma:wang_classical}. Let us bound the value $\norm{P(\lambdaMatrix{A} x)}{\ell_1}$: each coefficient from $P(x)$ is upper bounded by $2^{\beta(\kappaMatrix{A},\varepsilon)}$, since they are given by the binomial expressions  $\binom{\beta(\kappaMatrix{A}, \varepsilon)}{i}$, and thus
    \begin{align}
        \norm{P(\lambdaMatrix{A} x)}{\ell_1} \leq \sum_{k=0}^{2\beta(\kappaMatrix{A},\varepsilon)} 2^{\beta(\kappaMatrix{A},\varepsilon)} \lambdaMatrix{A}^k \leq 2\beta(\kappaMatrix{A},\varepsilon) 2^{\beta(\kappaMatrix{A},\varepsilon)} \max(1, \lambdaMatrix{A}^{2\beta(\kappaMatrix{A}, \varepsilon)})\,,
    \end{align}
    where the last inequality follows through a reasoning analogous to the one from Prop.~\ref{teo:hardness_chebyshev_1_norm}. 
    
    Finally, the algorithms for the local measurement version of the problem are obtained using Lemma~\ref{lemma:randomized_classical_lm}.
\end{proof}


        


\subsection{Time evolution}

In this section, we show our results for Problems~\ref{def:general_problem} and \ref{def:general_problem-lm} when $f_t(x) = e^{-itx}$. For this function, the former problem {can be used to define a promise problem} as follows: \newline

\noindent
\problem{
\problemHamiltonianSimulation{b(A)}{\accessModel{}}
}{
A $N \times N$ Hermitian matrix $A$ with norm $b(A) \leq 1$ and accessible through \accessModel{}, a positive real number $t$, integers $j, k \in \until{N}$, a precision $\varepsilon$ and a threshold $g$, such that $t, 1 / \varepsilon, g = \OC(\polylog{N})$.
}{
$b$ is an upper bound on $\opNorm{A}$, $t$ and $\frac{1}{\varepsilon}$ are $\polylog{N}$.
}{
\textsc{YES} if $\big|[e^{-itA}]_{j,k}\big| \geq g + \varepsilon$. \textsc{NO} if $\big|[e^{-itA}]_{j,k}\big| \leq g - \varepsilon$.
}

Similarly, for Problem \ref{def:general_problem-lm} we define:\newline

\noindent
\problem{
\problemHamiltonianSimulationLM{b(A)}{\accessModel{}}
}{
A $N \times N$ Hermitian matrix $A$ with norm $b(A) \leq 1$ and accessible through \accessModel{}, a positive real number $t$,  a precision $\varepsilon$ and a threshold $g$, such that $t, 1 / \varepsilon, g = \OC(\polylog{N})$.
}{
$b$ is an upper bound on $\opNorm{A}$, $t$ and $\frac{1}{\varepsilon}$ are $\polylog{N}$.
}{
Let $\pi = \ket{0}\bra{0} \otimes \identity_{N/2}$ and $r = \bra{0} e^{itA} \pi e^{-itA}\ket{0}$. Then, answer \textsc{YES} if $r \geq g + \varepsilon$ and \textsc{NO} if $r \leq g - \varepsilon$.
}

We start by proving \BQP{}-completeness of the problem under sparse access for the entry estimation problem. Our proof of \BQP{}-hardness is based on a slight modification of the circuit-to-Hamiltonian mapping proposed by Peres in~\cite{peres1985reversibleLogic}, one of the first clock constructions for circuit simulation. 

\begin{proposition}\label{prop:bqp_complete_ham_sim}
    The problems \problemHamiltonianSimulation{\norm{A}{1}}{\sparseAccess{}} and \problemHamiltonianSimulation{\lambdaMatrix{A}}{\PauliSparse{}} are \BQP{}- \textsc{complete} for constant error $\varepsilon$ and constant choice of row sparsity $(s=8)$.
\end{proposition}

\begin{proof}
    First, we argue that the problem is in \BQP{} and then show it is \BQPhard{}. For inclusion, according to Lemma \ref{lemma:simulation_sparse_pauli}, efficient algorithms to implement the Hamiltonian simulation operator $e^{-iAt}$ exist in both access models as long as $t,\, s,\, 1/\epsilon=\polylog{N}$ for constant $\|A\|
    $ (\sparseAccess{}) or constant $\lambdaMatrix{A}$ (\PauliAccess). This settles inclusion for \problemHamiltonianSimulation{\lambdaMatrix{A}}{\PauliSparse{}}, and inclusion for \problemHamiltonianSimulation{\norm{A}{1}}{\sparseAccess{}} follows by noting that $\|A\|
    \leq \|A\|_1$. Given an implementation of $e^{-iAt}$, the entry $\bra{j}e^{-iAt}\ket{i}$ can be encoded in the state of $2n+1$ qubits alike to the definition of \BQP{} (Eq.\ \eqref{eq:def_BQP}) with $\alpha_{\boldsymbol{x},0}=\frac{1}{\sqrt{2}}\sqrt{1+|\bra{j}e^{-iAt}\ket{i}|^2}$ and  $\ket{\psi_{\boldsymbol{x},0}}=\left[\ket{j}\otimes (e^{-iAt}\ket{i})+(e^{-iAt}\ket{i})\otimes \ket{j}\right]/\alpha_{\boldsymbol{x},0}$. This state is prepared by using a SWAP test  \cite{buhrman2001QuantumFingerprinting} between the states $\ket{j}$ and $e^{-iAt} \ket{i}$ in a $\poly{n}$ sized circuit. The spectral error $\varepsilon$ in implementing $e^{-iAt}$ directly translates into an equal additive error in the matrix entry to be estimated, which defines the gap of the \BQP{} problem.  
    
    For hardness, consider a circuit $U = U_T \ldots U_1$ on $r$ qubits and a bitstring $\inputVector{x}$ as the inputs to \problemBQPCircuitSimulation (Observation \ref{obs:circuit_simulation_bqp_complete}). 
    %
    From $C$ we can build a new circuit acting on $r'=r+1$ qubits with corresponding unitary transformation given as 
    $C'=(\identity\otimes C^\dagger               
  )\,(\text{CNOT}_{21}\otimes \identity^{\otimes r-1})\,(\identity\otimes C):=V_\tau \cdots V_1$, with $\tau=2T+1=\polylog{N}$. 
  This circuit is the result of acting $C$ on the original qubits, copying the computation result on the original first qubit into the additional ancilla qubit with a $\text{CNOT}$ gate, followed by uncomputation with $C^\dagger$. One can directly verify that 
  \begin{align}\label{eq:TE_overlap}
      \bra{0}\bra{\inputVector{x}}\bra{0^{r-n}}C'\ket{0}\ket{\inputVector{x}}\ket{0^{r-n}}=\alpha_{\inputVector{x},0}\quad\text{and}\quad \bra{1}\bra{\inputVector{x}}\bra{0^{r-n}}C'\ket{0}\ket{\inputVector{x}}\ket{0^{r-n}}=\alpha_{\inputVector{x},1}.
  \end{align}
  In the rest of the proof, we denote $\ket{\inputVector{x}'}=\ket{0}\ket{\inputVector{x}}$, that is the bit string $\inputVector{x}$ padded with a zero bit.

    {
        Consider the Hamiltonian working on a unary clock of dimension $\tau$ alongside $x'$ defined as
        \begin{align}\label{eq:clock-hamiltonian}
            A = \frac{1}{4\tau} \sum_{j=1}^{\tau} \sqrt{j(\tau+1-j)} \left( \clockTransition{}_j \otimes V_j + \clockTransition{}_j^\dagger \otimes V_j^\dagger \right)\,,
        \end{align}
    which has row sparsity 8 and is 5-local. The 1-norm of $A$ can be bounded as 
    \begin{align}
    \norm{A}{1} =\max_{j}\frac{\sqrt{j(\tau+1-j)}}{4 \tau}2\norm{V_j}{1} \leq \max_j \frac{\sqrt{j(\tau+1-j)}}{\tau}\leq\frac{1+\tau}{2\tau}\leq 1,
    \end{align}
    where we used the fact that $\norm{V_j}{1}\leq2$ for a 2-qubit unitary.
     Let us consider that the initial state of the evolution is $\step{0} \ket{x'}\ket{0}^{\otimes r-n}$. Given an ansatz for the time evolved state as $\ket{\psi(t)}=e^{-itA}\step{0} \ket{\inputVector{x}'} \zero{r-n}=c_0(t) \step{0}\ket{\inputVector{x}'} \zero{r-n}+\sum_{j=1}^\tau  c_j(t)\,\step{j} \otimes  V_j\cdots V_1 \ket{j} \ket{\inputVector{x}'} \zero{r-n}$, the corresponding Schrodinger equation can be solved to find the time-dependent coefficients. In particular, $c_0(t)=\left(\cos\frac{ t}{4 \tau}\right)^\tau$ while $c_\tau(t)=\left(i \sin\frac{ t}{4 \tau}\right)^\tau$. Therefore, for $t=2\pi\tau$ (and consequently $t=\poly{n}$) 
    it holds that 
    \begin{equation}\label{eq:ham_sim_bqp_comp}
        e^{-iA 2\pi\tau } \step{0} \ket{\inputVector{x}'} \zero{r-n} = i^\tau \step{\tau}  C' \ket{\inputVector{x}'} \zero{r-n}.
    \end{equation}
    Therefore, with $f_{2\pi \tau}(A) = e^{-iA2\pi \tau}$, we have 
     \begin{equation}
         \brastep{\tau}\bra{1}\bra{\inputVector{x}} \brazero{r-n} f_{2 \pi \tau}(A) \step{0} \ket{0} \ket{\inputVector{x}} \zero{r-n} = \brastep{\tau}\bra{1}\bra{\inputVector{x}} \brazero{r-n}  (i^\tau \step{\tau}  C' \ket{0}\ket{\inputVector{x}} \zero{r-n} )= i^\tau \alpha_{\inputVector{x},1},
     \end{equation}
where we used Eq.\eqref{eq:TE_overlap}. 

The Karp mapping follows with $\ket{j} = \step{0}  \ket{0} \ket{\inputVector{x}} \zero{r-n}$, $\ket{k} = \step{\tau}\ket{1}  \ket{\inputVector{x}} \zero{r-n}$, $g=\sfrac{1}{2}$, and $\varepsilon=\sfrac{1}{12}$.
    }
\end{proof}
\bigskip

The result for the local measurement problem follows analogously.

\begin{proposition}\label{prop:bqp_complete_ham_sim_local}
    The problems \problemHamiltonianSimulationLM{\norm{A}{1}}{\sparseAccess{}} and \problemHamiltonianSimulationLM{\lambdaMatrix{A}}{\PauliSparse{}} are $\BQPcomplete{}$ for constant error $\varepsilon$ and constant choice of row sparsity $(s=8)$.
\end{proposition}

\begin{proof}
    The proof of inclusion for the entry estimation problem holds for local measurement, with the difference that the local observable can be estimated directly from the state transformed by $e^{-iAt}$ without needing the SWAP test.

    The proof of hardness also follows straightforwardly from the construction for entry estimation, as can be seen from Eq.\ \eqref{eq:ham_sim_bqp_comp} with the local measurement being performed in the middle register.
\end{proof}

Further, we note that the norm condition can be strengthened.

\begin{proposition}\label{prop:bqp_complete_ham_sim_small-norm}
    The problems \problemHamiltonianSimulation{\norm{A}{1}}{\sparseAccess{}}, \problemHamiltonianSimulationLM{\norm{A}{1}}{\sparseAccess{}}, \\ \problemHamiltonianSimulation{\lambda_{A}}{\PauliSparse{}}, \problemHamiltonianSimulationLM{\lambda_{A}}{\PauliSparse{}} are $\BQPcomplete{}$ for constant error $\varepsilon$, a constant choice of row sparsity and $\|A\|_1, \lambdaMatrix{A} \leq 1/\polylog{N}$, respectively.    
\end{proposition}

\begin{proof}
    Inclusion follows from Lemma \ref{lemma:simulation_sparse_pauli} as before, as such bounded-norm matrices fall within the parameters of Lemma \ref{lemma:simulation_sparse_pauli}. To show hardness, we pick any $c= \Theta(\polylog{N})$. We normalize the Hamiltonian in Eq.~\eqref{eq:clock-hamiltonian} as $A'= A/ c$, and we see that this matrix satisfies the norm condition we impose for our problem class. Moreover, simulation of $A'$ for time $ c 2\pi(2T+1)$ simulates a circuit with $T$ gates, where $c 2\pi(2T+1)= \OC(\polylog{N})$ for $T= \OC(\polylog{N})$.
\end{proof}

We now move on to an efficient classical algorithm {for a general class of matrices}. We remark that the normalized and unnormalized versions of the local measurement problem are equivalent for time evolution since the function is unitary and, therefore, does not change the state's norm. 

\begin{proposition}\label{prop:ham-sim-classical}
    Problems \ref{def:general_problem} and \ref{def:general_problem-lm} for $e^{iAt}$ are classically easy with Hermitian $A\in \C^{N \times N}$, $\|A\|\leq 1$, for $\|A\|_1t = \OC(\log \log N)$ in the sparse access model and $\lambdaMatrix{A}\,t = \OC(\log \log N)$ in the Pauli access model. 
\end{proposition}

\begin{proof}
    We follow the quantum algorithm of \cite{berry2014HamSimTaylor}, but simulate it classically using the algorithms of Lemmas \ref{lemma:montanaro_classical} and \ref{lemma:wang_classical} 
    Let us denote $\hat{t}=\gamma t$, where $\gamma = \|A\|_1$ for the sparse access model and $\gamma = \lambdaMatrix{A}$ for the Pauli access model.
    Consider a fragmentation of the time evolution operator $e^{iAt}= (e^{iAt/r})^r$. We approximate each fragment of the time evolution via a truncated Taylor series $e^{iAt/r} \approx \sum_{k=0}^K\frac{(iAt/r)^k}{k!}$. The truncation error satisfies
    \begin{align}
        \left\| e^{iAt} - \Big(\sum_{i=0}^K\frac{(iAt/r)^k}{k!} \Big)^r \right\| &\leq r \left\| e^{iAt/r} - \sum_{k=0}^K \frac{(iAt/r)^k}{k!} \right\|\\
        &\leq \varepsilon\,,
    \end{align}
    where the first inequality is due to a series of triangle inequalities, and the second inequality is true for choice of $r={t}/\ln 2$, $K=\OC(\frac{\log {t}/\varepsilon}{\log \log {t}/\varepsilon})$. We note for later that the size of the sum $\sum_{k=0}^K \frac{(\hat{t}/r)^k}{k!}$ is upper bounded by $e^{\hat{t}/r}$, and thus the size of the sum $\left(\sum_{k=0}^K \frac{(\hat{t}/r)^k}{k!}\right)^r$ is upper bounded by $e^{\hat{t}}$. 

    We can approximate a matrix entry with a probabilistic distribution over matrix entries of matrix powers
    \begin{align}
        \bra{j} e^{iAt} \ket{m} \stackrel{\varepsilon}{\approx}  \bra{j} \left(\sum_{k=0}^K \frac{1}{k!} \Big(\frac{i\hat{t}}{r}\Big)^k \Big(\frac{A}{\gamma} \Big)^k \right)^r \ket{m} = \sum_{k_1,...,k_r=0}^K \frac{1}{k_1!} \cdots \frac{1}{k_r!} \Big(\frac{i\hat{t}}{r}\Big)^{k_1 + \cdots +k_r}  \bra{j} \Big(  \frac{A}{\gamma} \Big)^{k_1 + \cdots +k_r} \ket{m}\,,
    \end{align}
    where we denote $\stackrel{\varepsilon}{\approx}$ as an additive approximation to error $\varepsilon$. We recall the quantity on the right hand side may be approximated to additive error $\varepsilon$ with success probability at least $1-\delta$ with cost $K\frac{\alpha^2}{\varepsilon^2}\log(\frac{1}{\delta})$ using Lemmas \ref{lemma:montanaro_classical} and \ref{lemma:wang_classical}, where 
    \begin{equation}
        \alpha = \sum_{k_1,...,k_r=0}^K \frac{1}{k_1!} \cdots \frac{1}{k_r!} \Big(\frac{\hat{t}}{r}\Big)^{k_1 + \cdots +k_r} =\sum_{k=0}^K \frac{1}{k!}\Big(\frac{\hat{t}}{r}\Big)^r \leq e^{\hat{t}}\,.
    \end{equation}
    Thus, we have an efficient algorithm when $\hat{t} = \OC(\log \log(N))$.

    For the local measurement problem we have 
    \begin{align}
        \bra{j} &e^{iAt} \pi e^{iAt} \ket{j} \stackrel{2\varepsilon}{\approx} \\
        &= \bra{j}\sum_{k_1,...,k_r=0}^K \frac{1}{k_1!} \cdots \frac{1}{k_r!} \Big(\frac{i\hat{t}}{r}\Big)^{k_1 + \cdots +k_r} \sum_{k'_1,...,k'_r=0}^K \frac{1}{k'_1!} \cdots \frac{1}{k'_r!} \Big(\frac{i\hat{t}}{r}\Big)^{k'_1 + \cdots +k'_r}  \bra{j} \Big(  \frac{A}{\gamma} \Big)^{k_1 + \cdots +k_r} \pi \Big(  \frac{A}{\gamma} \Big)^{k'_1 + \cdots +k'_r} \ket{j}\,,
    \end{align}
    where similar to before the quantity on the right hand side may be approximated to additive error $\varepsilon$ with success probability at least $1-\delta$ with cost $K\frac{\alpha^4}{\varepsilon^2}\log(\frac{1}{\delta})$ using Lemma \ref{lemma:randomized_classical_lm}.
\end{proof}

{
Finally, we provide an algorithm for Hamiltonian Simulation for $\OC(1)$-sparse matrices $A$ satisfying $\opNorm{A} t \leq 1$.

\begin{proposition}[Constant time evolution]\label{prop:constant_time_evolution}

    Problems \ref{def:general_problem} and \ref{def:general_problem-lm} for $e^{iAt}$ are classically easy with $\OC(1)$-sparse Hermitian $A\in \C^{N \times N}$ satisfying $\opNorm{A} t = \OC(\log \log N)$ in the sparse access model.
    
\end{proposition}

\begin{proof}
    We may assume $\opNorm{A} = 1$ and $t = \OC(\log \log N)$ by working on $A / \opNorm{A}$ and evolving it for time $\opNorm{A} t$, which is $\OC(\log \log N)$ by hypothesis.

    Given any precision $\varepsilon$, using the Anger-Jacobi expansion (Lemma~\ref{lemma:anger_jacobi}) we may approximate $e^{ixt}$ up to precision $\varepsilon$ with a polynomial over Chebyshev polynomials of degree $m = \OC(t + \log(1/\varepsilon))$. Thus, we can solve both problems by computing these Chebyshev polynomials using the first algorithm from Prop.~\ref{prop:cheby_sparse_low_degree}, which has complexity $\OC(ms^m) = ((t + \log (1/\varepsilon)) s^{\OC(t + \log (1/\varepsilon))})$. Since $1/\varepsilon = \polylog{N}$ and $t = \OC(\log \log N)$ it holds that $\OC(t + \log (1 / \varepsilon)) = \OC(\log \log N)$ and thus the resulting algorithm is polylogarithmic on the dimension. 
\end{proof}
}

\subsection{Classical eigenvalue transform}

In this section, we present our classical algorithms for general classes of polynomials. We recall that hardness for general polynomials holds for both Problems \ref{def:general_problem} and \ref{def:general_problem-lm} even for constant precision (e.g., consider Chebyshev polynomials in Thm.~\ref{teo:chebychev-sparse-alternative-BQPcomplete} and Prop.~\ref{prop:chebyshev_lm_hard}). Thus, we should not expect efficient classical algorithms for too generic a class of matrices -- even $\OC(1)$-sparse matrices in sparse access. The algorithms which we elucidate here allow efficient processing of large matrices if they are very sparse (by direct application of Thm.~\ref{teo:pauli_power_super_sparse} and Prop.~\ref{prop:sparse_power_super_sparse}), or for much milder conditions on the sparsity in Pauli access only if they have an inverse-polynomial-sized norm (by combining the above ideas with importance sampling).

\begin{observation}\label{obs:pauli_power_super_sparse}
    Thm.~\ref{teo:pauli_power_super_sparse} allows exact solution to Problems \ref{def:general_problem} and \ref{def:general_problem-lm} in Pauli access to the matrix $A$ for any degree-$d$ polynomial in time $\OC(d^2\, L\, 2^L \log N)$ and $\OC((d^2\,L\,2^L + 2^{2L}) \log N)$, respectively, where we recall $L$ is the number of coefficients of $A$ in the Pauli basis. Prop.~\ref{prop:sparse_power_super_sparse} allows an exact solution to both Problem \ref{def:general_problem} and \ref{def:general_problem-lm} in $\OC(d^2 k^3)$ time when the $k$ non-zero entries of $A$ (in the computational basis) are given as a list. The additional factor of $d$ is due to the fact that a general degree-$d$ polynomial consists of $\OC(d)$ monomials.
\end{observation} 

\begin{theorem}[Super-sparse classical matrix processing]\label{teo:function-super-sparse}
    Consider a matrix $A$ satisfying $\|A\|\leq 1$. Consider a function $f(x)$ which is approximated as $|f(x) -g(x)|_{[-1,1]}\leq \varepsilon$, where $g(x)$ is a polynomial of degree $d_{f, \varepsilon}$ computed in time $t_{f, \varepsilon}$. The entry estimation problem can be solved classically for $f(A)$ 
    in time $\OC(t_{f, \varepsilon} + d^2_{f, \varepsilon}\cdot 2^L\log{N})$ 
    in the Pauli access model where $L$ denotes the number of Pauli terms; or in time $\OC(t_{f,\varepsilon} + d^2_{f, \varepsilon}\cdot k^3)$ where $k$ denotes the number of non-zero computational basis entries. The local measurement problem can be solved within polynomially-equivalent runtimes.
\end{theorem}

\begin{proof}
   Directly use the algorithm of Theorem \ref{teo:pauli_power_super_sparse} in the Pauli basis and Proposition \ref{prop:sparse_power_super_sparse} in the computational basis. The dependence on $t_{f, \varepsilon}$ comes simply as a one-off pre-processing step.
\end{proof}

Using Theorem \ref{teo:function-super-sparse} we directly have efficient algorithms for the inverse and time-evolution (complex exponential) functions whenever $L=\OC(\log\log N)$ in the Pauli model or $k = \OC(\polylog{ N})$ in the computational basis, whenever the condition number or evolution time is $\OC(\polylog{N})$. This can be seen from the fact that both functions have efficient polynomial approximations (Lemmas \ref{lemma:approximate_inverse_function} and \ref{lemma:anger_jacobi}).

We now move onto our second classical algorithm for general polynomials. Here, we combine the above algorithm with an importance-sampling sketch. This allows for matrix processing for generic sparsity, so long as one can efficiently sample from the Pauli coefficients. 

\begin{theorem}[Classical matrix processing with suppressed norm]\label{teo:pauli-sketch}
    Instantiate Problems \ref{def:general_problem} and \ref{def:general_problem-lm} with a degree-$d$ polynomial $p_d$ which is bounded as $|p_d(x)|_{[-1,1]}\leq 1$. Suppose that we can sample from the Pauli coefficients of $A=\sum_{\ell} a_{\ell} P_{\ell}$ such that, with probability $|a_{\ell}|/\lambdaMatrix{A}$, the triple $(|a_{\ell}|, a_{\ell}, \ell)$ is returned and $\lambdaMatrix{A}=\sum_{\ell} |a_{\ell}|$ is known. Then, both problems can be solved efficiently to inverse-polynomial failure probability and constant error if the condition
    \begin{equation}\label{eq:norm-surpression}
         d^2\lambdaMatrix{A}  \log\mathrm{rank}(A) \sqrt{\log (N)} = \OC(1)\,,
    \end{equation}
    is satisfied, such that $\lambdaMatrix{A}\leq 1-\eta$ for some $\eta = \Omega(1)$.
    In particular, when the polynomial degree satisfies $d= \OC(\polylog{N})$, this implies that there is an efficient algorithm starting from Pauli access for $A$ with some value of $\lambdaMatrix{A} = \OC(1/\polylog{N})$. 
\end{theorem}

\begin{proof}
    The algorithm follows by two steps: (1) perform a importance-sampling sketch on the Pauli coefficients; (2) use our algorithm for super-sparse Pauli matrix processing on the sketched matrix (Theorem \ref{teo:pauli_power_super_sparse}). 
    When we sample according to the distribution $\{|a_{\ell}|/\lambda_A\}_{\ell}$, each time upon obtaining index $\ell$ we output the (matrix-valued) random variable $X_{\ell} = a_{\ell}\lambda_A P_{\ell}/|a_{\ell}|$. This is an unbiased estimator for $A$.
    Sampling $L'=\frac{8\lambdaMatrix{A}^2}{\varepsilon'^2}\log(\frac{2N}{\delta})$ times, we obtain a Pauli representation of a matrix $A^{(L')}$ which satisfies an operator norm approximation $\big\| A^{(L')} - A\big\| \leq \varepsilon'$. This statement can be gleaned from operator Bernstein inequalities \cite[Thm.~6]{gross2011recovering} (see Lemmas \ref{lem:operator-bernstein}, \ref{lem:pauli-importance-sampling} in the Appendix). We emphasize that here we don't need to perform any matrix arithmetic in the computational basis; we will only need to keep track of the Pauli coefficients of $A^{(L')}$, which is efficient in $L'$. Moreover, we note that the operator norm condition allows us to write
    \begin{equation}
        \|A^{(L')}\| \leq \|A\| + \|A^{(L')} -A\| \leq 1- \eta + \varepsilon'\,,
    \end{equation}
    where we have used the triangle inequality and fact that $\|A\|\leq \lambda_A \leq 1-\eta$. Hereon we ensure that $\varepsilon' \leq \eta$ so that $\|A^{(L')}\|\leq 1$.
    
    Now we use Theorem \ref{teo:pauli_power_super_sparse} on our representation of the matrix $A^{(L')}$ to obtain an exact Pauli representation of a degree-$d$ polynomial $p_d(A^{(L')})$ to approximate both problems. We note that an operator norm approximation of a general function $f$ can be specified from known operator Lipschitz bounds \cite[Thm.~11.2]{aleksandrov2011EstimatesOfOperatorModuli} as 
    \begin{equation}
    \left\| f(A^{(L')}) - f(A)\right\| \leq C L_f \left\| A^{(L')} - A\right\| \cdot \log(\mathrm{min}\{\mathrm{rank}(A),\mathrm{rank}(A^{(L')}) \}) \leq  C L_f \varepsilon' \cdot \log\mathrm{rank}(A)\,,
    \end{equation}
    for some numerical constant $C$, where $L_f$ is the Lipschitz constant of $f$ on the eigenvalues of $A^{(L')}$ and $A$, which both lie within $[-1,1]$.
    Thus, for a choice of $\widetilde{L}$, given computation of $p_d(A^{(\widetilde{L})})$ from the Pauli coefficients of $A^{(\widetilde{L})}$ we have operator norm approximation 
    \begin{equation}
    \left\| p_d(A^{(\widetilde{L})}) - p_d(A) \right\| \leq \varepsilon\,,
    \end{equation}
    with probability at least $(1-\delta)$, where $A^{(\widetilde{L})}$ is constructed from $\widetilde{L} = \OC\left(\frac{\lambda_A^2 d^4}{\varepsilon^2} \log^2\mathrm{rank}(A) \log(\frac{2N}{\delta}) \right)$ Pauli terms. Here we have used the fact that bounded degree-$d$ polynomials on $[-1,1]$ have Lipschitz constant $d^2$ (see Lemma \ref{lemma:polynomial-lipschitz}). As the dominant factor in the runtime of Theorem \ref{teo:pauli_power_super_sparse} is $\OC(2^{\widetilde{L}})$ for Problem \ref{def:general_problem} and $\OC(2^{2\widetilde{L}})$ for Problem \ref{def:general_problem-lm}, we have an efficient algorithm if we ask for constant precision $\varepsilon$, inverse-polynomial failure probability, and when our stated condition is satisfied. Specifically, choosing any constant $\varepsilon\leq C \eta$ ensures that $\varepsilon' \leq \eta$ as previously required.
\end{proof}

Let us end with two contextualizations of Theorem \ref{teo:pauli-sketch}. The specified sampling access can be efficiently instantiated starting from Pauli access as a preprocessing step whenever the number of Pauli terms is $L=\OC(\polylog{N})$, or for larger $L$ whenever the coefficients have sufficient structure. Finally, we can understand condition \eqref{eq:norm-surpression} as a suppressed-norm condition on the $1$-norm of Pauli coefficients. Equally, we can see this as a condition on other matrix norms via standard norm conversions -- for instance, a sufficient condition which implies \eqref{eq:norm-surpression} is to suppress the Frobenius norm of a matrix as $\|A\|_F = \OC(\sqrt{N}/(d^2\sqrt{L}\log^{1.5}(N)))$.


\section*{Acknowledgements}
\addcontentsline{toc}{section}{\protect\numberline{}Acknowledgements}

The authors would like to thank Simon Apers, Ariel Bendersky, Fernando Brand\~{a}o,  Tom O’Leary, and James Watson for helpful discussions. MB acknowledges support from the EPSRC Grant number EP/W032643/1 and the Excellence Cluster - Matter and Light for Quantum Computing (ML4Q). SW and MB thank the Technology Innovation Institute for scientific visits, when part of this work was carried out. SC acknowledges financial support from the Technology Innovation Institute for a long-term internship.



{ \hypersetup{linkcolor=Emerald} \printbibliography }

@preamble{ "\newcommand{\lName}{0}" }

@preamble{"\setlength{\emergencystretch}{2em}"}

@preamble{"\newcommand{\origttfamily}{}"}

@preamble{"\let\origttfamily=\ttfamily"}

@preamble{"\renewcommand{\ttfamily}{\origttfamily \hyphenchar\font=`\-}"}

@preamble{ "\newcommand{\arxiv}[1]{arXiv:\href{https://arxiv.org/abs/#1}{\ttfamily{#1}}\removefirstdot}" }

@preamble{ "\newcommand{\arXiv}[1]{arXiv:\href{https://arxiv.org/abs/#1}{\ttfamily{#1}}\removefirstdot}" }

@preamble{ "\newcommand{\iacr}[1]{ePrint:\href{https://eprint.iacr.org/#1}{\ttfamily{#1}}\removefirstdot}" }

@preamble{ "\def\removefirstdot#1{\if.#1{}\else#1\fi}" }

@preamble{ "\providecommand{\multiletter}[1]{#1}\renewcommand{\multiletter}[1]{#1}" }

@preamble{ "\DeclareRobustCommand{\dutchPrefix}[2]{#2}" }

@preamble{ "\providecommand{\dutchPrefix}[2]{#2}\renewcommand{\dutchPrefix}[2]{#2}" }

@preamble{ "\newcommand{\skp}[3]{#2}" }

@preamble{"\newcommand{\focs       }[1]{\if\lName1\skp{  }{Proceedings of the #1 {IEEE} Symposium on Foundations of Computer Science ({FOCS})}{                            }\else{FOCS}\fi}"}

@preamble{"\newcommand{\stoc       }[1]{\if\lName1\skp{  }{Proceedings of the #1 {ACM} Symposium on the Theory of Computing ({STOC})}{                                     }\else{STOC}\fi}"}

@preamble{"\newcommand{\soda       }[1]{\if\lName1\skp{  }{Proceedings of the #1 {ACM-SIAM} Symposium on Discrete Algorithms ({SODA})}{                                    }\else{SODA}\fi}"}

@preamble{"\newcommand{\stacs      }[1]{\if\lName1\skp{  }{Proceedings of the #1 Symposium on Theoretical Aspects of Computer Science ({STACS})}{                          }\else{STACS}\fi}"}

@preamble{"\newcommand{\itcs       }[1]{\if\lName1\skp{  }{Proceedings of the #1 Innovations in Theoretical Computer Science Conference ({ITCS})}{                         }\else{ITCS}\fi}"}

@preamble{"\newcommand{\fsttcs     }[1]{\if\lName1\skp{  }{Proceedings of the #1 International Conference on Foundations of Software Technology and Theoretical Computer Science ({FSTTCS})}{ }\else{FSTTCS}\fi}"}

@preamble{"\newcommand{\mfcs       }[1]{\if\lName1\skp{  }{Proceedings of the #1 International Symposium on Mathematical Foundations of Computer Science ({MFCS})}{        }\else{MFCS}\fi}"}

@preamble{"\newcommand{\ccc        }[1]{\if\lName1\skp{  }{Proceedings of the #1 {IEEE} Conference on Computational Complexity ({CCC})}{                                   }\else{CCC}\fi}"}

@preamble{"\newcommand{\isit       }[1]{\if\lName1\skp{  }{Proceedings of the #1 {IEEE} International Symposium on Information Theory ({ISIT})}{                           }\else{ISIT}\fi}"}

@preamble{"\newcommand{\colt       }[1]{\if\lName1\skp{  }{Proceedings of the #1 Conference On Learning Theory ({COLT})}{                                                  }\else{COLT}\fi}"}

@preamble{"\newcommand{\nips       }[1]{\if\lName1\skp{  }{Advances in Neural Information Processing Systems #1 ({NeurIPS})}{                                                 }\else{NeurIPS}\fi}"}

@preamble{"\newcommand{\aistats    }[1]{\if\lName1\skp{  }{Proceedings of the #1 International Conference on Artificial Intelligence and Statistics ({AISTATS})}{          }\else{AISTATS}\fi}"}

@preamble{"\newcommand{\icml       }[1]{\if\lName1\skp{  }{Proceedings of the #1 International Conference on Machine Learning ({ICML})}{                                   }\else{ICML}\fi}"}

@preamble{"\newcommand{\icalp      }[1]{\if\lName1\skp{  }{Proceedings of the #1 International Colloquium on Automata, Languages, and Programming ({ICALP})}{              }\else{ICALP}\fi}"}

@preamble{"\newcommand{\esa        }[1]{\if\lName1\skp{  }{Proceedings of the #1 Annual European Symposium on Algorithms ({ESA})}{                                         }\else{ESA}\fi}"}

@preamble{"\newcommand{\tqc        }[1]{\if\lName1\skp{  }{Proceedings of the #1 Conference on the Theory of Quantum Computation, Communication, and Cryptography ({TQC})}{}\else{TQC}\fi}"}

@preamble{"\newcommand{\isca        }[1]{\if\lName1\skp{  }{Proceedings of the #1 International Symposium on       Computer Architecture ({ISCA})}{}\else{ISCA}\fi}"}

@preamble{"\newcommand{\isaac      }[1]{\if\lName1\skp{  }{Proceedings of the #1 International Symposium on Algorithms and Computation ({ISAAC})}{                         }\else{ISAAC}\fi}"}

@preamble{"\newcommand{\aaai       }[1]{\if\lName1\skp{  }{Proceedings of the #1 AAAI Conference on Artificial Intelligence}{                                              }\else{AAAI}\fi}"}

@preamble{"\newcommand{\socg       }[1]{\if\lName1\skp{  }{Proceedings of the #1 Annual Symposium on Computational geometry}{                                              }\else{SoCG}\fi}"}

@preamble{"\newcommand{\sofsem     }[1]{\if\lName1\skp{  }{SOFSEM #1: Theory and Practice of Computer Science}{                                                            }\else{SOFSEM}\fi}"}

@preamble{"\newcommand{\ecc        }[1]{\if\lName1\skp{  }{#1 European Control Conference (ECC)}{                                                                          }\else{ECC}\fi}"}

@preamble{"\newcommand{\crypto     }[1]{\if\lName1\skp{  }{Advances in Cryptology -- CRYPTO #1}{                                                                           }\else{CRYPTO}\fi}"}

@preamble{"\newcommand{\asiacrypt  }[1]{\if\lName1\skp{  }{Advances in Cryptology -- ASIACRYPT #1}{                                                                        }\else{ASIACRYPT}\fi}"}

@preamble{"\newcommand{\eurocrypt  }[1]{\if\lName1\skp{  }{Advances in Cryptology -- EUROCRYPT #1}{                                                                        }\else{EUROCRYPT}\fi}"}

@preamble{"\newcommand{\pqcrypto   }[1]{\if\lName1\skp{  }{Post-Quantum Cryptography}{                                                                                     }\else{PQCrypto}\fi}"}

@preamble{"\newcommand{\scConference}[1]{\if\lName1\skp{  }{Proceedings of the International Conference for High Performance Computing, Networking, Storage and Analysis}{  }\else{SC}\fi}"}

@preamble{"\newcommand{\fccm}[1]{\if\lName1\skp{  }{#1 {IEEE} Annual International Symposium on Field-Programmable Custom Computing Machines}{  }\else{FCCM}\fi}"}

@preamble{"\newcommand{\lattice       }[1]{\if\lName1\skp{  }{Proceedings of the #1 International Symposium on Lattice Field Theory}{                                              }\else{Lattice}\fi}"}

@preamble{"\newcommand{\aft       }[1]{\if\lName1\skp{  }{Proceedings of the #1 ACM Conference on Advances in Financial Technologies}{                                              }\else{AFT}\fi}"}

@preamble{"\newcommand{\secConference       }[1]{\if\lName1\skp{  }{Proceedings of the IEEE/ACM #1 Symposium on Edge Computing}{                                              }\else{SEC}\fi}"}

@preamble{"\newcommand{\ims       }[1]{\if\lName1\skp{  }{#1 {IEEE} {MTT}-{S} International Microwave Symposium (IMS)}{                                              }\else{IMS}\fi}"}

@preamble{"\newcommand{\jacm          }{\if\lName1\skp{  }{Journal of the ACM}{                             }\else{J. ACM}\fi}"}

@preamble{"\newcommand{\acmta         }{\if\lName1\skp{  }{ACM Transactions on Algorithms}{                 }\else{{ACM} Trans. Algorithms}\fi}"}

@preamble{"\newcommand{\acmtct        }{\if\lName1\skp{  }{ACM Transactions on Computation Theory}{         }\else{ACM Trans. Comput. Theory}\fi}"}

@preamble{"\newcommand{\acmtqc        }{\if\lName1\skp{  }{ACM Transactions on Quantum Computing}{          }\else{ACM Trans. Quantum Comput.}\fi}"}

@preamble{"\newcommand{\acmjetcs        }{\if\lName1\skp{  }{   ACM Journal on Emerging Technologies in Computing Systems }{          }\else{ACM J. Emerg. Technol. Comput. Syst.}\fi}"}

@preamble{"\newcommand{\canadianjmath }{\if\lName1\skp{  }{Canadian Journal of Mathematics      }{          }\else{Can. J. Math.}\fi}"}

@preamble{"\newcommand{\jams          }{\if\lName1\skp{  }{Journal of the AMS}{                             }\else{J. AMS}\fi}"}

@preamble{"\newcommand{\bullAMS        }{\if\lName1\skp{  }{Bulletin of the American Mathematical Society}{                                }\else{Bull. AMS}\fi}"}

@preamble{"\newcommand{\pams          }{\if\lName1\skp{  }{Proceedings of the AMS}{                         }\else{Proc. AMS}\fi}"}

@preamble{"\newcommand{\linalgappl    }{\if\lName1\skp{  }{Linear Algebra and its Applications}{            }\else{Linear Algebra App.}\fi}"}

@preamble{"\newcommand{\jalgo         }{\if\lName1\skp{  }{Journal of Algorithms}{                          }\else{J. Algorithms}\fi}"}

@preamble{"\newcommand{\jcss          }{\if\lName1\skp{  }{Journal of Computer and System Sciences}{        }\else{J. Comput. Syst. Sci.}\fi}"}

@preamble{"\newcommand{\jcomputapplmath}{\if\lName1\skp{  }{Journal of Computational and Applied Mathematics}{        }\else{J. Comput. Appl. Math.}\fi}"}

@preamble{"\newcommand{\cc            }{\if\lName1\skp{  }{Computational Complexity}{                       }\else{Comput. Complex.}\fi}"}

@preamble{"\newcommand{\algor         }{\if\lName1\skp{  }{Algorithmica}{                                   }\else{Algorithmica}\fi}"}

@preamble{"\newcommand{\comb          }{\if\lName1\skp{  }{Combinatorica}{                                  }\else{Combinatorica}\fi}"}

@preamble{"\newcommand{\cacm          }{\if\lName1\skp{  }{Communications of the ACM}{                      }\else{Commun. ACM}\fi}"}

@preamble{"\newcommand{\sigart        }{\if\lName1\skp{  }{SIGART Bulletin}{                                }\else{SIGART Bull.}\fi}"}

@preamble{"\newcommand{\sigactn       }{\if\lName1\skp{  }{SIGACT News}{                                    }\else{SIGACT News}\fi}"}

@preamble{"\newcommand{\eatcsbul      }{\if\lName1\skp{  }{Bulletin of the {EATCS}}{                        }\else{Bull. {EATCS}}\fi}"}

@preamble{"\newcommand{\siamrev       }{\if\lName1\skp{  }{SIAM Review}{                                    }\else{SIAM Rev.}\fi}"}

@preamble{"\newcommand{\siamjc        }{\if\lName1\skp{  }{SIAM Journal on Computing}{                      }\else{SIAM J. Comp.}\fi}"}

@preamble{"\newcommand{\siamjo        }{\if\lName1\skp{  }{SIAM Journal on Optimization}{                   }\else{SIAM J. Opt.}\fi}"}

@preamble{"\newcommand{\siamjdm       }{\if\lName1\skp{  }{SIAM Journal on Discrete Mathematics}{           }\else{SIAM J. Disc. Math.}\fi}"}

@preamble{"\newcommand{\siamjnum      }{\if\lName1\skp{  }{SIAM Journal on Numerical Analysis}{             }\else{SIAM J. Num. Anal.}\fi}"}

@preamble{"\newcommand{\siamjmathanal }{\if\lName1\skp{  }{SIAM Journal on Mathematical Analysis}{          }\else{SIAM J. Math. Anal.}\fi}"}

@preamble{"\newcommand{\discmath      }{\if\lName1\skp{  }{Discrete Mathematics}{                           }\else{Disc. Math.}\fi}"}

@preamble{"\newcommand{\das           }{\if\lName1\skp{  }{Discrete Applied Mathematics}{                   }\else{Disc. App. Math.}\fi}"}

@preamble{"\newcommand{\annmath      }{\if\lName1\skp{  }{Annals of Mathematics}{              }\else{Ann. Math.}\fi}"}

@preamble{"\newcommand{\amatstat      }{\if\lName1\skp{  }{Annals of Mathematical Statistics}{              }\else{Ann. Math. Stat.}\fi}"}

@preamble{"\newcommand{\rms           }{\if\lName1\skp{  }{Russian Mathematical Surveys}{                   }\else{Russ. Math. Surv.}\fi}"}

@preamble{"\newcommand{\invmath       }{\if\lName1\skp{  }{Inventiones Mathematicae}{                       }\else{Inv. Math.}\fi}"}

@preamble{"\newcommand{\jnumber       }{\if\lName1\skp{  }{Journal of Number Theory}{                       }\else{J. Num. Th.}\fi}"}

@preamble{"\newcommand{\tcs           }{\if\lName1\skp{  }{Theoretical Computer Science}{                   }\else{Theor. Comput. Sci.}\fi}"}

@preamble{"\newcommand{\numeralgorithms}{\if\lName1\skp{  }{Numerical Algorithms}{                          }\else{Numer. Algorithms}\fi}"}

@preamble{"\newcommand{\toc           }{\if\lName1\skp{  }{Theory of Computing}{                            }\else{Theory Comput.}\fi}"}

@preamble{"\newcommand{\cjtcs         }{\if\lName1\skp{  }{Chicago Journal of Theoretical Computer Science}{}\else{Chicago J. Theoret. Comput. Sci.}\fi}"}

@preamble{"\newcommand{\mathprogram   }{\if\lName1\skp{  }{Mathematical Programming}{}\else{Math. Program.}\fi}"}

@preamble{"\newcommand{\mathcomput    }{\if\lName1\skp{  }{Mathematics of Computation}{}\else{Math. Comput.}\fi}"}

@preamble{"\newcommand{\jfourieranalappl   }{\if\lName1\skp{  }{Journal of Fourier Analysis and Applications}{}\else{J. Fourier Anal. Appl.}\fi}"}

@preamble{"\newcommand{\dcg           }{\if\lName1\skp{  }{Discrete \& Computational Geometry}{}\else{Discrete Comput. Geom.}\fi}"}

@preamble{"\newcommand{\randstructalg }{\if\lName1\skp{  }{Random Structures \& Algorithms}{}\else{Rand. Struct. Algorithms}\fi}"}

@preamble{"\newcommand{\intjunconvent }{\if\lName1\skp{  }{International Journal of Unconventional Computing}{}\else{Int. J. Unconv. Comput.}\fi}"}

@preamble{"\newcommand{\machlearnscitech }{\if\lName1\skp{  }{Machine Learning: Science and Technology}{}\else{Mach. Learn.: Sci. Technol.}\fi}"}

@preamble{"\newcommand{\machlearning  }{\if\lName1\skp{  }{Machine Learning}{}\else{Mach. Learn.}\fi}"}

@preamble{"\newcommand{\neuralprocesslett } {\if\lName1\skp{  }{Neural Processing Letters}{}\else{Neural Process. Lett.}\fi}"}

@preamble{"\newcommand{\neuralcomput }{\if\lName1\skp{  }{Neural Computation}{}\else{Neural Comput.}\fi}"}

@preamble{"\newcommand{\frontartifintell}{\if\lName1\skp{  }{Frontiers in Artificial Intelligence}{}\else{Front. Artif. Intell.}\fi}"}

@preamble{"\newcommand{\jgloboptim}{\if\lName1\skp{  }{Journal of Global Optimization}{}\else{J. Glob. Optim.}\fi}"}

@preamble{"\newcommand{\epjdatasci}{\if\lName1\skp{  }{EPJ Data Science}{}\else{EPJ Data Sci.}\fi}"}

@preamble{"\newcommand{\epjquantumtechnol}{\if\lName1\skp{  }{EPJ Quantum Technology}{}\else{EPJ Quantum Technol.}\fi}"}

@preamble{"\newcommand{\applsci       }{\if\lName1\skp{  }{Applied Sciences}{    }\else{Appl. Sci.}\fi}"}

@preamble{"\newcommand{\applphysrev   }{\if\lName1\skp{  }{Applied Physics Reviews}{    }\else{Appl. Phys. Rev.}\fi}"}

@preamble{"\newcommand{\advquantumtechnol          }{\if\lName1\skp{  }{Advanced Quantum Technologies}{    }\else{Adv. Quantum Technol.}\fi}"}

@preamble{"\newcommand{\annphys       }{\if\lName1\skp{  }{Annals of Physics}{    }\else{Ann. Phys.}\fi}"}

@preamble{"\newcommand{\annualrevCMP          }{\if\lName1\skp{  }{Annual Review of Condensed Matter Physics}{    }\else{Annu. Rev. Condens. Matter Phys.}\fi}"}

@preamble{"\newcommand{\annualrevNPS          }{\if\lName1\skp{  }{Annual Review of Nuclear and Particle Science}{    }\else{Annu. Rev. Nucl. Part. Sci.}\fi}"}

@preamble{"\newcommand{\quantum       }{\if\lName1\skp{  }{Quantum}{                                          }\else{Quantum}\fi}"}

@preamble{"\newcommand{\quantumscitechnol       }{\if\lName1\skp{  }{Quantum Science and Technology}{                                          }\else{Quantum Sci. Technol.}\fi}"}

@preamble{"\newcommand{\cmp           }{\if\lName1\skp{  }{Communications in Mathematical Physics}{             }\else{Commun. Math. Phys.}\fi}"}

@preamble{"\newcommand{\frontphys     }{\if\lName1\skp{  }{Frontiers in Physics}{                               }\else{Front. Phys.}\fi}"}

@preamble{"\newcommand{\jmp           }{\if\lName1\skp{  }{Journal of Mathematical Physics}{                    }\else{J. Math. Phys.}\fi}"}

@preamble{"\newcommand{\rspa          }{\if\lName1\skp{  }{Proceedings of the Royal Society A}{                 }\else{Proc. R. Soc. A}\fi}"}

@preamble{"\newcommand{\philostransroyal}{\if\lName1\skp{  }{Philosophical Transactions of the Royal Society A}{                 }\else{Philos. Trans. R. Soc. A}\fi}"}

@preamble{"\newcommand{\qic           }{\if\lName1\skp{  }{Quantum Information and Computation}{                }\else{Quantum Inf. Comput.}\fi}"}

@preamble{"\newcommand{\qip           }{\if\lName1\skp{  }{Quantum Information Processing}{                }\else{Quantum Inf. Process.}\fi}"}

@preamble{"\newcommand{\physrev       }{\if\lName1\skp{  }{Physical Review}{                                    }\else{Phys. Rev.}\fi}"}

@preamble{"\newcommand{\pra           }{\if\lName1\skp{  }{Physical Review A}{                                  }\else{Phys. Rev. A}\fi}"}

@preamble{"\newcommand{\prb           }{\if\lName1\skp{  }{Physical Review B}{                                  }\else{Phys. Rev. B}\fi}"}

@preamble{"\newcommand{\prc           }{\if\lName1\skp{  }{Physical Review C}{                                  }\else{Phys. Rev. C}\fi}"}

@preamble{"\newcommand{\prd           }{\if\lName1\skp{  }{Physical Review D}{                                  }\else{Phys. Rev. D}\fi}"}

@preamble{"\newcommand{\pre           }{\if\lName1\skp{  }{Physical Review E}{                                  }\else{Phys. Rev. E}\fi}"}

@preamble{"\newcommand{\prr           }{\if\lName1\skp{  }{Physical Review Research}{                           }\else{Phys. Rev. Res.}\fi}"}

@preamble{"\newcommand{\prx           }{\if\lName1\skp{  }{Physical Review X}{                                  }\else{Phys. Rev. X}\fi}"}

@preamble{"\newcommand{\prxq          }{\if\lName1\skp{  }{PRX Quantum}{                                        }\else{PRX Quantum}\fi}"}

@preamble{"\newcommand{\prl           }{\if\lName1\skp{  }{Phys. Rev. Lett.}{                            }\else{Phys. Rev. Lett.}\fi}"}

@preamble{"\newcommand{\europhyslett  }{\if\lName1\skp{  }{Europhysics Letters}{                                }\else{Europhys. Lett.}\fi}"}

@preamble{"\newcommand{\epja          }{\if\lName1\skp{  }{The European Physical Journal A}{                    }\else{Euro. Phys. J. A}\fi}"}

@preamble{"\newcommand{\epjd          }{\if\lName1\skp{  }{The European Physical Journal D}{                    }\else{Euro. Phys. J. D}\fi}"}

@preamble{"\newcommand{\njp           }{\if\lName1\skp{  }{New Journal of Physics}{                             }\else{New J. Phys.}\fi}"}

@preamble{"\newcommand{\prapp         }{\if\lName1\skp{  }{Physical Review Applied}{                            }\else{Phys. Rev. Appl.}\fi}"}

@preamble{"\newcommand{\physrep       }{\if\lName1\skp{  }{Physics Reports}{                                    }\else{Phys. Rep.}\fi}"}

@preamble{"\newcommand{\rmp           }{\if\lName1\skp{  }{Reviews of Modern Physics}{                          }\else{Rev. Mod. Phys.}\fi}"}

@preamble{"\newcommand{\repprogphys   }{\if\lName1\skp{  }{Reports on Progress in Physics}{                     }\else{Rep. Prog. Phys.}\fi}"}

@preamble{"\newcommand{\physplasmas   }{\if\lName1\skp{  }{Physics of Plasmas}{                                 }\else{Phys. Plasmas}\fi}"}

@preamble{"\newcommand{\phystoday     }{\if\lName1\skp{  }{Physics Today}{                                      }\else{Phys. Today}\fi}"}

@preamble{"\newcommand{\physics       }{\if\lName1\skp{  }{Physics}{                                            }\else{Phys.}\fi}"}

@preamble{"\newcommand{\nature        }{\if\lName1\skp{  }{Nature}{                                             }\else{Nature}\fi}"}

@preamble{"\newcommand{\natcomm       }{\if\lName1\skp{  }{Nature Communications}{                              }\else{Nat. Commun.}\fi}"}

@preamble{"\newcommand{\natcomputsci  }{\if\lName1\skp{  }{Nature Computational Science}{                       }\else{Nat. Comput. Sci.}\fi}"}

@preamble{"\newcommand{\natphys       }{\if\lName1\skp{  }{Nature Physics}{                                     }\else{Nat. Phys.}\fi}"}

@preamble{"\newcommand{\natphotonics  }{\if\lName1\skp{  }{Nature Photonics}{                                     }\else{Nat. Photonics}\fi}"}

@preamble{"\newcommand{\natrevphys    }{\if\lName1\skp{  }{Nature Reviews Physics}{                                     }\else{Nat. Rev. Phys.}\fi}"}

@preamble{"\newcommand{\natrevmater   }{\if\lName1\skp{  }{Nature Reviews Materials}{                                     }\else{Nat. Rev. Mater.}\fi}"}

@preamble{"\newcommand{\natrevmethodsprimers}{\if\lName1\skp{}{Nature Reviews Methods Primers}{                               }\else{Nat. Rev. Methods Primers}\fi}"}

@preamble{"\newcommand{\npjqi         }{\if\lName1\skp{  }{npj Quantum Information}{                            }\else{npj Quant. Inf.}\fi}"}

@preamble{"\newcommand{\scirep        }{\if\lName1\skp{  }{Scientific Reports}{                                 }\else{Sci. Rep.}\fi}"}

@preamble{"\newcommand{\science       }{\if\lName1\skp{  }{Science}{                                            }\else{Science}\fi}"}

@preamble{"\newcommand{\sciadv      }{\if\lName1\skp{  }{Science Advances}{                                   }\else{Sci. Adv.}\fi}"}

@preamble{"\newcommand{\scibull      }{\if\lName1\skp{  }{Science Bulletin}{                                   }\else{Sci. Bull.}\fi}"}

@preamble{"\newcommand{\jchemphys           }{\if\lName1\skp{  }{The Journal of Chemical Physics}{ }\else{J. Chem. Phys.}\fi}"}

@preamble{"\newcommand{\jphyschemlett      }{\if\lName1\skp{  }{The Journal of Physical Chemistry Letters}{ }\else{J. Phys. Chem. Lett.}\fi}"}

@preamble{"\newcommand{\jpa           }{\if\lName1\skp{  }{Journal of Physics A: Mathematical and Theoretical}{ }\else{J. Phys. A}\fi}"}

@preamble{"\newcommand{\jpg           }{\if\lName1\skp{  }{Journal of Physics G: Nuclear and Particle Physics}{ }\else{J. Phys. G}\fi}"}

@preamble{"\newcommand{\ijtp          }{\if\lName1\skp{  }{International Journal of Theoretical Physics}{       }\else{Int. J. Th. Phys.}\fi}"}

@preamble{"\newcommand{\jmo           }{\if\lName1\skp{  }{Journal of Modern Optics}{                           }\else{J. Mod. Opt.}\fi}"}

@preamble{"\newcommand{\jhep           }{\if\lName1\skp{  }{Journal of High Energy Physics}{                           }\else{J. High Energy Phys.}\fi}"}

@preamble{"\newcommand{\jstatph       }{\if\lName1\skp{  }{Journal of Statistical Physics}{                     }\else{J. Stat. Phys.}\fi}"}

@preamble{"\newcommand{\jcompphys       }{\if\lName1\skp{  }{Journal of Computational Physics}{                 }\else{J. Comput. Phys.}\fi}"}

@preamble{"\newcommand{\computphyscommun}{\if\lName1\skp{  }{Computer Physics Communications}{               }\else{Comput. Phys. Commun.}\fi}"}

@preamble{"\newcommand{\jstatmech     }{\if\lName1\skp{  }{Journal of Statistical Mechanics: Theory and Experiment}{ }\else{J. Stat. Mech. Theory Exp.}\fi}"}

@preamble{"\newcommand{\pnas          }{\if\lName1\skp{  }{Proceedings of the National Academy of Sciences}{    }\else{Proc. Natl. Acad. Sci.}\fi}"}

@preamble{"\newcommand{\avsquantsci          }{\if\lName1\skp{  }{AVS Quantum Science}{    }\else{AVS Quantum Sci.}\fi}"}

@preamble{"\newcommand{\quantummachintell    }{\if\lName1\skp{  }{Quantum Machine Intelligence}{    }\else{Quantum Mach. Intell.}\fi}"}

@preamble{"\newcommand{\natmachintell    }{\if\lName1\skp{  }{Nature Machine Intelligence}{    }\else{Nat. Mach. Intell.}\fi}"}

@preamble{"\newcommand{\scichinainfsci   }{\if\lName1\skp{  }{ Science China Information Sciences}{    }\else{Sci. China Inf. Sci.}\fi}"}

@preamble{"\newcommand{\jmagnreson   }{\if\lName1\skp{  }{Journal of Magnetic Resonance}{    }\else{J. Magn. Reson.}\fi}"}

@preamble{"\newcommand{\lncs          }{\if\lName1\skp{  }{Lecture Notes in Computer Science}{                  }\else{L. Notes Comp. Sci.}\fi}"}

@preamble{"\newcommand{\lnai          }{\if\lName1\skp{  }{Lecture Notes in Artificial Intelligence}{           }\else{L. Notes Art. Int.}\fi}"}

@preamble{"\newcommand{\lnm           }{\if\lName1\skp{  }{Lecture Notes in Mathematics}{                       }\else{L. Notes Math.}\fi}"}

@preamble{"\newcommand{\tams          }{\if\lName1\skp{  }{Transactions of the American Mathematical Society}{  }\else{Trans. AMS}\fi}"}

@preamble{"\newcommand{\ieeetit       }{\if\lName1\skp{  }{{IEEE} Transactions on Information Theory}{          }\else{{IEEE} Trans. Inf. Theory}\fi}"}

@preamble{"\newcommand{\ieeetnnls       }{\if\lName1\skp{  }{{IEEE} Transactions on Neural Networks and Learning Systems}{          }\else{{IEEE} Trans. Neural Netw. Learn. Syst.}\fi}"}

@preamble{"\newcommand{\ieeetcad       }{\if\lName1\skp{  }{{IEEE} Transactions on Computer-Aided Design of Integrated Circuits and Systems}{          }\else{{IEEE} Trans. Comput.-Aided Des. Integr. Circuits Syst.}\fi}"}

@preamble{"\newcommand{\ieeetqe       }{\if\lName1\skp{  }{{IEEE} Transactions on Quantum Engineering}{          }\else{{IEEE} Trans. Quantum Eng.}\fi}"}

@preamble{"\newcommand{\ieeejsac       }{\if\lName1\skp{  }{{IEEE} Journal on Selected Areas in Communications}{          }\else{{IEEE} J. Sel. Areas Commun.}\fi}"}

@preamble{"\newcommand{\ieeebits       }{\if\lName1\skp{  }{{IEEE} {BITS} the Information Theory Magazine}{          }\else{{IEEE} {BITS} Inf. Theory Mag.}\fi}"}

@preamble{"\newcommand{\quantumeng       }{\if\lName1\skp{  }{Quantum Engineering}{          }\else{Quantum Eng.}\fi}"}

@preamble{"\newcommand{\iscs          }{\if\lName1\skp{  }{International Series in Computer Science}{           }\else{Int. Ser. Comp. Sci.}\fi}"}

@preamble{"\newcommand{\tocl          }{\if\lName1\skp{  }{Theory of Computing Library}{                        }\else{Th. Comp. Lib.}\fi}"}

@preamble{"\newcommand{\actanumer     }{\if\lName1\skp{  }{Acta Numerica}{                        }\else{Acta Numer.}\fi}"}

@preamble{"\newcommand{\jderiv     }{\if\lName1\skp{  }{The Journal of Derivatives}{                        }\else{J. Deriv.}\fi}"}

@preamble{"\newcommand{\chemrev     }{\if\lName1\skp{  }{Chemical Reviews}{                        }\else{Chem. Rev.}\fi}"}

@preamble{"\newcommand{\WIREsCMS     }{\if\lName1\skp{  }{WIREs Molecular Computational Science}{                        }\else{WIREs Comput. Mol. Sci.}\fi}"}

@preamble{"\newcommand{\accchemres     }{\if\lName1\skp{  }{Accounts of Chemical Research}{                        }\else{Acc. Chem. Res.}\fi}"}

@preamble{"\newcommand{\chemphyslett     }{\if\lName1\skp{  }{Chemical Physics Letters}{                        }\else{Chem. Phys. Lett.}\fi}"}

@preamble{"\newcommand{\molphys     }{\if\lName1\skp{  }{Molecular Physics}{                        }\else{Mol. Phys.}\fi}"}

@preamble{"\newcommand{\commchem     }{\if\lName1\skp{  }{Communications Chemistry}{                        }\else{Commun. Chem.}\fi}"}

@preamble{"\newcommand{\acscentsci     }{\if\lName1\skp{  }{ACS Central Science}{                        }\else{ACS Cent. Sci.}\fi}"}

@preamble{"\newcommand{\jchemtheorycomput}{\if\lName1\skp{  }{Journal of Chemical Theory and Computation}{                        }\else{J. Chem. Theory Comput.}\fi}"}

@preamble{"\newcommand{\chemsci        }{\if\lName1\skp{  }{Chemical Science}{                        }\else{Chem. Sci.}\fi}"}

@preamble{"\newcommand{\intjqchem        }{\if\lName1\skp{  }{International journal of quantum chemistry}{                        }\else{Int. J. Quantum Chem.}\fi}"}

@preamble{"\newcommand{\jfqa        }{\if\lName1\skp{  }{Journal of Financial and Quantitative Analysis}{                        }\else{J. Financial Quant. Anal.}\fi}"}

@article{Venegas_Andraca_2012,
   title={Quantum walks: a comprehensive review},
   volume={11},
   ISSN={1573-1332},
   url={http://dx.doi.org/10.1007/s11128-012-0432-5},
   DOI={10.1007/s11128-012-0432-5},
   number={5},
   journal={Quantum Information Processing},
   publisher={Springer Science and Business Media LLC},
   author={Venegas-Andraca, Salvador Elías},
   note={\arxiv{1201.4780}},
   year={2012},
   month=jul, pages={1015–1106} }

@article{Seki2021,
  title = {Quantum Power Method by a Superposition of Time-Evolved States},
  author = {Seki, Kazuhiro and Yunoki, Seiji},
  note = {\arxiv{2008.03661}},
  journal = {PRX Quantum},
  volume = {2},
  issue = {1},
  pages = {010333},
  numpages = {45},
  year = {2021},
  publisher = {American Physical Society},
  doi = {10.1103/PRXQuantum.2.010333},
  url = {https://link.aps.org/doi/10.1103/PRXQuantum.2.010333}
}

@article{Bespalova_2021,
   title={Hamiltonian Operator Approximation for Energy Measurement and Ground-State Preparation},
   volume={2},
   ISSN={2691-3399},
   url={http://dx.doi.org/10.1103/PRXQuantum.2.030318},
   DOI={10.1103/prxquantum.2.030318},
   number={3},
   journal={PRX Quantum},
   publisher={American Physical Society (APS)},
   author={Bespalova, Tatiana A. and Kyriienko, Oleksandr},
   note = {\arxiv{2009.03351}},
   year={2021},
   month=aug }

@article{shpilka2010arithmetic,
  title={Arithmetic circuits: A survey of recent results and open questions},
  author={Shpilka, Amir and Yehudayoff, Amir and others},
  journal={Foundations and Trends{\textregistered} in Theoretical Computer Science},
  volume={5},
  number={3--4},
  pages={207--388},
  year={2010},
  publisher={Now Publishers, Inc.}
}

@unpublished{wang2024comprehensive,
  title={A Comprehensive Study of Quantum Arithmetic Circuits},
  author={Wang, Siyi and Li, Xiufan and Lee, Wei Jie Bryan and Deb, Suman and Lim, Eugene and Chattopadhyay, Anupam},
  note={\arxiv{2406.03867}},
  url={https://arxiv.org/abs/2406.03867},
  doi={10.48550/arXiv.2406.03867},
  year={2024}
}

@article{Zhang_2024,
   title={Parallel Quantum Algorithm for Hamiltonian Simulation},
   volume={8},
   ISSN={2521-327X},
   url={http://dx.doi.org/10.22331/q-2024-01-15-1228},
   DOI={10.22331/q-2024-01-15-1228},
   journal={Quantum},
   note={\arxiv{2105.11889}},
   publisher={Verein zur Forderung des Open Access Publizierens in den Quantenwissenschaften},
   author={Zhang, Zhicheng and Wang, Qisheng and Ying, Mingsheng},
   year={2024},
   pages={1228} }

@article{bennett1973logical,
  title={Logical reversibility of computation},
  author={Bennett, Charles H},
  journal={IBM journal of Research and Development},
  volume={17},
  number={6},
  pages={525--532},
  year={1973},
  publisher={IBM}
}

@article{reif1986logarithmic,
  title={Logarithmic depth circuits for algebraic functions},
  author={Reif, John H},
  journal={SIAM Journal on Computing},
  volume={15},
  number={1},
  pages={231--242},
  year={1986},
  publisher={SIAM},
doi={https://doi.org/10.1137/0215017}
}

@unpublished{low2024quantumeigenvalueprocessing,
      title={Quantum eigenvalue processing}, 
      author={Guang Hao Low and Yuan Su},
      year={2024},
      note		= {\arxiv{2401.06240}},
      url={https://arxiv.org/abs/2401.06240}, 
}

@article{aharonov2013guest,
  title={Guest column: the quantum PCP conjecture},
  author={Aharonov, Dorit and Arad, Itai and Vidick, Thomas},
  journal={ACM SIGACT news},
  volume={44},
  number={2},
  pages={47--79},
  year={2013},
  publisher={ACM New York, NY, USA},
note={\arxiv{1309.7495}},
doi={https://dl.acm.org/doi/10.1145/2491533.2491549}
}

@Article{	  aharononv2007AdiabaticQStateGeneration,
  author	= {Aharonov, Dorit and Ta‐Shma, Amnon},
  doi		= {10.1137/060648829},
  journal	= {\siamjc},
  note		= {Earlier version in
		  \href{https://doi.org/10.1145/780542.780546}{\emph{STOC'03}},
		  \arxiv{quant-ph/0301023}},
  number	= {1},
  pages		= {47-82},
  title		= {Adiabatic Quantum State Generation},
  volume	= {37},
  year		= {2007}
}

@Article{	  aleksandrov2011EstimatesOfOperatorModuli,
  author	= {Alexei B. Aleksandrov and Vladimir V. Peller},
  doi		= {10.1016/j.jfa.2011.07.009},
  journal	= {Journal of Functional Analysis},
  note		= {\arxiv{1104.3553}},
  number	= {10},
  pages		= {2741--2796},
  title		= {Estimates of operator moduli of continuity},
  volume	= {261},
  year		= {2011}
}

@unpublished{	  an2023QAlgNonUnitary,
  author	= {An, Dong and Childs, Andrew M and Lin, Lin},
  note		= {\arxiv{2312.03916}},
  title		= {Quantum algorithm for linear non-unitary dynamics with
		  near-optimal dependence on all parameters},
  year		= {2023}
}

@article{an2023LinCombHamSimulation,
  title={Linear combination of Hamiltonian simulation for nonunitary dynamics with optimal state preparation cost},
  author={An, Dong and Liu, Jin-Peng and Lin, Lin},
  journal={Physical Review Letters},
  volume={131},
  number={15},
  pages={150603},
  year={2023},
  publisher={APS},
doi={https://doi.org/10.1103/PhysRevLett.131.150603},
note={\arxiv{2303.01029}}
}

@Unpublished{	  apers2022SimpleBetti,
  author	= {Apers, Simon and Sen, Sayantan and Szab{\'o}, D{\'a}niel},
  note		= {\arXiv{2211.09618}},
  title		= {A (simple) classical algorithm for estimating Betti
		  numbers},
  year		= {2022}
}

@article{apers2024quantumWalksWaveEqn,
  title={Quantum walks, the discrete wave equation and Chebyshev polynomials},
  author={Apers, Simon and Miclo, Laurent},
  note={\arxiv{2402.07809}},
  year={2024}
}

@article{bansal2009classical,
  title={Classical approximation schemes for the ground-state energy of quantum and classical ising spin hamiltonians on planar graphs},
  author={Bansal, Nikhil and Bravyi, Sergey and Terhal, Barbara M},
  journal={Quantum Information \& Computation},
  volume={9},
  number={7},
  pages={701--720},
  year={2009},
  publisher={Rinton Press, Incorporated Paramus, NJ},
note={\arxiv{0705.1115}},
doi={10.5555/2011814.2011826}
}

@book{berestycki2016mixing,
  title={Mixing times of markov chains: techniques and examples},
  author={Berestycki, Nathana{\"e}l},
  publisher={Lecture Notes},
  year={2016}
}

@Article{	  berry2005EffQAlgSimmSparseHam,
  author	= {Berry, Dominic W. and Ahokas, Graeme and Cleve, Richard
		  and Sanders, Barry C.},
  doi		= {10.1007/s00220-006-0150-x},
  journal	= {\cmp},
  note		= {\arxiv{quant-ph/0508139}},
  number	= {2},
  pages		= {359--371},
  title		= {Efficient Quantum Algorithms for Simulating Sparse
		  {H}amiltonians},
  volume	= {270},
  year		= {2007}
}

@Article{	  berry2014HamSimTaylor,
  author	= {Berry, Dominic W. and Childs, Andrew M. and Cleve, Richard
		  and Kothari, Robin and Somma, Rolando D.},
  doi		= {10.1103/PhysRevLett.114.090502},
  journal	= {\prl},
  note		= {\arxiv{1412.4687}},
  number	= {9},
  numpages	= {5},
  pages		= {090502},
  title		= {Simulating {H}amiltonian Dynamics with a Truncated
		  {T}aylor Series},
  volume	= {114},
  year		= {2015}
}

@InCollection{	  brassard2002AmpAndEst,
  author	= {Gilles Brassard and Peter H{\o}yer and Michele Mosca and
		  Alain Tapp},
  booktitle	= {Quantum Computation and Quantum Information: A Millennium
		  Volume},
  doi		= {10.1090/conm/305/05215},
  note		= {\arxiv{quant-ph/0005055}},
  pages		= {53--74},
  publisher	= {AMS},
  series	= {Contemporary Mathematics Series},
  title		= {Quantum Amplitude Amplification and Estimation},
  volume	= {305},
  year		= {2002}
}

@Article{	  buhrman2001QuantumFingerprinting,
  author	= {Buhrman, Harry and Cleve, Richard and Watrous, John and
		  {\dutchPrefix{Wolf}{d}}e Wolf, Ronald},
  doi		= {10.1103/PhysRevLett.87.167902},
  journal	= {\prl},
  note		= {\arxiv{quant-ph/0102001}},
  number	= {16},
  pages		= {167902},
  title		= {Quantum Fingerprinting},
  volume	= {87},
  year		= {2001}
}

@Article{	  campbell2019randomCompiler,
  author	= {Earl Campbell},
  doi		= {10.1103/physrevlett.123.070503},
  journal	= {\prl},
  month		= {8},
  note		= {\arxiv{1811.08017}},
  number	= {7},
  publisher	= {American Physical Society ({APS})},
  title		= {Random Compiler for Fast Hamiltonian Simulation},
  volume	= {123},
  year		= 2019
}

@InProceedings{	  chia2019SampdSubLinLowRankFramework,
  author	= {Chia, Nai-Hui and Gily\'{e}n, Andr\'{a}s and Li, Tongyang
		  and Lin, Han-Hsuan and Tang, Ewin and Wang, Chunhao},
  booktitle	= {\stoc{52nd}},
  doi		= {10.1145/3357713.3384314},
  note		= {\arxiv{1910.06151}},
  numpages	= {14},
  pages		= {387--400},
  title		= {Sampling-Based Sublinear Low-Rank Matrix Arithmetic
		  Framework for Dequantizing Quantum Machine Learning},
  year		= {2020}
}

@Article{	  childs2015QLinSysExpPrec,
  author	= {Andrew M. Childs and Robin Kothari and Rolando D. Somma},
  doi		= {10.1137/16M1087072},
  journal	= {\siamjc},
  note		= {\arxiv{1511.02306}},
  number	= {6},
  pages		= {1920--1950},
  title		= {Quantum Algorithm for Systems of Linear Equations with
		  Exponentially Improved Dependence on Precision},
  volume	= {46},
  year		= {2017}
}

@Article{	  feynman1982SimQPhysWithComputers,
  author	= {Feynman, Richard P.},
  doi		= {10.1007/BF02650179},
  journal	= {\ijtp},
  number	= {6-7},
  pages		= {467--488},
  title		= {Simulating physics with computers},
  volume	= {21},
  year		= {1982}
}

@inproceedings{gharibian2022dequantizingSVT,
  title={Dequantizing the quantum singular value transformation: hardness and applications to quantum chemistry and the quantum PCP conjecture},
  author={Gharibian, Sevag and Le Gall, Fran{\c{c}}ois},
  booktitle={\stoc{54th}},
  pages={19--32},
  year={2022},
note={\arxiv{2111.09079}},
doi={10.1145/3519935.3519991}
}

@Article{	  gilyen2020ImprovedQInspiredAlgorithmForRegression,
  author	= {Gily{\'{e}}n, Andr{\'{a}}s and Song, Zhao and Tang, Ewin},
  doi		= {10.22331/q-2022-06-30-754},
  journal	= {\quantum},
  note		= {\arxiv{2009.07268}},
  pages		= {754},
  publisher	= {{Verein zur F{\"{o}}rderung des Open Access Publizierens
		  in den Quantenwissenschaften}},
  title		= {An improved quantum-inspired algorithm for linear
		  regression},
  volume	= {6},
  year		= {2022}
}

@InProceedings{	  gilyen2018QSingValTransf,
  author	= {Andr{\'a}s Gily{\'e}n and Yuan Su and Guang Hao Low and
		  Nathan Wiebe},
  booktitle	= {\stoc{51st}},
  doi		= {10.1145/3313276.3316366},
  note		= {\arxiv{1806.01838}},
  numpages	= {12},
  pages		= {193--204},
  title		= {Quantum singular value transformation and beyond:
		  {E}xponential improvements for quantum matrix arithmetics},
  year		= {2019}
}

@article{gross2011recovering,
  title={Recovering low-rank matrices from few coefficients in any basis},
  author={Gross, David},
  journal={IEEE Transactions on Information Theory},
  volume={57},
  number={3},
  pages={1548--1566},
  year={2011},
  publisher={IEEE},
doi={10.1109/TIT.2011.2104999},
note={\arxiv{0910.1879}}
}

@Article{	  haah2018ProdDecPerFuncQSignPRoc,
  author	= {Haah, Jeongwan},
  doi		= {10.22331/q-2019-10-07-190},
  journal	= {\quantum},
  note		= {\arxiv{1806.10236}},
  pages		= {190},
  title		= {Product {D}ecomposition of {P}eriodic {F}unctions in
		  {Q}uantum {S}ignal {P}rocessing},
  volume	= {3},
  year		= {2019}
}

@Article{	  harrow2009QLinSysSolver,
  author	= {Harrow, Aram W. and Hassidim, Avinatan and Lloyd, Seth},
  doi		= {10.1103/PhysRevLett.103.150502},
  journal	= {\prl},
  note		= {\arxiv{0811.3171}},
  number	= {15},
  pages		= {150502},
  title		= {Quantum algorithm for linear systems of equations},
  volume	= {103},
  year		= {2009}
}

@unpublished{janzing2006BQPmixing,
  title={BQP-complete problems concerning mixing properties of classical random walks on sparse graphs},
  author={Janzing, Dominik and Wocjan, Pawel},
  year={2006},
note={\arxiv{quant-ph/0610235}}
}

@article{janzing2007simpleBQP,
  title={A simple PromiseBQP-complete matrix problem},
  author={Janzing, Dominik and Wocjan, Pawel},
  journal={Theory of computing},
  volume={3},
  number={1},
  pages={61--79},
  year={2007},
  publisher={Theory of Computing Exchange},
note={\arxiv{quant-ph/0606229}},
doi={10.4086/TOC.2007.V003A004}
}

@article{kirby2023exact,
  title={Exact and efficient Lanczos method on a quantum computer},
  author={Kirby, William and Motta, Mario and Mezzacapo, Antonio},
  journal={Quantum},
  volume={7},
  pages={1018},
  year={2023},
  publisher={Verein zur F{\"o}rderung des Open Access Publizierens in den Quantenwissenschaften},
note={\arxiv{2208.00567}}
}

@unpublished{kitaev1995QuantumMeasurement,
  title={Quantum measurements and the Abelian stabilizer problem},
  author={Kitaev, A Yu},
  journal={arXiv preprint quant-ph/9511026},
  year={1995},
note={\arxiv{quant-ph/9511026}}
}

@Unpublished{	  lin2022LectureNotes,
  author	= {Lin, Lin},
  note		= {\arXiv{2201.08309}},
  title		= {Lecture notes on quantum algorithms for scientific
		  computation},
  year		= {2022}
}

@Article{	  lin2022HeisenbergLimited,
  author	= {Lin, Lin and Tong, Yu},
  doi		= {10.1103/PRXQuantum.3.010318},
  issue		= {1},
  journal	= {\prxq},
  month		= {2},
  note		= {\arxiv{2102.11340}},
  numpages	= {21},
  pages		= {010318},
  publisher	= {American Physical Society},
  title		= {Heisenberg-Limited Ground-State Energy Estimation for
		  Early Fault-Tolerant Quantum Computers},
  volume	= {3},
  year		= {2022}
}

@Article{	  low2016HamSimQSignProc,
  author	= {Low, Guang Hao and Chuang, Isaac L.},
  doi		= {10.1103/PhysRevLett.118.010501},
  journal	= {\prl},
  note		= {\arxiv{1606.02685}},
  number	= {1},
  numpages	= {5},
  pages		= {010501},
  title		= {Optimal {H}amiltonian Simulation by Quantum Signal
		  Processing},
  volume	= {118},
  year		= {2017}
}

@Article{	  low2016HamSimQubitization,
  author	= {Low, Guang Hao and Chuang, Isaac L.},
  doi		= {10.22331/q-2019-07-12-163},
  journal	= {\quantum},
  note		= {\arxiv{1610.06546}},
  pages		= {163},
  title		= {Hamiltonian Simulation by Qubitization},
  volume	= {3},
  year		= {2019}
}

@Article{	  markov1889sur,
  author	= {Markov, A},
  journal	= {Bulletin of the Academy of Sciences of St.
Petersburg},
  pages		= {1-24},
  title		= {Sur une question posée par Mendeleieff,},
  volume	= {62},
  year		= {1889}
}

@Article{	  martyn2021GrandUnificationQAlgs,
  author	= {Martyn, John M. and Rossi, Zane M. and Tan, Andrew K. and
		  Chuang, Isaac L.},
  doi		= {10.1103/PRXQuantum.2.040203},
  journal	= {\prx},
  note		= {\arXiv{2105.02859}},
  number	= {4},
  numpages	= {40},
  pages		= {040203},
  title		= {Grand Unification of Quantum Algorithms},
  volume	= {2},
  year		= {2021}
}

@inproceedings{montanaro2024quantumclassicalquery,
  title={Quantum and classical query complexities of functions of matrices},
  author={Montanaro, Ashley and Shao, Changpeng},
  booktitle	= {\stoc{56th}},
  pages={573--584},
  year={2024},
doi={10.1145/3618260.3649665},
note={\arxiv{2311.06999}}
}

@Unpublished{	  nakaji2023qswift,
  author	= {Nakaji, Kouhei and Bagherimehrab, Mohsen and Aspuru-Guzik,
		  Alan},
  note		= {\arxiv{2302.14811}},
  title		= {qSWIFT: High-order randomized compiler for Hamiltonian
		  simulation},
  year		= {2023}
}

@Book{		  nielsen2002QCQI,
  author	= {Nielsen, Michael A. and Chuang, Isaac L.},
  doi		= {10.1017/CBO9780511976667},
  publisher	= {Cambridge University Press},
  title		= {Quantum computation and quantum information},
  year		= {2000}
}

@unpublished{oleary2024partitioned,
  title={Partitioned Quantum Subspace Expansion},
  author={O'Leary, Tom and Anderson, Lewis W and Jaksch, Dieter and Kiffner, Martin},
  journal={arXiv preprint arXiv:2403.08868},
  year={2024},
note={\arxiv{2403.08868}}
}

@Article{	  sachdeva2014FasterAlgsViaApxTheory,
  author	= {Sachdeva, Sushant and Vishnoi, Nisheeth K.},
  doi		= {10.1561/0400000065},
  issn		= {1551-305X},
  journal	= {Found. Trends Theor. Comput. Sci.},
  number	= {2},
  numpages	= {86},
  pages		= {125--210},
  publisher	= {Now Publishers Inc.},
  title		= {Faster Algorithms via Approximation Theory},
  url		= {http://theory.epfl.ch/vishnoi/Publications_files/approx-survey.pdf},
  volume	= {9},
  year		= {2014}
}

@Article{	  Shao2022FQILSS,
  address	= {New York, NY, USA},
  author	= {Shao, Changpeng and Montanaro, Ashley},
  doi		= {10.1145/3520141},
  journal	= {\acmtqc},
  note		= {\arxiv{2103.10309}},
  number	= {4},
  title		= {Faster Quantum-Inspired Algorithms for Solving Linear
		  Systems},
  volume	= {3},
  year		= {2022}
}

@unpublished{silva2022fourier,
  title={Fourier-based quantum signal processing},
  author={Silva, Thais de Lima and Borges, Lucas and Aolita, Leandro},
  note={\arxiv{2206.02826}},
  year={2022}
}

@InProceedings{	  tang2018QuantumInspiredRecommSys,
  author	= {Tang, Ewin},
  booktitle	= {\stoc{51st}},
  doi		= {10.1145/3313276.3316310},
  note		= {\arxiv{1807.04271}},
  numpages	= {12},
  pages		= {217--228},
  title		= {A Quantum-Inspired Classical Algorithm for Recommendation
		  Systems},
  year		= {2019}
}

@Article{	  tang2018QInspiredClassAlgPCA,
  author	= {Tang, Ewin},
  doi		= {10.1103/PhysRevLett.127.060503},
  journal	= {\prl},
  note		= {\arxiv{1811.00414}},
  number	= {6},
  numpages	= {6},
  pages		= {060503},
  publisher	= {American Physical Society},
  title		= {Quantum Principal Component Analysis Only Achieves an
		  Exponential Speedup Because of Its State Preparation
		  Assumptions},
  volume	= {127},
  year		= {2021}
}

@article{tang2022dequantizing,
  title={Dequantizing algorithms to understand quantum advantage in machine learning},
  author={Tang, Ewin},
  journal={Nature Reviews Physics},
  volume={4},
  number={11},
  pages={692--693},
  year={2022},
  publisher={Nature Publishing Group UK London},
doi={https://doi.org/10.1038/s42254-022-00511-w}
}

@Article{tosta2023randomizedSemiQuantum,
  title={Randomized semi-quantum matrix processing},
  author={Tosta, Allan and Silva, Thais de Lima and Camilo, Giancarlo and Aolita, Leandro},
  journal= {npj Quantum Inf},
  volume = {10},
  pages = {93},
  note={\arxiv{2307.11824}},
  doi = {10.1038/s41534-024-00883-0},
  year={2024}
}

@Article{	  wan2021RandPhaseEst,
  author	= {Wan, Kianna and Berta, Mario and Campbell, Earl T.},
  doi		= {10.1103/PhysRevLett.129.030503},
  issue		= {3},
  journal	= {\prl},
  month		= {7},
  note		= {\arxiv{2110.12071}},
  numpages	= {7},
  pages		= {030503},
  publisher	= {American Physical Society},
  title		= {Randomized Quantum Algorithm for Statistical Phase
		  Estimation},
  url		= {https://link.aps.org/doi/10.1103/PhysRevLett.129.030503},
  volume	= {129},
  year		= {2022}
}

@article{wang2023QAlgGSEE,
  title={Quantum algorithm for ground state energy estimation using circuit depth with exponentially improved dependence on precision},
  author={Wang, Guoming and Fran{\c{c}}a, Daniel Stilck and Zhang, Ruizhe and Zhu, Shuchen and Johnson, Peter D},
  journal={Quantum},
  volume={7},
  pages={1167},
  year={2023},
  publisher={Verein zur F{\"o}rderung des Open Access Publizierens in den Quantenwissenschaften},
note={\arxiv{2209.06811}},
doi={https://doi.org/10.22331/q-2023-11-06-1167}
}

@unpublished{wang2023fasterGSEE,
  title={Faster ground state energy estimation on early fault-tolerant quantum computers via rejection sampling},
  author={Wang, Guoming and Fran{\c{c}}a, Daniel Stilck and Rendon, Gumaro and Johnson, Peter D},
  note={\arxiv{2304.09827}},
  year={2023}
}

@article{	  wang2023qubitefflinalg,
  author	= {Wang, Samson and McArdle, Sam and Berta, Mario},
  note		= {\arxiv{2302.01873}},
  title		= {Qubit-efficient randomized quantum algorithms for linear
		  algebra},
  journal={PRX Quantum},
  volume={5},
  number={2},
  pages={020324},
  year={2024},
  publisher={APS},
  doi={https://doi.org/10.1103/PRXQuantum.5.020324}
}

@article{   zhang2023circuitCompleixtyEncodingClassicalData,
   title={Circuit complexity of quantum access models for encoding classical data},
   volume     = {10},
   ISSN      = {2056-6387},
   url   = {http://dx.doi.org/10.1038/s41534-024-00835-8},
   DOI      = {10.1038/s41534-024-00835-8},
   number     =   {1},
   journal      = {npj Quantum Information},
   publisher       = {Springer Science and Business Media LLC},
   author    = {Zhang, Xiao-Ming and Yuan, Xiao},
   year     = {2024},
   note		= {\arxiv{2311.11365}}
}

@article{peres1985reversibleLogic,
  title = {Reversible logic and quantum computers},
  author = {Peres, Asher},
  journal = {Phys. Rev. A},
  volume = {32},
  issue = {6},
  pages = {3266--3276},
  numpages = {0},
  year = {1985},
  publisher = {American Physical Society},
  doi = {10.1103/PhysRevA.32.3266},
  url = {https://link.aps.org/doi/10.1103/PhysRevA.32.3266}
}

@article{Feynman1985,
author = {Richard P. Feynman},
journal = {Optics News},
number = {2},
pages = {11--20},
publisher = {Optica Publishing Group},
title = {Quantum Mechanical Computers},
volume = {11},
year = {1985},
url = {https://www.optica-opn.org/abstract.cfm?URI=on-11-2-11},
doi = {10.1364/ON.11.2.000011},
}

@article{Nagaj_2010,
   title={Fast universal quantum computation with railroad-switch local Hamiltonians},
   volume={51},
   ISSN={1089-7658},
   url={http://dx.doi.org/10.1063/1.3384661},
   DOI={10.1063/1.3384661},
   note = {\arxiv{0908.4219}},
   number={6},
   journal={Journal of Mathematical Physics},
   publisher={AIP Publishing},
   author={Nagaj, Daniel},
   year={2010},
   month=jun }
\addcontentsline{toc}{section}{\protect\numberline{}References}

\appendix

\section{Appendix}

\subsection{Useful lemmas}

We recall the Chebyshev polynomials and the Bessel functions alongside some of their properties.

\begin{definition}[Chebyshev polynomials and the Bessel functions]\label{def:chebyshev_polynomials}

The Chebyshev polynomials of the first kind are obtained from the recurrence relation

\begin{align}\label{eq:tcheby_polynomial}
    T_0(x) &= 1\nonumber\\
    T_1(x) &= x\\
    T_{n+1}(x) &= 2xT_n(x) - T_{n-1}(x).\nonumber
\end{align}

The Chebyshev polynomials of the second kind $U_n$ are obtained following the same recurrence, but considering $U_1(x) = 2x$. They satisfy

\begin{align}
    &|T_n(x)| \leq 1 \text{ for } x \in [-1,1]\\
    &|U_n(x)| \leq n+1 \text{ for } x \in [-1,1] \nonumber\\
    &T_n\left( \cos\left( \frac{\pi j}{n} \right) \right) = (-1)^j \nonumber\\
    &T_n'(x) = nU_{n-1}(x). \nonumber
\end{align}

The Bessel functions of the first kind, denoted as $J_\alpha(x)$ where $\alpha \in \R$ are defined as

\begin{align}
    J_\alpha(x) = \sum_{m=0}^{\infty} \frac{(-1)^m}{m!\, \Gamma(m+\alpha+1)} \left( \frac{x}{2}\right)^{2m+\alpha}\,.
\end{align}

For $\alpha = 0,1,2,\ldots$ it holds that $|J_\alpha(x)| \leq 1$ for all $x \in \R^{\geq 0}$.
    
\end{definition}

Employing them it is possible to approximate the function $e^{itx}$ with a polynomial of low degree. Moreover, the coefficients of these polynomials are also small.

\begin{lemma}[Bound on the coefficients of Chebyshev polynomials]\label{lemma:bound_tcheby_coeff}
    The coefficients of the polynomial $T_n(x)$ are upper bounded by $4^n$.
\end{lemma}

\begin{proof}
    We prove this simple fact by induction. Observe that it holds for $n=0,1$. For the inductive case, let $c_n$ denote the biggest coefficient of the $n$-th Chebyshev polynomial. Then, by Eq. \eqref{eq:tcheby_polynomial} it holds that

    \begin{align}
        c_n \leq 2c_{n-1} + c_{n-2} \leq  2 \times 4^{n-1} + 4^{n-2} \leq 4^{n}\,.
    \end{align}
\end{proof}


Whenever $A$ is hermitian the value $\bra{i} A \ket{j}$ can be expressed as a linear combination of terms of the form $\bra{\psi} A \ket{\psi}$ for different vectors $\ket{\psi}$. Thus, any algorithm that computes $\bra{\psi} A \ket{\psi}$ for arbitrary vectors $\ket{\psi}$ can be employed to solve the proposed problem. We prove this simple algebraic property for completeness.

\begin{lemma}[Decomposition of off-diagonal entries]\label{lemma:non-diagonal-hermitian}
    Let $A \in \C^{N\times N}$ be a Hermitian matrix. Then, for any $k,j \in \until{N}$ the value $\bra{k} A \ket{j}$ can be written as a linear combination of a constant number of terms of the form $\bra{\psi} A \ket{\psi}$.
\end{lemma}

\begin{proof}
    We show that both the real part $\Re(\bra{k} A \ket{j})$ and the imaginary part $\Im(\bra{k} A \ket{j})$ can be written as a linear combination of terms of the form $\bra{\psi} A \ket{\psi}$.

    \begin{align}
        2\Re{\bra{k} A \ket{j}} &= 
        \bra{k} A \ket{j} + \bra{j} A \ket{k}\nonumber\\
        &= (\bra{k} + \bra{j} - \bra{j}) A \ket{j} + \bra{j} A \ket{k} \nonumber\\
        &= (\bra{k} + \bra{j}) A \ket{j} - \bra{j} A \ket{j} + \bra{j} A \ket{k} \\
        &= (\bra{k} + \bra{j}) A (\ket{k} + \ket{j} - \ket{k}) - \bra{j} A \ket{j} + \bra{j} A \ket{k} \nonumber\\
        &=(\bra{k} + \bra{j}) A (\ket{k} + \ket{j}) - \bra{k} A \ket{k} - \bra{j} A \ket{j} \nonumber\,,
    \end{align}
\vspace{-16pt}
    \begin{align}
        2 i \Im{\bra{k} A \ket{j}} &= 
        \bra{k} A \ket{j} - \bra{j} A \ket{k} \nonumber\\
        &= (\bra{k} + i\bra{j} - i\bra{j}) A \ket{j} - \bra{j} A \ket{k} \nonumber\\
        &= i(\bra{k} + i\bra{j}) A (-i \ket{j}) - i\bra{j} A \ket{j} - \bra{j} A \ket{k} \\
        &= i(\bra{k} + i\bra{j}) A (\ket{k} - i \ket{j} - \ket{k}) - i\bra{j} A \ket{j} - \bra{j} A \ket{k} \nonumber\\
        &= i(\bra{k} + i\bra{j}) A (\ket{k} - i \ket{j}) - i\bra{k}A\ket{k} - i\bra{j} A \ket{j} \nonumber\,.
    \end{align}
\end{proof}

We now demonstrate a lemma on the well-known spectral decomposition of the cyclic operators.

\begin{lemma}[Spectral decomposition of cyclic shifts]\label{lemma:cyclic_shifts_decomposition}

    Consider the cyclic shift operator $S = \sum_{\ell=0}^{M-1} \ket{\ell + 1}\bra{\ell}$, where the $+$ operation is understood modulo $M$. Then, the eigenvalues of $S$ are $e^{\frac{i 2\pi k}{M}}$ with corresponding eigenvectors $\ket{\psi_k} = \frac{1}{\sqrt{M}}\sum_{\ell=0}^{M-1} e^{- \frac{i 2 \pi k \ell}{M}}\ket{\ell}$, for $k = 0,\ldots, M-1$. Moreover, $\ket{0} = \frac{1}{\sqrt{M}} \sum_{k=0}^{M-1} \ket{\psi_k}$.

    Consider the cyclic shift with an additional $-1$ phase factor, defined as $S' = \sum_{\ell = 0}^{M-2} \ket{\ell + 1} \bra{\ell} - \ket{0} \bra{M-1}$. Then, the eigenvalues of $S'$ are $e^{\frac{i\pi (2k + 1)}{M}}$ with corresponding eigenvectors $\ket{\psi_k'} =\frac{1}{\sqrt{M}} \sum_{\ell=0}^{M-1} e^{-\frac{i \pi (2k + 1) \ell}{M}} \ket{\ell}$, for $k=0,\ldots,M-1$. Moreover, $\ket{0} = \frac{1}{\sqrt{M}}\sum_{k=0}^{M-1} \ket{\psi_k'}$.

    
\end{lemma}

\begin{proof}
    Observe that $S^M - \identity = 0$. Thus, all eigenvalues of $S$ are of the form $e^{\frac{i 2 \pi k}{M}}$ for $k=0,\ldots,M-1$. By direct computation

    \begin{align}
        S \ket{\psi_k} = \frac{1}{\sqrt{M}}\sum_{\ell = 0}^{M-1} e^{- \frac{i 2 \pi k \ell}{M}} S \ket{\ell} = \frac{1}{\sqrt{M}}\sum_{\ell = 0}^{M-1} e^{- \frac{i 2 \pi k \ell}{M}} \ket{\ell + 1} = e^{\frac{i 2 \pi k}{M}} \frac{1}{\sqrt{M}}\sum_{\ell=0}^{M-1} e^{- \frac{i 2 \pi k (\ell + 1)}{M}} \ket{\ell + 1} = e^{\frac{i 2 \pi k}{M}} \ket{\psi_k}\,.
    \end{align}

    Finally $\frac{1}{\sqrt{M}}\sum_{k = 0}^{M-1} \ket{\psi_k} = \frac{1}{M} \sum_{k = 0}^{M-1} \sum_{\ell =0}^{M-1} e^{- \frac{i 2 \pi k \ell}{M}} \ket{\ell} = \frac{1}{M} \sum_{\ell=0}^{M-1} \left(\sum_{k=0}^{M-1} e^{-\frac{i 2 \pi k \ell}{M}}\right) \ket{\ell} = \ket{0}$.

    Regarding $S'$, note that $(S')^M + \identity = 0$, and thus all its eigenvalues are of the form $e^{\frac{i \pi (2k + 1)}{M}}$ for $k = 0,\ldots, M-1$. The rest follows from the same computations as before.
\end{proof}

\begin{lemma}\label{lemma:projecting_sx}
    Let $W$ be defined in Eq.~\eqref{eq:def_w} and let $\ket{s_{\inputVector{x}}} = \step{0} \ket{\inputVector{x}} \zero{r-n}$ be as in Eq.~\eqref{eq:def_BQP}. Then 
    \begin{align}
       \frac{Q^+\ket{s_{\inputVector{x}}}}{|\alpha_{\inputVector{x},0}|} =\step{0}\ket{\phi_0^+}\quad\text{s.t.} \quad C\ket{\phi_0^+}=\frac{\alpha_{\inputVector{x},0}}{|\alpha_{\inputVector{x},0}|}\ket{0}\ket{\psi_{\inputVector{x},0}}\label{eq:project_plus}\\
       \frac{Q^-\ket{s_{\inputVector{x}}}}{|\alpha_{\inputVector{x},1}|} =\step{0}\ket{\phi_0^-}\quad\text{s.t.} \quad C\ket{\phi_0^-}=\frac{\alpha_{\inputVector{x},1}}{|\alpha_{\inputVector{x},1}|}\ket{1}\ket{\psi_{\inputVector{x},1}},\label{eq:project_minus}
    \end{align}
    where $Q^\pm=\frac{\identity\pm  W^M}{2}$ are the projectors onto the eigenspaces $\mathcal{S}^\pm$ of $W^M$ with eigenvalues $+1$ and $-1$.
\end{lemma}

\begin{proof}
    The normalization factor of the vector $Q^+ \ket{s_{\inputVector{x}}}$ can be computed as
    \begin{align}
        \bra{s_{\inputVector{x}}} Q^+ \ket{s_{\inputVector{x}}} = \frac{1}{2} \bra{0} \bra{\inputVector{x}} \bra{0} \identity + W^M \ket{0} \ket{\inputVector{x}} \ket{0} = \frac{1}{2} (1 + \bra{0}\bra{\inputVector{x}} C^\dagger (Z\otimes\identity^{r-1}) C \ket{\inputVector{x}}\ket{0}) = |\alpha_{\inputVector{x},0}|^2\,,
    \end{align}
    where the last two equalities follow by using Eqs.~\eqref{eq:W_M} and \eqref{eq:def_BQP}, respectively. Analogously, $\bra{s_{\inputVector{x}}} Q^- \ket{s_{\inputVector{x}}}=|\alpha_{\inputVector{x},1}|^2$.
By using Eq~\eqref{eq:W_M}, we obtain
\begin{equation}
    \frac{Q^\pm\ket{s_{\inputVector{x}}}}{|\alpha_{\inputVector{x},{0 \atop 1}}|}=\frac{\identity\pm  W^M}{2}\ket{s_{\inputVector{x}}}
    =\frac{1}{2}\step{0}\left(\ket{\inputVector{x}}\zero{r-n}\pm C^\dagger(Z\otimes\identity^{r-1})C\ket{\inputVector{x}}\zero{r-n}\right):=\step{0}\ket{\phi_0^\pm},
\end{equation}
from which we can directly calculate $C\ket{\phi_0^\pm}$ by using Eq.~\eqref{eq:def_BQP}.
\end{proof}

\begin{lemma}\label{lemma:W_spectrum}
Let $W$ be defined in Eq.~\eqref{eq:def_w}. Then the eigenvalues of $W$ are $e^{-\frac{i 2 \pi  \ell}{M}}$ and $e^{\frac{i\pi (2\ell + 1)}{M}}$, with $\ell=0,\cdots,M-1$.   

Let $\ket{s_{\inputVector{x}}} = \step{0}  \ket{\inputVector{x}} \zero{r-n}$ be the input bitstring to \problemBQPCircuitSimulation{}. Denote $P_\ell^+$ and $P_\ell^-$ the projectors onto the subspace corresponding to eigenvalues $e^{-\frac{i 2 \pi k \ell}{M}}$ and $e^{\frac{i\pi (2k + 1)}{M}}$, respectively. Then
    $\omega_\ell^+:=\bra{s_{\inputVector{x}}}P_\ell^+\ket{s_{\inputVector{x}}}=\frac{|\alpha_{\inputVector{x},0}|^2}{M}$ and    $\omega_\ell^-:=\bra{s_{\inputVector{x}}}P_\ell^-\ket{s_{\inputVector{x}}}=\frac{|\alpha_{\inputVector{x},1}|^2}{M}$, for $\ell=0,\cdots,M-1$.
    
\end{lemma}

\begin{proof}

To prove the first part of the lemma, consider the sequence of states $\ket{\phi_0}, \,\ket{\phi_1}=V_0\ket{\phi_0}, \, \ket{\phi_2}=V_1\ket{\phi_1},\cdots, \ket{\phi_{M-1}}=V_{M-2}\ket{\phi_{M-2}}$ built from a state $\ket{\phi_0}$ of $r$ qubits. Notice that  if $C\ket{\phi_0}=\alpha_0\ket{0}\ket{\psi_0}^{r-1}$ with $|\alpha_0|=1$ then $V_{M-1}\ket{\phi_{M-1}}=C^{\dagger}(Z\otimes\identity^{r-1})C\ket{\phi_0}=\ket{\phi_0}$. Therefore, similar to Lemma~\ref{lemma:cyclic_shifts_decomposition}, one can verify by direct calculation that the state $\ket{\psi^+_k}= \frac{1}{\sqrt{M}}\sum_{\ell=0}^{M-1} e^{- \frac{i 2 \pi k \ell}{M}}\step{\ell}\otimes\ket{\phi_\ell}$ is an eigenstate of $W$ with eigenvalue $e^{i\frac{2\pi k\ell}{M}}$. Similarly, if $C\ket{\phi_0}=\alpha_1\ket{1}\ket{\psi_1}^{r-1}$ with $|\alpha_1|=1$ then $V_{M-1}\ket{\phi_{M-1}}=C^{\dagger}(Z\otimes\identity^{r-1})C\ket{\phi_0}=-\ket{\phi_0}$ and one can verify that $\ket{\psi^-_k}= \frac{1}{\sqrt{M}}\sum_{\ell=0}^{M-1} e^{- \frac{i \pi (2k+1) \ell}{M}}\step{\ell}\otimes\ket{\phi_\ell}$ is an eigenstate of $W$ with eigenvalue $e^{i\frac{\pi  (2k+1) \ell}{M}}$.

To prove the second part, we start by noticing that, analogous to Lemma \ref{lemma:cyclic_shifts_decomposition}, $\frac{1}{\sqrt{M}}\sum_{k=0}^{M-1}\ket{\psi^+_k}=\step{0}\otimes\ket{\phi_0}$ and, therefore the overlap $\bra{step_0}\otimes\bra{\phi_0}P_k^+\step{0}\otimes\ket{\phi_0}=\frac{1}{M}$. From this fact and Eq.~\eqref{eq:project_plus} in Lemma \ref{lemma:projecting_sx}, we get the desired overlap $\omega_\ell^+$, since $\ket{s_{\inputVector{x}}} = Q^+ \ket{s_{\inputVector{x}}} \;+ \; Q^- \ket{s_{\inputVector{x}}}$ and $\ket{\phi_0^+}$ can be used to build a sequence $\ket{\phi_0^+},\cdots, \ket{\phi_{M-1}^+}=V_{M-2}\ket{\phi_{M-2}^+}$. The overlap $\omega_\ell^-$ is obtained analogously.

\end{proof}

\begin{lemma}\label{lemma:size_generalized_generated}
    Let $\mathcal{G} = \{P_\ell\}_{1\leq \ell \leq L}$ be a set of $L$ $2^n \times 2^n$ generalized Pauli matrices, and denote by $\langle \mathcal{G} \rangle$ the Pauli sub-group generated by $\mathcal{G}$. Then, $|\langle \mathcal{G} \rangle| \leq 2^{L+1}$. 
\end{lemma}

\begin{proof}
    Since any pair of generalized Pauli matrices either commutes or anticommutes (and $P_\ell^2 = \identity$ for any $P_{\ell}$) it holds that 

    \begin{align}
        \langle \mathcal{G} \rangle = \{(-1)^{k_0} P_1^{k_1} \ldots P_L^{k_L} : k_i \in \{0,1\}, 0 \leq i \leq L\}.
    \end{align}
    Therefore, $|\langle \mathcal{G} \rangle| \leq 2^{L+1}$.
\end{proof}






\begin{lemma}[Pauli decomposition of basis elements]\label{lemma:pw-computational-basis-element}
    Any computational basis matrix element $\ket{i} \bra{j} \in \C^{N\times N}$ has the form

    \begin{align}
        \ket{i} \bra{j} = \sum_{\ell = 1}^{2^{\ceil{\log N}}} a_\ell^{(ij)} P_\ell\,,
    \end{align}
    where $|a_\ell^{ij}| = 2^{- \ceil{\log N}}$. Thus, $\ket{i} \bra{j}$ has Pauli norm 1.
\end{lemma}

\begin{proof}
    Over $\C^{N\times N}$ the basis matrix elements can be expressed in terms of single qubit Pauli matrices:

    \begin{align}
        &\ket{0} \bra{0} = \frac{1}{2} (\identity + \pauliZ)\\
        &\ket{1} \bra{1} = \frac{1}{2} (\identity - \pauliZ) \nonumber\\
        &\ket{0} \bra{1} = \frac{1}{2} (\pauliX + i \pauliY) \nonumber\\
        &\ket{1} \bra{0} = \frac{1}{2} (\pauliX - i \pauliY)\,.  \nonumber    
    \end{align}

    Then, $\ket{i} \bra{j}$ can be written as a tensor product of $\ceil{\log N}$ of these elements, which gives a Pauli decomposition of $2^{\ceil{\log N}}$ terms with coefficients of magnitude $2^{-\ceil{\log N}}$. Furthermore, each index can be computed classically in $\OC(\log N)$ given $i,j$.
\end{proof}

\begin{lemma}[Pauli decomposition of universal gates]\label{lemma:decompositionOfUniversalGates}
    It holds that
    \begin{align}
        &H = \frac{\identity + \pauliX}{\sqrt{2}}\,,\\
        &T = \frac{3}{4} \identity \identity \identity + \frac{1}{4} \pauliZ \identity \identity + \frac{1}{4} \identity \pauliZ \identity - \frac{1}{4} \pauliZ \pauliZ \identity + \frac{1}{4} \identity \identity \pauliX -\frac{1}{4} \pauliZ \identity \pauliX - \frac{1}{4} \identity \pauliZ \pauliX + \frac{1}{4} \pauliZ \pauliZ \pauliX\,, \nonumber
    \end{align}
    where we denote $H$ as the Hadamard gate,  $T$ as the Toffoli gate, and we use $P_1 P_2 P_3$ as shorthand for $P_1 \otimes P_2 \otimes P_3$.
\end{lemma}

\begin{lemma}[Pauli norm of unitary]\label{lemma:pw-unitary}
    The Pauli norm of any $n$-qubit unitary satisfies $\lambda \leq 4^{n}$.
   \end{lemma}
\begin{proof}
    We prove this by showing that the magnitude of each coefficient in the Pauli decomposition cannot be larger than 1. Denote an arbitrary unitary as $U=\sum_{\ell=1}^{4^{n}} a_{\ell} P_{\ell}$. First observe that 
    \begin{align}
        \Tr[U P_{\ell}] = \Tr[a_{\ell} P_{\ell}^2] = 2^n a_{\ell}\,,
    \end{align}
    and moreover that 
    \begin{align}
        |\Tr[U P_{\ell}]| \leq {\|U\|_2 \|P_{\ell}\|_2} = 2^n\,,
    \end{align}
    where we have used the Cauchy-Schwarz inequality, followed by the fact that $U$ and $P_{\ell}$ are unitary. Together these two equations imply that $|a_\ell| \leq 1$ for all $\ell \in [4^{n}]$.
\end{proof}

\begin{lemma}[Pauli norm is multiplicative]\label{lemma:pw-multiplicative}
    If $A$ and $B$ are Hermitian matrices with Pauli norm $\lambdaMatrix{A}$ and $\lambda_B$, then $A \otimes B$ satisfies

    \begin{align}
        \lambda_{A \otimes B} = \lambdaMatrix{A} \lambda_B\,.
    \end{align}
    
\end{lemma}

\begin{proof}
    Let $A = \sum_{\ell} a_\ell P_\ell, B = \sum_k b_k P_k$, then

    \begin{align}
        A \otimes B = \sum_{\ell, k} a_\ell b_k P_\ell \otimes P_k\,,
    \end{align}
    and

    \begin{align}
        \lambda_{A \otimes B} = \sum_{k, \ell} |a_\ell b_k| = \sum_{k, \ell} |a_\ell| |b_k| = \lambdaMatrix{A} \lambda_B\,.
    \end{align}
\end{proof}

\begin{lemma}[Pauli decomposition of clock construction]\label{lemma:pw-clock-construction}

    For any $M$ and $0 \leq k \leq M-1$ it holds that the operator $\identity^{\otimes k} \otimes \ket{10} \bra{01}  \otimes \identity^{\otimes M-k-2}$ has a Pauli decomposition of weight 1 and $\OC(1)$ Pauli terms.
    
\end{lemma}

\begin{proof}
    By Lemma~\ref{lemma:pw-computational-basis-element} the operator $\ket{10} \bra{01}$ has Pauli weight 1 and is written down as a sum of $4$ Pauli terms. We can pad these terms with $M-k-2$ identities on the back and $k$ identities upfront to obtain the result.
\end{proof}

\begin{lemma}[Operator-Bernstein inequality; adapted from \cite{gross2011recovering}, Theorem 6]\label{lem:operator-bernstein}
    Let $X_i \in \C^{N \times N}$ be i.i.d.~Hermitian matrix-valued random variables. Take $p,q \in \R$ such that $\|\E[(X_i-\E[X_i])^2]\| \leq p^2$ and $\|X_i-\E[X_i]\|\leq q$. Then, for any $\varepsilon\leq 2mp/q$ we have
    \begin{equation}\label{eq:operator-bernstein}
        \operatorname{Prob}\left[\Big\| \frac{1}{m}\sum_i^m X_i-\E[X_i]\Big\| > \varepsilon\right] \leq 2N \operatorname{exp}\left(- \frac{\varepsilon^2 m}{4p^2} \right)\,.
    \end{equation}
\end{lemma}

Let us now inspect what this implies for importance sampling matrices in the Pauli basis.

\begin{lemma}[Importance sampling in the Pauli basis]\label{lem:pauli-importance-sampling}
    For Hermitian $A = \sum_l a_{\ell} P_{\ell}$ decomposed in the Pauli basis, denote $\lambda_A = \sum_l |a_{\ell}|$. Suppose we sample according to the distribution $\{|a_{\ell}|/\lambda_A\}_l$ and each time upon obtaining index $l$ output random variable $X_l = (a_{\ell}/|a_{\ell}|)\lambda_A P_{\ell}$. This is an unbiased estimator for $A$, and we obtain $\|\sum_{i=1}^m X_i - A \| \leq \varepsilon \leq 1$ with probability at least $(1-\delta)$ for any number of samples
    \begin{equation}
        m \geq \frac{8\lambda_A^2}{\varepsilon^2} \log \left(\frac{2N}{\delta} \right)\,.
    \end{equation}
\end{lemma}

\begin{proof}
    We directly use Lemma \ref{lem:operator-bernstein}. Firstly, we see that $\E[X_i]=A$. We also have the following bound:
    \begin{align}
        \|X_i-A\| \leq \|X_i\|+\|A\| \leq 2\lambda_A\,,
    \end{align}
    where we have used the triangle inequality and the fact that $\|A\|\leq \lambda_A$. Additionally, we have 
    \begin{align}
        \|\E[(X_i-A)^2]\| = \|\E[(X_i)^2]-A^2]\| \leq \|\E[(X_i)^2]\|+ \|A^2\| \leq 2\lambda_A^2\,,
    \end{align}
    where we have used the submultiplicativity of the operator norm and the fact that $X_i^2 = \lambda_A^2 \identity$. Thus, based on these two bounds we can take $p=\sqrt{2}\lambda_A$ and $q = 2\lambda_A$. Using these values for Eq.~\eqref{eq:operator-bernstein} we obtain the stated result.
\end{proof}

\begin{lemma}[Lipschitz constant of bounded polynomial (Markov, 1889 \cite{markov1889sur})]\label{lemma:polynomial-lipschitz}
Let $p_d$ be a polynomial of degree $d$ and let $c = \max_x |p_d(x)|_{[a,b]}$ over some interval $[a,b]$. Then, the derivative of $p_d(x)$ (denote as $p_d^{\prime}(x)$) satisfies
\begin{equation}
    \big| p_d^{\prime}(x)\big|_{[a,b]} \leq \frac{2 c \cdot d^2}{b-a} \,.   
\end{equation}
\end{lemma}

\subsection{Additional results and proofs}\label{appdx:additional-results}

We start by providing a proof of Lemma \ref{lemma:randomized_classical_lm}, which gives a classical randomized algorithm for polynomials for the local measurement problem (Problem \ref{def:general_problem-lm}).

\begin{proof}[Proof of Lemma \ref{lemma:randomized_classical_lm}]
Let us first deal with the algorithm for Pauli access. We can explicitly write
\begin{align}\label{eq:randomized-classical-lm-full}
    \bra{i}f(A) \pi f(A) \ket{i} = W \sum_{r,r'} \sum_{\ell_1...\ell_r} \sum_{\ell'_1...\ell'_{r'}}  \frac{\alpha_r \alpha_{r'} (a_{\ell_1}...a_{\ell_{r}})(a_{\ell'_1}...a_{\ell'_{r'}})}{2W} \bra{i} P_{\ell_1}\cdots P_{\ell_r}\, (\identity^{\otimes n} + Z \otimes \identity^{\otimes (n-1)})\, P'_{\ell'_1}\cdots P_{\ell'_{r'}}' \ket{i}\,,
\end{align}
where we denote $W = \sum_{r,r'} \sum_{\ell_1...\ell_k} \sum_{\ell'_1...\ell'_{k'}} \big|\alpha_r \alpha_{r'} (a_{\ell_1}...a_{\ell_{r}})(a_{\ell'_1}...a_{\ell'_{r'}})\big|$, which can be bounded as 
\begin{equation}
    W \leq \left( \sum_{r} \sum_{\ell_1...\ell_r} \big|\alpha_r (a_{\ell_1}...a_{\ell_{r}})\big| \right)^2 \leq \left( \sum_{r} \big|\alpha_r \lambda_A^r \big| \right)^2 = \left\| f(\lambda_A x) \right\|^2_{\ell_1}
\end{equation}
We can interpret Eq.~\eqref{eq:randomized-classical-lm-full} as $W$ multiplied by a probabilistic sum over (the diagonal entry of) Pauli strings with a phase factor, each appearing with probability $\{ \big|\alpha_r \alpha_{r'} (a_{\ell_1}...a_{\ell_{r}})(a_{\ell'_1}...a_{\ell'_{r'}})\big| / 2W\}$. Thus, we can sample from this probability distribution. By Hoeffding's inequality, the sample complexity required to attain precision $\varepsilon$ with probability at least $(1-\delta)$ is
\begin{equation}
    C_{samp} = \OC\left( \frac{W^2}{\varepsilon^2} \log \left( \frac{1}{\delta} \right)\right) = \OC\left( \frac{\norm{f(\lambdaMatrix{A}\, x)}{\ell_1}^4}{\varepsilon^2} \log \left( \frac{1}{\delta} \right)\right)\,.
\end{equation}
For each sample, we must evaluate the diagonal entry of a product of up to $m+1$ Pauli strings. This costs $\OC(m\log N)$ time complexity per sample.

For the sparse problem, a similar sampling procedure is used following the technique of \cite{apers2022SimpleBetti} and \cite{montanaro2024quantumclassicalquery}, with the simple observation that $\pi$ is a 1-sparse matrix of $1$s on the diagonal, and so is trivially integrated into a path integral Monte Carlo algorithm. 


\end{proof}

Next, we show a statement of hardness for the local measurement problem, when it is \text{normalized} (i.e., built  from normalized quantum states). Let us start by formally defining the normalized problem.

\medskip

\noindent\problem{$\textsc{Normalized-LM-Monomial}_{\|A\|}^{\accessModel{}}$}{
An $N \times N$ Hermitian matrix $A$ with $\|A\| \leq 1$ and accessible through \accessModel{}, a positive real number $m$,  a precision $\varepsilon$ and a threshold $g$, such that $m, 1 / \varepsilon, g = \OC(\polylog{N})$.
}{
$b$ is an upper bound on $\opNorm{A}$, $t$ and $\frac{1}{\varepsilon}$ are $\polylog{N}$.
}{
Denote $\pi = \ket{0}\bra{0} \otimes \identity_{N/2}$ and $r = \bra{0} A^m \pi A^m\ket{0}/\|A^m\ket{0}\|^2$. Then, answer \textsc{YES} if $r \geq g + \varepsilon$ and \textsc{NO} if $r \leq g - \varepsilon$.
}

\begin{proposition}\label{prop:matrixPowerSparseAccessBQPcomplete-lm-normalized}
    $\textsc{Normalized-LM-Monomial}_{\|A\|}^{\sparseAccess}$ and $\textsc{Normalized-LM-Monomial}_{\|A\|}^{\PauliAccess}$ are \BQPhard{}, even for constant precision $1/\varepsilon = \Omega(1)$.
\end{proposition}

\begin{proof}
    We consider the same matrix $A$ as in Eq.~\eqref{eq:monomial-walk-operator} for the proof of the unnormalized problem. Now, for $\ket{\phi_m} = A^M\ket{0}/\|A^M\ket{0}\|$  explicit evaluation gives 
    \begin{equation}
        \bra{\phi_m}\mathsf{\pi}\ket{\phi_m} = \frac{\sum_{k+1\leq \ell \leq 2k+2} p^2_t(\ell)}{\sum_{\ell} p^2_t(\ell)} |\alpha_{\vec{0},1}|^2\,.
    \end{equation}
    For $p_{\infty}=u$ the stationary ratio in the above expression is now $\frac{k+2}{M} \geq \frac{1}{3}$ (denote this value as $a$). We now show that for large enough $t$ the ratio is still a constant, and thus any problem in BQP can be simulated by solving the normalized local measurement problem. For an arbitrary distribution $\inputVector{p}_m$ satisfying $\left\| \inputVector{p}_m - \inputVector{u} \right\|_1 = \varepsilon$ the numerator of the ratio is minimized when $\inputVector{p}_m$ takes uniform value $\frac{1}{M}-\frac{a\varepsilon}{2}$ across all $k+1\leq \ell \leq 2k+2$. The denominator is maximized for the peaked distribution where $\inputVector{p}_m(k+1)=\frac{1}{M}-\frac{\varepsilon}{2}$, $\inputVector{p}_m(k)=\frac{1}{M}+\frac{\varepsilon}{2}$, and $\inputVector{p}_m(\ell) = \frac{1}{M}$ otherwise. Thus we have 
    \begin{align}
        \frac{\sum_{k+1\leq \ell \leq 2k+2} p^2_t(\ell)}{\sum_{\ell} p^2_t(\ell)} &\geq \frac{\sum_{k+1\leq \ell \leq 2k+2} \left(\frac{1}{M}- \frac{a\varepsilon}{2}  \right)^2}{\left(\frac{1}{M}-\frac{\varepsilon}{2} \right)^2 + \left(\frac{1}{M}+\frac{\varepsilon}{2} \right)^2 + (M-2)\frac{1}{M^2}} \\
        &\geq \frac{1-\varepsilon M/2}{3 + \varepsilon^2 M/4} \\
        &\geq \frac{2}{13}\,,
    \end{align}
    where the last inequality is true for any $\varepsilon\leq \frac{1}{M}$ (nothing that $M\geq 1$). Inpsecting Eq.~\eqref{eq:rapid-mixing}, it is thus sufficient to take $t= \OC(M^2\log M) = \OC(\polylog{N})$ for $M= \OC(\polylog{N})$. This ensures that $\bra{\phi_m}\mathsf{\pi}\ket{\phi_m} \geq \frac{2}{13} |\alpha_{\vec{0},1}|^2$ which can be determined by solving the monomial problem to error $\OC(1)$. As shown in the proofs of Theorem \ref{teo:matrixPowerSparseAccessBQPcomplete} and Proposition \ref{prop:matrixPowerPauliAccessBQPcomplete}, $A$ can be instantiated in both sparse and Pauli access efficiently.
\end{proof}

\begin{proposition}
\label{teo:matrix_inversion_bqp_hard_full_proof}
    The problems \problemMatrixInversion{\norm{A}{1}}{\sparseAccess{}} and \problemMatrixInversion{\lambdaMatrix{A}}{\PauliSparse} are \BQPhard{}, even under the conditions $\norm{A}{1} =\OC(1/\polylog{N})$ and $\lambdaMatrix{A} =\OC(1/\polylog{N})$ (hardness statement from Theorem \ref{teo:inversion_bqp_complete}).
\end{proposition}

\begin{proof}

    We begin with \problemMatrixInversion{\norm{A}{1}}{\sparseAccess{}}. We use once again Eq.~\eqref{eq:value_of_diagonal_entry}, but considering the matrix $\frac{A}{2}$, where $A$ is defined in Eq.~\eqref{eq:A}. It holds that $\norm{\frac{A}{2}}{1}\leq 1$, and the eigenvalues of $\frac{A}{2}$ are the same eigenvalues as those from $A$, but divided by two. 
    

    The largest eigenvalue of $A$ (in magnitude) is $\cos\left( 0 \right) = 1$. If $T = 0 \bmod 2$ the  eigenvalue with smallest magnitude is $\cos \left( \frac{\pi T}{2T+1} \right)$, while if $T = 1 \bmod 2$ it is $\cos \left( 
    \frac{\pi (T+1)}{2T + 1} \right)$. Using the fact that $\cos(x) = - \sin \left( x - \frac{\pi}{2} \right)$ and the standard bounds $x - \frac{x^3}{6} \leq \sin(x) \leq x$ it can be seen that $\kappaMatrix{A} = \OC(\polylog{M})$: in the former case $\kappaMatrix{A} = \cos \left( 
    \frac{\pi T}{2T + 1} \right)^{-1} = \sin \left( \frac{\pi}{2(2T+1)} \right)^{-1} \geq \frac{2(2T+1)}{\pi} = \frac{2M}{\pi}$, whilst in the latter $\kappaMatrix{A} = \big\lvert \cos \left( \frac{\pi (T+1)}{2T + 1} \right) \big\rvert^{-1} = \sin \left( \frac{\pi}{2 (2T   + 1)} \right)^{-1} \geq  \frac{2 M}{\pi}$.

    Following the notation from Thm.~\ref{teo:matrixPowerSparseAccessBQPcomplete} and Eq.~\eqref{eq:jantzig_condition2} (the inverse function is odd), we need to show that 
        \begin{equation}\label{eq:inversion_threshold}
        \frac{2}{M}\left(1+\sum_{\ell=1}^{\frac{M-1}{2}}\frac{2}{\eigen_\ell^+}\right) \geq k
    \end{equation}
    for some constant $k$.

   Assume that $T = 0 \mod 2$, and observe that $\left(\eigen_{\ell+1}^+\right)^{-1} + \left(\eigen_{\frac{M-1}{2} - \ell}^+\right)^{-1} \geq 0$
   for all $\ell = 0, \ldots, \frac{M-1}{4}-1$. Therefore,
    \begin{align}
        \sum_{\ell=1}^{\frac{M-1}{2}} \left(\eigen_{\ell}^+\right)^{-1} &= \sum_{\ell=0}^{\frac{M-1}{4}-1} \left(\eigen_{\ell+1}^+\right)^{-1} + \left(\eigen_{\frac{M-1}{2}-\ell}^+\right)^{-1} \nonumber\\
        &\geq \left(\eigen_{\frac{M-1}{4}}^+\right)^{-1} + \left(\eigen_{\frac{M-1}{4}+1}^+\right)^{-1}\nonumber\\
        &\geq  \cos \left( \frac{\pi \left( \frac{M-1}{2} \right)}{M} \right)^{-1} + \cos \left( \frac{\pi \left(\frac{M-1}{2} + 2\right)}{M} \right)^{-1} \nonumber\\
        &=\sin \left( \frac{\pi}{2M} \right)^{-1} - \sin \left( \frac{3 \pi}{2M} \right)^{-1}\label{eq:sum_of_inverse_sines}\\
        &\geq \frac{2M}{\pi} - \frac{2M}{3\pi - \frac{27 \pi^3}{24M^2}}\,. \nonumber
    \end{align}
    If $M$ is big enough ($M \geq 4$ suffices) then $\frac{27 \pi^3}{24M^2} \leq \pi$, and consequently
    \begin{align*}
        \frac{2M}{\pi} - \frac{2M}{3\pi - \frac{162\pi}{4M^2}} \geq \frac{2M}{\pi} - \frac{2M}{2\pi} = \frac{M}{\pi}
    \end{align*}
    and thus Eq.~\eqref{eq:inversion_threshold} is lower bounded by $\frac{1}{\pi}$. We observe that constraining that $M\geq 7$ does not affect the proof, since the problem \problemBQPCircuitSimulation{} is still \BQPhard{} if its input is conditioned this way.

    Meanwhile, if $T = 1 \bmod 2$ it holds that $\left(\eigen_{\ell}^+\right)^{-1} + \left(\eigen_{\frac{M-1}{2} - \ell}^+\right)^{-1} \leq 0$ for $\ell = 0,\ldots, \frac{T-1}{2}$. Thus
    \begin{align}
        1 + \sum_{\ell=1}^{\frac{M-1}{2}} \frac{2}{\eigen_{\ell}^+} &= 1 + 2\left(\eigen_{\frac{M-1}{2}}^+\right)^{-1} + 2\sum_{\ell=1}^{\frac{T-1}{2}} \left(\left( \eigen_{\ell}^+\right)^{-1} + \left( \eigen_{\frac{M-1}{2} - \ell}^+\right)^{-1}\right)\nonumber\\
        &\leq 2 \left( \left(\eigen_{\frac{T-1}{2}}^+\right)^{-1} + \left( \eigen_{\frac{T+1}{2}}^+ \right) \right) \nonumber \\
        &= 2 \left(\cos \left( \frac{\pi(T-1)}{M} \right)^{-1} + \cos\left( \frac{T+1}{2}\right)^{-1} \right)\nonumber\\
        &= 2 \left( \sin \left( \frac{3\pi}{2M} \right)^{-1} 
        - \sin \left( \frac{\pi}{2M} \right)^{-1} \right). \label{eq:another_sum_of_inverse_sines}
    \end{align}
    Eq.~\eqref{eq:another_sum_of_inverse_sines} is the same as Eq.~\eqref{eq:sum_of_inverse_sines} but with opposite signs. Therefore, we conclude that if $T = 1 \bmod 2$ then the expression in Eq.~\eqref{eq:inversion_threshold} is upper bounded by a constant. Therefore, we can distinguish between acceptance and rejection of the original circuit with a fixed constant precision for both cases.

    Now let us consider the hypothesis where $\|A\|_1 = \OC(1/\polylog{N})$. We pick some $c=\Theta(\polylog{N})$, and consider the rescaled matrix $A' = \frac{A}{c}$. Repeating the proof steps, we see that the possible values of $\left[ A'^{-1} \right]_{j,j}$ are separated by $\Omega(c)$, while $\kappaMatrix{A'}=\kappaMatrix{A}$. Thus, the value $|\alpha_{\inputVector{x},1}|^2$ is still determinable via a constant-error solution to the inverse problem.

    Regarding the hardness \problemMatrixPower{\lambdaMatrix{A}}{\PauliSparse{}} it is possible to employ the previous arguments but considering the matrix $A' = \frac{A}{\lambdaMatrix{A}}$, which satisfies $\lambdaMatrix{A'} \leq 1$ and $A'$ is Pauli sparse (thus satisfying the hypothesis). The proof for the stronger hypothesis $\lambdaMatrix{A} = \OC(1/\polylog{N})$ follows similarly.
\end{proof}

\begin{theorem}\label{teo:matrix_inversion_lm_bqp_hard_full_proof}
    The problems \problemMatrixInversionLM{\norm{A}{1}}{\sparseAccess{}} and \problemMatrixInversionLM{\lambdaMatrix{A}}{\PauliSparse{}} are \BQPhard{}.
\end{theorem}

\begin{proof}
    We employ the construction from~\cite{harrow2009QLinSysSolver} with some extra tweaks and consider as elemental gates the set $\{T,H\}$ containing the Toffoli and Hadamard gates, which correspond to matrices with real entries only.

    {
        Given a BQP circuit of $T$ gates $U_{T-1} \ldots U_0$ consider the following unitary clock construction
    \begin{align}
        U = \sum_{t=0}^{T-1} \clockTransition{}_t \otimes U_t + \clockTransition{}_{t+T} \otimes \identity + 
        \clockTransition{}_{t+2T}\otimes U_{T-1-t}^\dagger
    \end{align}
    which essentially amounts to a clock construction over the circuit $U_0^{\dagger} \ldots U_{T-1}^{\dagger} \identity \ldots \identity U_{T-1} \ldots U_0$ that computes and uncomputes the answer of the \BQP{} circuit, but keeps the computed solution for $T$ steps.\footnote{This trick is known as \textit{idling}.} Now consider the matrix
    \begin{align}
        A = \identity - Ue^{-1/T}
    \end{align}
    which is invertible, since
    \begin{align}
        A\ket{\inputVector{x}} = 0 \iff \ket{\inputVector{x}} = e^{-1/T}U\ket{\inputVector{x}} \implies \norm{\ket{\inputVector{x}}}{2} = e^{-1/T} \norm{\ket{\inputVector{x}}}{2} \implies \ket{\inputVector{x}} = \textbf{0}
    \end{align}

    It holds that $\kappaMatrix{A} = \OC(T)$, and
    \begin{align}
        A^{-1} = \sum_{k\geq 0} U^k e^{-k/T}
    \end{align}

    We can compute $A^{-1} \step{0} \ket{0}$ straightforwardly observing that:
    \begin{align}
        U^k \step{0} \ket{0} = \begin{cases}
            \step{k \bmod 3T} \otimes U_{k \bmod T} \ldots U_0 \ket{0} & 0 \leq k \bmod 3T < T\\
            \step{k \bmod 3T} \otimes U_{T-1} \ldots U_0 \ket{0} & T \leq k \bmod 3T < 2T\\
            \step{k \bmod 3T} \otimes U_{(-1-k) \bmod T} \ldots U_0 \ket{0} & 2T \leq k \bmod 3T < 3T
        \end{cases} 
    \end{align}
    were the expression $U_i\ldots U_0$ should be understood as applying $U_{i-1}$, $U_{i-2}$, $\ldots$, until $U_0$. Thus,
    \begin{align}
        A^{-1} \step{0} \ket{0} &= \sum_{k \geq 0} e^{-k/T} U^k \step{0} \ket{0} \\
        &= \sum_{\substack{k \geq 0 \\ 0 \leq k \bmod 3T < T}} e^{-k/T} \step{k \bmod 3T} \otimes U_{k\bmod T} \ldots U_0 \ket{0} \\
        &+ \sum_{\substack{k \geq 0 \\ T \leq k \bmod 3T < 2T}} e^{-k/T} \step{k \bmod 3T} \otimes U_{T-1} \ldots U_0 \ket{0}\\
        &+ \sum_{\substack{k \geq 0 \\ 2T \leq k \bmod 3T < 3T}} e^{-k/T} \step{k \bmod 3T} \otimes U_{(-1-k) \bmod T} \ldots U_0 \ket{0}
    \end{align}
    We can simplify each summation further as
    \begin{align}
        \sum_{\substack{k \geq 0 \\ 0 \leq k \bmod 3T < T}} &e^{-k/T} \step{k \bmod 3T} \otimes U_{k\bmod T} \ldots U_0 \ket{0} \\
        &= \sum_{k=0}^{T-1} \sum_{m \geq 0} e^{-(k + 3mT) / T} \step{k} \otimes U_k \ldots U_0 \ket{0} \\
        &= \sum_{k=0}^{T-1} e^{-k/T} \step{k} \otimes U_k \ldots U_1 \ket{0} \sum_{m \geq 0} (e^{-3})^{m}\\
        &= \frac{e^{3}}{e^{3} - 1} \sum_{k=0}^{T-1} e^{-k/T} \step{k} \otimes U_k \ldots U_0 \ket{0}
    \end{align}
    and conclude that
    \begin{align}
        A^{-1} \ket{0} &= \frac{e^3}{e^3 - 1} \Big(\sum_{k=0}^{T-1} e^{-k/T} \step{k} \otimes U_k \ldots U_0 \ket{0} \\
        &+ \sum_{k=T}^{2T-1} e^{-k/T} \step{k} \otimes U_{T-1}\ldots U_0 \ket{0} \\
        &+ \sum_{k=2T}^{3T-1} e^{-k/T} \step{k} \otimes U_{3T-1-k} \ldots U_0 \ket{0} \Big)
    \end{align}

    We can assume without loss of generality that the circuit $U_{T-1} \ldots U_0$ uses the first qubit to store the acceptance probability $\alpha_{\boldsymbol{0},1}$, and that it remains in the $\ket{0}$ state until the last gate. Let $\pi = \ket{1} \bra{1} \otimes \identity_{N/2}$. Therefore, $\bra{0} U^k \pi U^k \ket{0} = 0$ for $0 \leq k < T$ and $2T \leq k < 3T$, and
    \begin{align}
        \bra{step_0}\bra{0} A^{-1}\, \pi\, A^{-1} \step{0} \ket{0} &= \left(\frac{e^3}{e^3-1}\right)^2 |\alpha_{\boldsymbol{0},1}|^2 \sum_{k=T}^{2T-1} e^{-2k/T} = \left(\frac{e^3}{e^3-1}\right)^2 |\alpha_{\boldsymbol{0},1}|^2 \left(\frac{e^{-2}(1-e^{-2})}{1-e^{-2/T}}\right)
    \end{align}

    Thus, we can distinguish between $|\alpha_{\boldsymbol{0},1}|^2 \geq \frac{2}{3}$ and $|\alpha_{\boldsymbol{0},1}|^2\leq \frac{1}{3}$ by estimating $ \bra{step_0}\bra{0} A^{-1}\, \pi\, A^{-1} \step{0}\ket{0}$. Moreover, the gap between both cases is
    \begin{align}
        \frac{1}{3} \left(\frac{e^3}{e^3-1}\right)^{2} \frac{e^{-2}(1-e^{-2})}{1-e^{-2/T}} \geq \frac{1}{3} \left(\frac{e^3}{e^3-1}\right)^{2} e^{-2}   
    \end{align}
    and thus, we can distinguish them with constant precision. If we consider the normalized measurement case, we can compute the norm of $\opNorm{A^{-1}\ket{0}}$ as
    \begin{align}
        \opNorm{A^{-1}\ket{0}} = \frac{e^3}{e^3-1} \sqrt{\sum_{k=0}^{3T-1} e^{-2k/T}} = \frac{e^3}{e^3-1} \sqrt{\frac{1- e^{-6}}{1-e^{-2/T}}}
    \end{align}
    and then the measurement gives the result
    \begin{align}
        \frac{\bra{0} A^{-1}\, \pi\, A^{-1} \ket{0}}{\opNorm{A^{-1}\ket{0}}^2} = \frac{e^{-2}(1-e^{-2})}{1-e^{-6}}|\alpha_{\boldsymbol{0},1}|^2 = \frac{e^{-2}}{1-e^{-2} - e^{-4}} |\alpha_{\boldsymbol{0},1}|^2
    \end{align}
    over which we can distinguish between acceptance and rejection of the original \BQP{} circuit with constant precision.
    
    Finally, note that $A$ is not symmetric in general, but if we extend the system to have an extra qubit and we take
    \begin{align}
        A' = \begin{bmatrix}
            0 & A\\
            A^\dagger & 0,
    \end{bmatrix}
    \end{align}
    then
    \begin{align}
        (A')^{-1} = \begin{bmatrix}
            0 & (A^\dagger)^{-1}\\
            A^{-1} & 0
    \end{bmatrix}
    \end{align}
    and $A'$ is symmetric, since $A^\dagger = A^T$ due to our choice of elemental gates (Toffoli and Hadamard). Moreover, $\frac{\bra{0}\bra{0} A'^{-1}\, \pi\, A'^{-1} \ket{0}\ket{0}}{\opNorm{A'^{-1}\ket{0}\ket{0}}^2}=\frac{\bra{0} A^{-1}\, \pi\, A^{-1} \ket{0}}{\opNorm{A^{-1}\ket{0}}^2}$ and the local measurement has to be performed on the second qubit, slightly different of our previous convention for the local measurement problem.

    It is possible to build efficient sparse access for $A'$, and moreover it is Pauli sparse. Finally, if we take $A'' = \frac{A'}{2}$ or rather $A'' = \frac{A'}{\lambdaMatrix{A'}}$ we obtain each of the desired reductions.
    }

\end{proof}

\end{document}